%% file: Driver_QP_vs_VI.tex
\documentclass[11pt,reqno]{amsproc}
\allowdisplaybreaks
\usepackage{fullpage}
\usepackage[semicolon,square,authoryear]{natbib}
\numberwithin{equation}{section}
\usepackage{cite}
\usepackage{enumerate}
\usepackage{enumitem}
\usepackage[framed,numbered]{mcode}
\usepackage{caption}
\usepackage{color}
\usepackage{graphicx}
\usepackage{psfrag,epsfig}
\usepackage{epstopdf}
\usepackage{subfig}
\usepackage{multirow}
\usepackage{listings}
\usepackage{longtable}
\usepackage{booktabs,tabu}
\usepackage[debug=false, colorlinks=true, pdfstartview=FitV, linkcolor=blue, citecolor=blue, urlcolor=blue]{hyperref}
\newtheorem{remark}{Remark}
\newtheorem{theorem}{Theorem}
\newtheorem{corollary}{Corollary}
\usepackage{morefloats}

\newlength{\drop}
\definecolor{amethyst}{rgb}{0.6, 0.4, 0.8}
\definecolor{burgundy}{rgb}{0.5, 0.0, 0.13}

\title{\textbf{Variational inequality approach 
to enforcing the non-negative constraint for advection-diffusion equations}}
\author{\textbf{J.~Chang}, and \textbf{K.~B.~Nakshatrala} \\
{\small Department of Civil and Environmental Engineering, 
  University of Houston. \\
  \textbf{Correspondence to:}~\textsf{knakshatrala@uh.edu}}}

\date{\today}
\begin{document}

\begin{titlepage}
  \drop=0.1\textheight
  \centering
  \vspace*{\baselineskip}
  \rule{\textwidth}{1.6pt}\vspace*{-\baselineskip}\vspace*{2pt}
  \rule{\textwidth}{0.4pt}\\[\baselineskip]
  {\LARGE \textbf{\color{burgundy}
    Variational inequality approach to enforcing the\\[0.4\baselineskip] 
    non-negative constraint for advection-diffusion\\[0.4\baselineskip]equations}}\\[0.4\baselineskip]
    \rule{\textwidth}{0.4pt}\vspace*{-\baselineskip}\vspace{3.2pt}
    \rule{\textwidth}{1.6pt}\\[\baselineskip]
    \scshape
    An e-print of the paper is available 
    on arXiv:~1611.08758. \par
    \vspace*{0.25\baselineskip}
    Authored by \\[0.25\baselineskip]
    {\Large J.~Chang \par}
    {\itshape Graduate Student, University of Houston. \par}
    \vspace*{0.3\baselineskip}
    {\Large K.~B.~Nakshatrala\par}
    {\itshape Department of Civil \& Environmental Engineering \\
    University of Houston, Houston, Texas 77204--4003. \\ 
    \textbf{phone:} +1-713-743-4418, \textbf{e-mail:} knakshatrala@uh.edu \\
    \textbf{website:} http://www.cive.uh.edu/faculty/nakshatrala\par}
    \vspace*{0.3\baselineskip}
    \begin{figure}[h]
    \includegraphics[scale=0.45,clip]{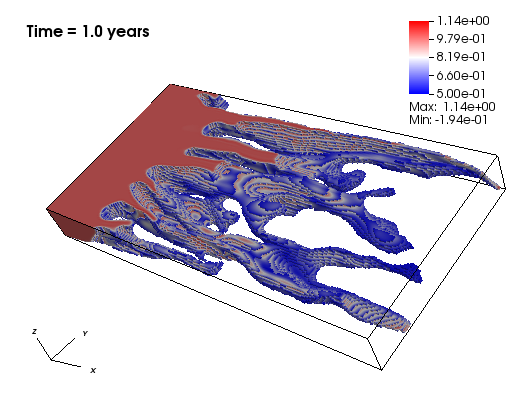}
    \includegraphics[scale=0.45,clip]{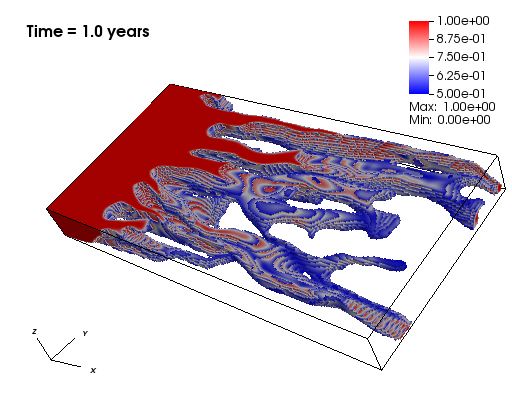}
    \caption*{\emph{These figures depict the concentration profiles of the 
    advection-diffusion equation under the Discontinuous Galerkin formulation
    (left) and the Discontinuous Galerkin formulation with bounded constraints
    enforced through a variational inequality (right). Only
    the regions where concentrations meet the threshold of 0.5 and above 
    are shown.}}
\end{figure}
    \vfill
    {\scshape 2016} \\
    {\small Computational \& Applied Mechanics Laboratory} \par
\end{titlepage}

\begin{abstract} 
Predictive simulations are crucial for the success of many 
subsurface applications, and it is highly desirable to obtain 
accurate non-negative solutions for transport equations in 
these numerical simulations. 
To this end, optimization-based methodologies based on quadratic programming 
(QP) have been shown to be a viable approach to ensuring discrete 
maximum principles and the non-negative constraint for anisotropic 
\emph{diffusion} equations. 
In this paper, we propose a computational framework 
based on the variational inequality (VI) which can also be used to 
enforce important mathematical properties (e.g., maximum principles) 
and physical constraints (e.g., the non-negative constraint). We demonstrate 
that this framework is not only applicable to diffusion equations but 
also to non-symmetric advection-diffusion equations. 
An attractive feature of the proposed framework is that it works with 
with any weak formulation for the advection-diffusion equations, 
including single-field formulations, which are computationally 
attractive. 
A particular emphasis is placed on the parallel and algorithmic performance 
of the VI approach across 
\emph{large-scale} and \emph{heterogeneous} problems. It is also 
shown that QP and VI are equivalent under certain conditions. 
State-of-the-art QP and VI solvers available from 
the PETSc library are used on a variety of steady-state
2D and 3D benchmarks, and a comparative study on the scalability 
between the QP and VI solvers is presented. We then extend the proposed 
framework to transient problems by simulating the miscible 
displacement of fluids in a heterogeneous porous medium and 
illustrate the importance of enforcing maximum 
principles for these types of coupled problems. 
Our numerical experiments indicate that VIs are indeed a viable approach for 
enforcing the maximum principles and the non-negative constraint 
in a large-scale computing environment. Also provided are Firedrake
project files as well as a discussion on the computer implementation
to help facilitate readers in understanding the proposed framework.
\end{abstract}
\maketitle

\input{S0_Model_Notation.tex}

\input{S1_VI_Intro}

\input{S2_VI_Continuous}

\input{S3_VI_Weak}

\input{S4_VI_Discrete}

\input{S5_VI_NR}

\input{S6_VI_Transient}

\input{S7_VI_CR}

\section*{ACKNOWLEDGMENTS}
The authors acknowledge the financial support
from the Department of Energy Office of Science
Graduate Student Research (SCGSR) award. The authors 
also acknowledge the use of the Opuntia Cluster from 
the Center of Advanced Computing and Data Systems (CACDS) 
at the University of Houston to carry out the research 
presented here. The opinions expressed in this paper 
are those of the authors and do not necessarily 
reflect that of the sponsors.

\appendix
\input{A1_Code}
\input{A2_Darcy}

\bibliographystyle{plainnat}
\bibliography{references}

\end{document}

%% file: S0_Model_Notation.tex

\section*{A list of abbreviations} 
\label{Sec:S0_Notation}
\begin{longtable*}{|p{.15\textwidth} | p{.75\textwidth}|} \hline
  \small
  ABC & Arnold-Beltrami-Childress \\
  CG & Conjugate Gradient method \\
  DG & Discontinuous Galerkin \\
  GAL & (Continuous) Galerkin \\
  GMRES & Generalized Minimal Residual method \\
  KSP & Krylov subspace iterative solver\\
  MCP & Mixed complementarity problem \\
  MP & Minimization problem \\
  MPI & Message Passing Interface \\
  PETSc & Portable Extensible Toolkit for Scientific Computation \citep{petsc-user-ref} \\
  QP & Quadratic programming \\
  QP - TRON & Trust region Newton method \\
  SP & Strong problem \\
  SUPG & Streamlined Upwind Petrov-Galerkin \\
  TAO & Toolkit for Advanced Optimization \citep{tao-user-ref}\\
  VI & Variational Inequality \\
  VI - SS & Semi-smooth method \\
  VI - RS & Reduced-space active-set method \\
  WF & Weak formulation \\
  \hline
\end{longtable*}

%% file: S1_VI_Intro.tex
\section{INTRODUCTION AND MOTIVATION}
\label{Sec:S1_VI_Intro}
This paper presents a numerical methodology based on
\emph{variational inequalities} for \emph{anisotropic}
diffusion and advection-diffusion equations that satisfies
discrete maximum principles, meets the non-negative constraint,
and is well-suited for solving large-scale problems
using \emph{parallel computing}. We now provide a motivation
behind the current work, a discussion on prior works,
and an outline of our approach highlighting the
significance of the work presented herein.

The diffusion and advection-diffusion equations are 
important partial differential equations which are
commonly used to model flow and transport of chemical
species in porous media. Some of the applications
include subsurface remediation \citep{uscleaning_EPA_2004,
Harp_SERR_2013,Heikoop_CG_2014} and transport of
radionuclides \citep{Hammond_WRR_2010,Genty_TPM_2011}. 
Since these important problems are not analytically
tractable, one needs to rely on predictive 
numerical simulations. An important aspect in 
a predictive simulation of these equations is 
to satisfy the non-negative constraint of
concentration of chemical species.

Research efforts over the years have successfully created numerical 
models and discretization for these
equations, but they are not without their setbacks. For example, 
non-monotone discretizations like the finite element method may result 
in spurious oscillations with high P\'eclet numbers. Other common issues 
that may occur within highly heterogeneous and anisotropic diffusion-type
equations are violations of the maximum principle and the non-negative constraint 
\citep{Ciarlet_CMAME_1973,Lipnikov_JCP_2007,Liska_CiCP_2008,Genty_TPM_2011}. 
Such numerical setbacks can result in algorithmic failures or sharp 
fronts that may result in erroneous approximations of reactive 
transport. Moreover, several important applications which require
accurate predictive capabilities of transport solvers are often
large-scale and cannot be solved on a single computer. It is important
for numerical algorithms to not only ensure maximum principle but scale 
well with respect to both problem size and computing concurrency. Obtaining 
numerical solutions within a reasonable amount of time is the ultimate 
goal when selecting or designing algorithms that are robust 
and can ensure non-negative concentrations for a 
wide range of subsurface transport applications.

\subsection{Prior works on non-negative formulations}
The prior non-negative formulations can be broadly classified 
into the following five categories:
\begin{enumerate}[leftmargin=*,label=(\alph*)]
\item \emph{Reporting the violations}: In \citep{Payette_IJNME_2012},
several cases of violations of the maximum principle and the non-negative 
constraint have been showcased for different anisotropic 
diffusivity tensors. This paper also demonstrates that $h$- and $p$-refinements
do not eliminate these violations. The adverse effects due to violations
of the non-negative constraint for non-linear ecological models
and chemically reacting flows have been illustrated in \citep{mudunuru2015local}.
Neither of these papers have provided any fix to overcome these violations.
\item \emph{Mesh restrictions}: The first work on  maximum 
principles under the finite element method can be traced back to 
the seminal paper by \citep{Ciarlet_CMAME_1973}. This paper considered
\emph{isotropic} diffusion, and has shown that an acute-angled triangular 
mesh (which is a restriction on the mesh) will satisfy the maximum principle 
under the finite element method. Anisotropic diffusion equations have been 
addressed in \citep{huang2015discrete}, wherein they developed 
an algorithm to generate metric-based meshes to satisfy the maximum principle
for such equations. \citep{mudunuru2016mesh} have addressed various versions
of maximum principles for diffusion and advection-diffusion equations and 
studied the performance of metric-based meshes for these equations. This
paper highlighted the main deficiency of metric-based meshes, which is the
need to alter the mesh for different diffusivity tensors.
A comprehensive list and discussion of other prior works 
related to enforcing mesh restrictions to meet the maximum principle and the 
non-negative constraint can also be found in \citep{mudunuru2016mesh}.
\item \emph{Developing or altering formulations in the continuum setting}: 
Two works that fall under this category are \citep{harari2004stability,
Pal_IJNME_2016}, both of which addressed transient transport problems.
\citep{harari2004stability} utilized a stabilized method that
is available for Helmholtz-type equations to construct a stabilized formulation
for transient \emph{isotropic} diffusion equations to meet the maximum principle.
This approach, as presented in \citep{harari2004stability}, is applicable to 
one-dimensional setting. \citep{Pal_IJNME_2016} meets maximum principles 
for transient transport equations 
by employing two techniques. They rewrote transient transport 
equations, which are parabolic in nature, into modified Maxwell-Cattaneo 
equations, which are hyperbolic in nature, and employed the
\emph{space-time} Discontinuous Galerkin approach.
\item \emph{Non-finite element approaches}: A finite volume-based approach, 
to enforce the non-negative constraint, as proposed in \citep{LePotier_CRM_2005}, 
involves a non-linear iterative procedure to select 
appropriate collocation points for cell concentrations. This technique 
has been refined by several others including \citep{Lipnikov_JCP_2007} 
and \citep{sheng2016new}. Other similar approaches include the 
mimetic finite difference method \citep{da2014mimetic}, which ensures
monotonicity and positivity. Since neither the finite difference
nor finite volume methods are based on weak formulations, these mentioned works 
cannot be trivially extended to the finite element method.
\item \emph{Optimization-based techniques at the discrete level}: 
Several studies over the years
\citep{Nakshatrala_JCP_2009,Nagarajan_IJNMF_2011,
Nakshatrala_JCP_2013,Nakshatrala_CiCP_2016}
have focused on the development of optimization-based 
methodologies that enforce the maximum principle and the non-negative 
constraint for diffusion problems. An optimization-based methodology based on 
the work of the aforementioned studies has been applied to 
enforce maximum principles advection-diffusion equations 
\citep{Mudunuru_JCP_2016}. By reformulating the advection-diffusion
problem as a mixed finite element formulation under the least-squares
formalism, one introduces flux variables into the problem. The discrete
formulation is also symmetric and positive-definite, so one can easily
apply both bounded constraints and equality constraints to ensure
non-negative solutions and local mass conservation respectively.
It should be noted that one may also employ normal equations or the 
least-squares approach to ensure that the minimization problem for 
non-symmetric problems is convex \citep{Demmel_SIAM_1997,Burdakov_JCP_2012,
Pal_IJNME_2016,Chang_JPM_2016}. All of these studies have employed 
\textsf{quadratic programming (QP)} techniques to enforce the maximum principle
on 2D academic problems, but the problems studied are small-scale and 
do not require state-of-the-art \textsf{Krylov subspace (KSP)} 
iterative solvers and preconditioners. 
Moreover, it is extremely difficult to find solvers for least-squares or
penalty-type problems. To this end, we are interested in numerical formulations 
and solvers that are suitable for \emph{large-scale} applications. 
It has been shown recently that parallel optimization-based solvers 
can handle large-scale heterogeneous and anisotropic 
diffusion in \citep{Chang_JOMP_2016}, so we want to 
extend the work presented in that study to advection-diffusion
equations.
\end{enumerate}

\subsection{Our approach and its salient features}
The main contribution of this paper is 
to present a finite element computational framework 
applicable to both diffusion and advection-diffusion 
equations that meets the maximum principle and satisfies the 
non-negative constraint. The 
framework is built by rewriting the \textsf{weak 
formulation (WF)} as a \textsf{variational inequality (VI)}
\citep{Ulbrich_SIAM_2011}. 

The field of VIs grew from a problem posed by Antonio Signorini
\citep{signorini1933sopra,zbMATH03149633}.
This problem was later coined as ``Signorini
problem'' by Gaetano Fichera, who was a
student of Signorini. Fichera posed the problem more
precisely and obtained a variational inequality
corresponding to the problem using which he
established existence and uniqueness of
solutions \citep{fichera1964problemi,fichera1965linear}.
VIs have been employed to study contact
problems \citep{kikuchi1988contact,hlavacek2012solution},
obstacle problems \citep{rodrigues1987obstacle}, elastoplastic
problems \citep{hlavacek2012solution,han2012plasticity} and
other problems arising in mechanics and mathematics
\citep{Kinderlehrer_VI_2000}.
If the bilinear form under the WF 
is symmetric, one can rewrite the WF 
as a QP, which is a special 
case of VIs \citep{Chipot_elliptic_2009}. 
Most of the existing \emph{single-field} WFs 
for advection-diffusion equations do not have
symmetric bilinear forms, and hence one cannot construct
equivalent problems under QP. To the best of the authors' 
knowledge, VIs have not be employed to develop 
numerical formulations to satisfy maximum principles 
and the non-negative constraint for anisotropic 
advection-diffusion equations. 

The framework is particularly suited for large-scale 
problems, which is the case with many practical subsurface 
applications. The proposed framework enjoys the following
salient features: 
\begin{enumerate}[label=(\Roman*)]
\item One can enforce bounded constraints for any 
  transport problems that may be non-symmetric or 
  nonlinear.
\item One can employ any numerical formulation, 
  even a single-field formulation, for solving 
  advection-diffusion equations.
\item One can leverage on existing high performance 
  computing libraries and toolkits (e.g., solvers 
  and preconditioners). 
\item The framework is amenable for parallel computing, 
  which will be illustrated using both strong and weak 
  scaling studies. 
\end{enumerate}

The rest of the paper is organized as follows. In Section
\ref{Sec:S2_VI_Continuous}, we present the boundary value
problem for steady-state diffusion and advection-diffusion
equations. In Section \ref{Sec:S3_VI_Weak}, we present the 
variational inequality (VI) and the various single-field weak 
formulations (WF) in the continuous setting. In Section 
\ref{Sec:S4_VI_Discrete}, we propose the computational 
framework in a discrete setting and discuss in detail the
specific solvers and implementation procedure.
In Section \ref{Sec:S5_VI_NR}, numerical results for the 
steady-state governing equations under the proposed framework 
are shown, and we conduct a thorough 
strong and weak-scaling study to demonstrate the parallel
performance. In Section \ref{Sec:S6_VI_TR}, we provide an 
extension of the proposed framework to transient problems 
and illustrate the performance of miscible displacement in porous media, 
which is a coupled non-linear phenomenon. Conclusions are 
drawn in Section \ref{Sec:S7_VI_CR}. To facilitate the reader 
to be able to reproduce the results given in this paper, sample 
Firedrake project files along with the discussion on the solution 
strategy for large-scale Darcy equations
are provided in Appendices \ref{A1:code} and \ref{A2:darcy}.

On the notational front, we denote all the continuum vectors
by lowercase boldface unitalicized letters, and the vectors
in the discrete setting are denoted by lowercase boldface
italic letters. We denote all the continuum second-order
tensors by boldface uppercase unitalicized letters, and all
the finite element matrices are denoted by uppercase boldface
italicized letters. In this paper, repeated indices do not
imply summation. Other notational conventions are introduced
as needed.

%% file: S2_VI_Continuous.tex
\section{GOVERNING EQUATIONS IN THE CONTINUOUS SETTING}
\label{Sec:S2_VI_Continuous}
Let $\Omega \subset \mathbb{R}^{d}$ be a bounded open
domain, where `$d$' denotes the number of spatial
dimensions. The boundary is denoted by $\partial
\Omega = \overline{\Omega} - \Omega$, where a 
superposed bar denotes the set closure. We denote
the set of all $k$-times continuously differentiable
functions on $\Omega$ by $C^{k}(\Omega)$. We denote
the set of all functions in $C^{0}(\Omega)$ that are
continuous to the boundary by $C^{0}(\overline{\Omega})$.
A spatial point is denoted by $\mathbf{x} \in
\overline{\Omega}$. The gradient and divergence
operators with respect to $\mathbf{x}$ are, respectively,
denoted by $\mathrm{grad}[\cdot]$ and $\mathrm{div}[\cdot]$.
The unit outward normal to boundary is denoted by
$\widehat{\mathbf{n}}(\mathbf{x})$. 

Let $c(\mathbf{x})$ denote the concentration field. The
boundary is divided into two parts:~$\Gamma^{\mathrm{D}}$
and $\Gamma^{\mathrm{N}}$, such that $\Gamma^{\mathrm{D}} \cup
\Gamma^{\mathrm{N}} = \partial \Omega$ and $\Gamma^{\mathrm{D}}
\cap \Gamma^{\mathrm{N}} = \emptyset$. $\Gamma^{\mathrm{D}}$
is that part of the boundary on which Dirichlet boundary
conditions are enforced (i.e., concentration is prescribed).
$\Gamma^{\mathrm{N}}$ is the part of the boundary on which
Neumann boundary conditions are enforced (i.e., flux is
prescribed). When advection is present, the Neumann
boundary is further divided into inflow and outflow regions,
which are defined as follows:
\begin{subequations}
\begin{align}
  \label{Eqn:S2_inflow}
  \Gamma^{\mathrm{N}}_{\mbox{\small inflow}} &:= 
  \left\{\mathbf{x}\in\Gamma^{\mathrm{N}}\;\Bigm|\;
  \mathbf{v}(\mathbf{x})\cdot\widehat{\mathbf{n}}(\mathbf{x})
  < 0\right\}\\
  \label{Eqn:S2_outflow}
  \Gamma^{\mathrm{N}}_{\mbox{\small outflow}} &:= 
  \left\{\mathbf{x}\in\Gamma^{\mathrm{N}}\;\Bigm|\;
  \mathbf{v}(\mathbf{x})\cdot\widehat{\mathbf{n}}(\mathbf{x})
  \ge 0\right\}
\end{align}
\end{subequations}
For uniqueness of the solution under a steady-state response, we
assume that concentration is prescribed on a non-zero part of
the boundary (i.e., $\mathrm{meas}\left(\Gamma^{\mathrm{D}} \right) > 0$). 

\subsection{Strong problems (SP)}
The \textsf{strong problem (SP)} for steady-state diffusion 
reads:~Find $c(\mathbf{x}) \in C^{2}(\Omega) \cap C^{0}
(\overline{\Omega})$ such that we have 
\begin{subequations}
  \label{Eqn:S2_D_GE}
  \begin{alignat}{2}
    -&\mathrm{div}
    [\mathbf{D}(\mathbf{x})\mathrm{grad}[c(\mathbf{x})]] = f(\mathbf{x})
    &&\quad \mathrm{in} \; \Omega \\
    &c(\mathbf{x}) = c^{\mathrm{p}}(\mathbf{x})
    &&\quad \mathrm{on} \; \Gamma^{\mathrm{D}} \\
    -&\widehat{\mathbf{n}}(\mathbf{x}) \cdot \mathbf{D}(\mathbf{x})
    \mathrm{grad}[c(\mathbf{x})] = q^{\mathrm{p}}(\mathbf{x})
    &&\quad \mathrm{on} \; \Gamma^{\mathrm{N}}
  \end{alignat}
\end{subequations}
and the SP for steady-state advection-diffusion
reads:~Find $c(\mathbf{x}) \in C^{2}(\Omega) \cap C^{0}(\overline{\Omega})$ 
such that we have
\begin{subequations}
  \label{Eqn:S2_AD_GE}
  \begin{alignat}{2}
    &\mathbf{v}(\mathbf{x})\cdot\mathrm{grad}[c(\mathbf{x})] - 
    \mathrm{div}[\mathbf{D}(\mathbf{x})\mathrm{grad}[c(\mathbf{x})]] = f(\mathbf{x})
    &&\quad \mathrm{in} \; \Omega \\
    &c(\mathbf{x}) = c^{\mathrm{p}}(\mathbf{x})
    &&\quad \mathrm{on} \; \Gamma^{\mathrm{D}} \\
    &\widehat{\mathbf{n}}(\mathbf{x}) \cdot \left(\mathbf{v}(\mathbf{x})
    c(\mathbf{x}) - \mathbf{D}(\mathbf{x})\mathrm{grad}[c(\mathbf{x})]\right) = 
    q^{\mathrm{p}}(\mathbf{x})
    &&\quad \mathrm{on} \; \Gamma^{\mathrm{N}}_{\mbox{\small inflow}} \\
    -&\widehat{\mathbf{n}}(\mathbf{x}) \cdot \mathbf{D}(\mathbf{x})\mathrm{grad}[c(\mathbf{x})] = 
    q^{\mathrm{p}}(\mathbf{x})
    &&\quad \mathrm{on} \; \Gamma^{\mathrm{N}}_{\mbox{\small outflow}}
  \end{alignat}
\end{subequations}
where $\mathbf{v}(\mathbf{x})$ is the advective velocity,
$f(\mathbf{x})$ is the prescribed volumetric source/sink,
$c^{\mathrm{p}}(\mathbf{x})$ is the prescribed concentration on
the boundary, $q^{\mathrm{p}}(\mathbf{x})$ is the prescribed
flux on the boundary, and $\mathbf{D}(\mathbf{x})$ is the
second-order diffusivity tensor. The diffusivity tensor
is assumed to be bounded and uniformly elliptic. That is,
there exist two constants $0 < \xi_{1} \leq \xi_{2} <
+\infty$ such that
\begin{align}
  \xi_{1} \mathbf{y} \cdot \mathbf{y} \leq
  \mathbf{y} \cdot \mathbf{D}(\mathbf{x}) \mathbf{y} \leq
  \xi_{2} \mathbf{y} \cdot \mathbf{y}
  \quad \forall \mathbf{y} \in \mathbb{R}^{d}
\end{align}
Moreover, the diffusivity tensor is
assumed to be symmetric. That is,
\begin{align}
  \mathbf{D}(\mathbf{x}) = \mathbf{D}^{\mathrm{T}}(\mathbf{x})
  \quad \forall \mathbf{x} \in \Omega 
\end{align}
A solution to SP is commonly referred 
to as a \emph{classical} solution.

\subsection{Maximum principle and the non-negative constraint}
From the theory of partial differential equations, 
it is well-known that a 
classical solution for the above mentioned 
SPs satisfies maximum principles. For completeness, 
we provide below the statement of the classical 
maximum principle of second-order elliptic 
partial differential equations with Dirichlet 
boundary conditions on the entire boundary.
\begin{theorem}{(Classical maximum principle)}
  If $\Gamma^{\mathrm{D}} = \partial \Omega$, 
  $c(\mathbf{x}) \in C^{2}(\Omega) \cap 
  C^{0}(\overline{\Omega})$ and $f(\mathbf{x}) 
  \leq 0$, then
  \begin{align}
    \mathop{\mathrm{max}}_{\mathbf{x} \in \Omega} \; c(\mathbf{x}) \leq 
    \mathop{\mathrm{max}}_{\mathbf{x} \in \partial \Omega} \; 
    c^{\mathrm{p}}(\mathbf{x}) 
  \end{align}
\end{theorem}
\begin{proof}
  A proof can be found in \citep{Gilbarg}.
\end{proof}
A generalization of the classical maximum principle that is 
relevant to this paper is provided
in \citep{Mudunuru_JCP_2016}. Specifically, they have extended 
the classical maximum principle on four fronts: the regularity 
of the solution is relaxed to $C^{1}(\Omega) \cap C^{0}(\overline{\Omega})$, 
the regularity of the volumetric source $f(\mathbf{x})$ 
is relaxed to the space of square integrable functions, the 
boundary can have both Dirichlet and Neumann boundary 
conditions, and the Neumann boundary conditions are 
further divided into inflow and outflow (i.e., similar 
to equations \eqref{Eqn:S2_inflow}--\eqref{Eqn:S2_outflow}).
For the sake of brevity, we defer all interested readers 
to the suggested reference.
Another property that is relevant to this paper 
is non-negative solutions, which can be shown 
to be a special case of maximum principles. In 
particular, the above maximum principle implies 
the following result:
\begin{corollary}{(Non-negative solutions)}
  If $\Gamma^{\mathrm{D}} = \partial \Omega$, 
  $c(\mathbf{x}) \in C^{2}(\Omega) \cap 
  C^{0}(\overline{\Omega})$, $c^{\mathrm{p}}(\mathbf{x}) 
  \geq 0$, and $f(\mathbf{x}) \geq 0$, then
  \begin{align}
    0 \leq c(\mathbf{x}) \quad \forall \mathbf{x} \in \overline{\Omega}
  \end{align}
\end{corollary}

\emph{The central aim of this paper is to obtain numerical 
solutions 
to the above governing equations (i.e., equations 
\eqref{Eqn:S2_D_GE} and \eqref{Eqn:S2_AD_GE}) that 
respect maximum principles and the non-negative 
constraint.} 

The main task will then be to find an appropriate 
setting for numerical solutions.
The finite difference method directly discretizes the SP. 
However, under the finite element method, the 
SP is rewritten as a WF, which is equivalent 
to the SP under some regularity assumptions. A solution to 
a WF is referred to as a weak solution. As 
mentioned in Section \ref{Sec:S1_VI_Intro} 
and will be shown using several examples later
in this paper, a WF does not guarantee 
non-negative solutions in the discrete setting. 
To overcome this deficiency, some non-negative 
formulations, especially for diffusion-type
equations, have rewritten the WF as an equivalent 
\textsf{minimization problem (MP)} and augmented 
with bound constraints. However, it 
needs to be emphasized that such conversion is 
not always possible, which is the case with the 
typical WF for advection-diffusion 
equations, as these formulations have non-symmetric 
bilinear forms.  
In order to handle non-self-adjoint differential
operators (e.g., advection-diffusion equation)
and WFs with non-symmetric bilinear forms, we 
rewrite a given WF as a VI. In order to satisfy 
maximum principles and the non-negative constraint, 
we restrict the feasible solution space of the VI 
formulation using bound constraints.
It needs to be mentioned that one can pose the 
VI as an equivalent MP only if the bilinear 
form is symmetric.
Figure \ref{Fig:S2_QP_vs_VI_description} illustrates
the various ways of rewriting the governing equations,
and the conditions under which one form is equivalent 
to the other.
We now present various WFs for diffusion and 
advection-diffusion equations, which will form 
the basis for our proposed VI-based formulations. 

\begin{figure}[t]
  \centering
  \subfloat{\includegraphics[scale=1.0]{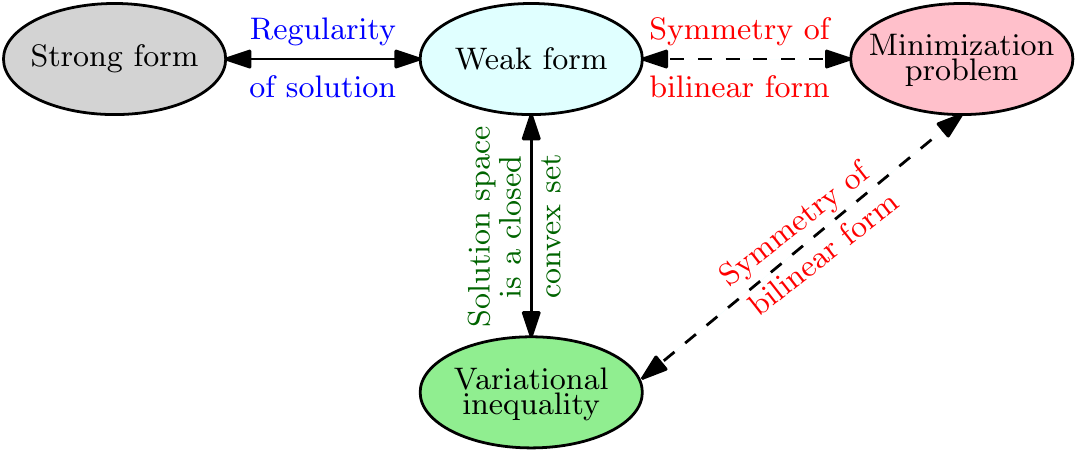}}
  \caption{Relationships between the strong problem (SP), 
    a weak formulation (WF), a variational inequality (VI), and
    the minimization problem (MP).
    \label{Fig:S2_QP_vs_VI_description}}
\end{figure}

%% file: S3_VI_Weak.tex
\section{VARIATIONAL INEQUALITIES AND WEAK FORMULATIONS}
\label{Sec:S3_VI_Weak}
The non-negative constraint and maximum principles restrict 
the feasible solution space to a closed convex set. A variational 
inequality (VI) is basically a variational problem on a convex 
set, which need not be a vector space. To this end, let 
$\mathcal{C}$ denote the solution space for the concentration 
field, and  $\mathcal{K}$ be a closed convex subset of 
$\mathcal{C}$. The subset $\mathcal{K}$ is defined by 
the underlying maximum principles and the non-negative 
constraint. The formulation based on VIs 
corresponding to the mention SPs can be compactly written 
as:~Find $c(\mathbf{x}) \in \mathcal{K}$ such that we have
\begin{align}
  \mathcal{B}(w - c;c) \geq \mathcal{L}(w - c) \quad
  \forall w(\mathbf{x}) \in \mathcal{K}
\end{align}
where $\mathcal{B}(\cdot;\cdot)$ is a bilinear form 
and $\mathcal{L}(\cdot)$ is a linear functional, 
whose specific choices are provided by the associated 
weak formulation.  
A WF can be abstractly written as: Find $c(\mathbf{x}) \in 
\mathcal{C}$ such that we have
\begin{align}
  \label{Eqn:VI_Weak_formulation}
  \mathcal{B}(w;c) = \mathcal{L}(w) \quad
  \forall w(\mathbf{x}) \in \mathcal{W}
\end{align}
where $\mathcal{C}$ and $\mathcal{W}$ are appropriate 
function spaces for a given WF. 
Our intention is to illustrate the applicability of the 
proposed VI framework to a variety of WFs. 
To this end, we employ the \textsf{continuous Galerkin (GAL)}, 
\textsf{Streamlined Upwind Petrov-Galerkin (SUPG)}, and 
\textsf{Discontinuous Galerkin (DG)} formulations, which 
are documented below. 
For convenience, the
standard $L_2$ inner-product over $K$ is denoted as follows:
\begin{align}
  (a;b)_K = \int_K a(\mathbf{x})\cdot b(\mathbf{x})\;\mathrm{d}K
\end{align}

\subsection{Continuous Galerkin}
The relevant function spaces are:
\begin{align}
  \mathcal{C} &:= \left\{c(\mathbf{x}) \in H^{1}(\Omega)
  \; \Bigm\vert \; c(\mathbf{x}) = c^{\mathrm{p}}(\mathbf{x})
  \; \mathrm{on} \; \Gamma^{\mathrm{D}} \right\} \\
  \mathcal{W} &:= \left\{w(\mathbf{x}) \in H^{1}(\Omega)
  \; \Bigm\vert \; w(\mathbf{x}) = 0 
  \; \mathrm{on} \; \Gamma^{\mathrm{D}} \right\} 
\end{align}
where $H^{1}(\Omega)$ is a Sobolev space \citep{Brezis_PDE_2010}.
We assume that $f(\mathbf{x}) \in H^{-1}(\Omega)$, which
is a dual space corresponding to $H^{1}(\Omega)$. We employ the GAL 
formulation for the diffusion problem, for which the bilinear 
form and linear functional are:
\label{Eqn:S2_D_weak}
\begin{align}
  &\mathcal{B}_{\mathrm{GAL}}(w;c) := \Big(\mathrm{grad}[w(\mathbf{x})];\;
  \mathbf{D}(\mathbf{x}) \mathrm{grad}[c(\mathbf{x})]\Big)_\Omega\\
  &\mathcal{L}_{\mathrm{GAL}}(w) := \Big(w(\mathbf{x});\;f(\mathbf{x})
  \Big)_\Omega - \Big(w(\mathbf{x});\;q^{\mathrm{p}}(\mathbf{x})\Big)_{\Gamma^{\mathrm{N}}}
\end{align}
For the advection-diffusion problem, spurious oscillations may
arise under the GAL formulation for high P\'eclet numbers. Herein, 
we employ the SUPG formulation \citep{Brooks_CMAME_1982}, and the corresponding
bilinear form and linear functional are:
\begin{align}
\label{Eqn:S2_AD_weak}
  &\mathcal{B}_{\mathrm{SUPG}}(w;c) := \mathcal{B}_{\mathrm{RES}}(w;c) + 
  \Big(w(\mathbf{x});\;\mathbf{v}(\mathbf{x})\cdot\mathrm{grad}[c(\mathbf{x})]
  \Big)_\Omega + \Big(\mathrm{grad}[w(\mathbf{x})];\;
  \mathbf{D}(\mathbf{x}) \mathrm{grad}[c(\mathbf{x})]\Big)_\Omega\\
  &\mathcal{L}_{\mathrm{SUPG}}(w) := \mathcal{L}_{\mathrm{RES}}(w) + \Big(
  w(\mathbf{x});\;f(\mathbf{x})\Big)_\Omega - \Big(
  w(\mathbf{x});\;q^{\mathrm{p}}(\mathbf{x})\Big)_{\Gamma^{\mathrm{N}}}
\end{align}
where the residual terms
\begin{align}
  \label{Eqn:S2_AD_bilinear_stab}
  &\mathcal{B}_{\mathrm{RES}}(w;c) := 
  \left(\frac{h}{2\|\mathbf{v}(\mathrm{x})\|}\mathbf{v}(\mathrm{x})
  \cdot\mathrm{grad}[w(\mathbf{x})];\;\mathbf{v}(\mathrm{x})\cdot
  \mathrm{grad}[c(\mathbf{x})] - \mathrm{div}\left[\mathbf{D}(\mathbf{x})
  \mathrm{grad}[c(\mathbf{x})]\right] \right)_{\Omega}\\
  \label{Eqn:S2_AD_linear_stab}
  &\mathcal{L}_{\mathrm{RES}}(w) := \left(\frac{h}{2\|
  \mathbf{v}(\mathrm{x})\|}\mathbf{v}(\mathrm{x})\cdot\mathrm{grad}
  [w(\mathbf{x})];\;f(\mathrm{x})\right)_{\Omega}
\end{align}
and $h$ denotes the element-wise diameter.
\subsection{Discontinuous Galerkin} For several transport applications, 
it is highly desirable to possess element-wise mass balance 
property, as it is an important fundamental physical law 
\citep{Turner_JH_2011}. This is 
particularly true when the transport is coupled with chemical reactions 
and biofilm growth \citep{knutson2005pore,von2009three}. The GAL 
and SUPG formulations do not possess this property without any 
further modification or enrichment to their formulations.
One way to ensure this property under a single-field 
finite element framework is through the use of the 
DG formulations (see \citep{Arnold_SIAMJNA_2002,Riviere_CNME_2002,
Cockburn_ZAMM_2003,Li_CMAME_2015,Li_CG_2015,Li_TPM_2016,Pal_IJNME_2016} 
and the references within for further details).
To present the DG formulation employed in the paper, we now 
introduce relevant notation. 

The domain $\Omega$ is divided into $S$ subdomains: 
\begin{align}
  \Omega = \bigcup_{i=1}^{S} \omega_{i}
\end{align}
The boundary of the subdomain $\omega_i$ is 
denoted by $\partial \omega_{i}$. The interior 
face between $\omega_{i}$ and $\omega_{j}$ 
is denoted by 
$\Gamma_{ij}$. That is, 
\begin{align}
  \Gamma_{ij} = \partial \omega_{i} \cap \partial \omega_{j}
\end{align}
The set of all points on the interior faces is 
denoted by $\Gamma_{\mathrm{int}}$. Mathematically, 
\begin{align}
  \Gamma_{\mathrm{int}} = \bigcup_{i=1,i<j}^{S} \Gamma_{ij}
\end{align} 
For an interior face, we denote the subdomains shared by 
this face by $\omega^{+}$ and $\omega^{-}$. The outward normals 
on this face for these subdomains are, respectively, 
denoted by $\widehat{\mathbf{n}}^{+}$ and $\widehat{\mathbf{n}}^{-}$. 
Employing Brezzi's notation \citep{Arnold_SIAMJNA_2002}, 
the average and jump operators on an interior face 
are defined as follows: 
\begin{align}
  \big\{c\big\} := \frac{c^{+} + c^{-}}{2} 
  \quad \mathrm{and} \quad 
  \big[\!\big[c\big]\!\big] := c^{+} \widehat{\mathbf{n}}^{+} 
  + c^{-} \widehat{\mathbf{n}}^{-} 
\end{align}
where 
\begin{align}
  c^{+} = c\vert_{\partial \omega^{+}}  
  \quad \mathrm{and} \quad 
  c^{-} = c\vert_{\partial \omega^{-}} 
\end{align}

One of the most popular DG formulations is the \emph{Interior 
Penalty} method, which for 
equation \eqref{Eqn:S2_D_GE} is written as: 
\begin{align}
 \mathcal{B}_{\mathrm{DG}}(w;c) &:= \Big(\mathrm{grad}[w(\mathbf{x})];\;\mathbf{D}(\mathbf{x})
 \mathrm{grad}[c(\mathbf{x})]\Big)_\Omega - \Big(\big[\!\big[w(\mathbf{x})\big]\!\big];\;\big\{\mathbf{D}(\mathbf{x})\mathrm{grad}[c(\mathbf{x})]\big\}\Big)_{\Gamma_\mathrm{int}} \nonumber \\
 &+\epsilon\Big(\big\{\mathbf{D}(\mathbf{x})\mathrm{grad}[w(\mathbf{x})]\big\};\;\big[\!\big[c(\mathbf{x})\big]\!\big]\Big)_{\Gamma_\mathrm{int}} + \frac{\gamma}{h}\Big(\big[\!\big[w(\mathbf{x})\big]\!\big];\;\big[\!\big[c(\mathbf{x})\big]\!\big]\Big)_{\Gamma_\mathrm{int}}\\
 \mathcal{L}_{\mathrm{DG}}(w) &:= \Big(w(\mathbf{x});\;f(\mathbf{x})
  \Big)_\Omega - \Big(w(\mathbf{x});\;q^{\mathrm{p}}(\mathbf{x})\Big)_{\Gamma^{\mathrm{N}}}
\end{align}
where the penalty term $\gamma = 2\frac{(d + 1)}{d}$ \citep{Shahbazi_SNE_2005} 
for first-order elements and $\epsilon\in [-1,0,1]$ denotes the 
Symmetric, Incomplete, and Non-symmetric Interior Penalty  
methods respectively. For equation \eqref{Eqn:S2_AD_GE}, 
the DG formulation can be written as:
\begin{align}
 \mathcal{B}_{\mathrm{DG}}(w;c) &:= \Big(\mathrm{grad}[w(\mathbf{x})];\;\mathbf{D}(\mathbf{x})
 \mathrm{grad}[c(\mathbf{x})]\Big)_\Omega - \Big(\big[\!\big[w(\mathbf{x})\big]\!\big];\;\big\{\mathbf{D}(\mathbf{x})\mathrm{grad}[c(\mathbf{x})]\big\}\Big)_{\Gamma_\mathrm{int}} \nonumber \\
 &+\epsilon\Big(\big\{\mathbf{D}(\mathbf{x})\mathrm{grad}[w(\mathbf{x})]\big\};\;\big[\!\big[c(\mathbf{x})\big]\!\big]\Big)_{\Gamma_\mathrm{int}} + \frac{\gamma}{h}\Big(\big[\!\big[w(\mathbf{x})\big]\!\big];\;\big[\!\big[c(\mathbf{x})\big]\!\big]\Big)_{\Gamma_\mathrm{int}} \nonumber \\
 &-\Big(w(\mathbf{x});\;\mathbf{v}(\mathbf{x})\cdot\mathrm{grad}[c(\mathbf{x})]\Big)_\Omega -\Big(\big[\!\big[w(\mathbf{x})\big]\!\big];\;c^{\mathrm{up}}(\mathbf{x})\mathbf{v}(\mathbf{x})\Big)_{\Gamma_{\mathrm{int}}}\\
 \mathcal{L}_{\mathrm{DG}}(w) &:= \Big(w(\mathbf{x});\;f(\mathbf{x})
  \Big)_\Omega - \Big(w(\mathbf{x});\;q^{\mathrm{p}}(\mathbf{x})\Big)_{\Gamma^{\mathrm{N}}}
\end{align}
where the upwinding term $c^{\mathrm{up}}(\mathbf{x})$ is defined as:
\begin{align}
c^{\mathrm{up}}(\mathbf{x})&= \left\{\begin{array}{ll}
  c^{+}(\mathbf{x}) & \mathrm{if}\quad\mathbf{v}(\mathbf{x})\cdot\widehat{\mathbf{n}}^{+}(\mathbf{x}) > 0 \\
  c^{-}(\mathbf{x}) & \mathrm{otherwise} 
  \end{array}
  \right.
\end{align}
For the remainder of this paper, we shall consider only the 
Symmetric Interior Penalty method where $\epsilon=-1$.

\subsection{A theoretical discussion}
The bilinear form is assumed to be continuous (i.e.,
bounded above). That is, there exists a constant
$\kappa_{1} > 0$ such that
\begin{align}
\mathcal{B}(w;c) \leq \kappa_1 \|c\|\|w\|
\quad \forall c(\mathbf{x}), w(\mathbf{x})
\in \mathcal{C}
\end{align}
In addition, the bilinear form is assumed to be coercive.
That is, there exists a constant $\kappa_{2} > 0$
such that
\begin{align}
\kappa_2 \|c\|^2 \leq \mathcal{B}(c;c)
\quad \forall c(\mathbf{x}) \in \mathcal{C}
\end{align}
Recall that $\mathcal{L}(\cdot)$ is assumed to be a linear continuous
functional on $\mathcal{C}$. Then, from the Lax-Milgram 
theorem \citep{Brenner_Scott}, it is known that a unique 
solution exists under the WF.
Under the same conditions on the bilinear form and the 
linear functional, a unique solution exists for the
associated VI if $\mathcal{K} \subset \mathcal{C}$ is a closed
convex subset \citep{lions1967variational}. A solution of 
the VI is a solution of the WF if $\mathcal{C}=\mathcal{K}$. 
Moreover, if the bilinear form is symmetric, that is, 
\begin{align}
&\mathcal{B}(w;c) = B(c;w)
\end{align}
then the WF and the VI are equivalent to the following 
MP:~Find 
$c(\mathbf{x}) \in \mathcal{C}$ such that
\begin{align}
\label{Eqn:VI_Minimization_problem}
\mathop{\mathrm{minimize}}_{c(\mathbf{x}) \in \mathcal{C}}
\quad \frac{1}{2}\mathcal{B}(c;c) - \mathcal{L}(c)
\end{align}
These relations are pictorially described in Figure 
\ref{Fig:S2_QP_vs_VI_description}. From a theoretical point of view, 
it is important to note that the VIs that we will be dealing with for steady-state
problems will be elliptic of first kind. For further 
details on infinite-dimensional VIs, see \citep{Glowinski_Numerical_1984,
duvaut1976inequalities}.

%% file: S4_VI_Discrete.tex
\section{PROPOSED COMPUTATIONAL FRAMEWORK IN A DISCRETE SETTING}
\label{Sec:S4_VI_Discrete}
We denote the total number of degrees-of-freedom by ``$ndofs$''. We
also denote the vector of ones by $\boldsymbol{1}$, whose
size will be apparent from the context in which it
is used. The component-wise inequalities are denoted
by $\preceq$ and $\succeq$. That is,
\begin{subequations}
  \begin{align}
    &\boldsymbol{a} \preceq \boldsymbol{b}
    \quad \mbox{implies that} \quad a_{i}
    \leq b_i \; \forall i \\
    &\boldsymbol{a} \succeq \boldsymbol{b}
    \quad \mbox{implies that} \quad a_{i}
    \geq b_i \; \forall i
  \end{align}
\end{subequations}
The vector of unknown nodal concentrations is denoted by
$\boldsymbol{c}$, and the corresponding nodal
source vector is denoted by $\boldsymbol{f}$.
The coefficient matrix after a finite element
discretization is denoted by $\boldsymbol{K}$.
Note that the vectors $\boldsymbol{c}$ and
$\boldsymbol{f}$ are of size $ndofs \times 1$,
and the matrix $\boldsymbol{K}$ is of size
$ndofs \times ndofs$. We denote the standard
inner-product in Euclidean spaces by $\langle
\cdot ; \cdot \rangle$. That is,
\begin{align}
  \langle \boldsymbol{a} ; \boldsymbol{b} \rangle
  = \sum_{i}^{ndofs} a_i b_i \quad \forall \boldsymbol{a},
  \boldsymbol{b} \in \mathbb{R}^{ndofs}
\end{align}

The formulation in the discrete setting will
be posed as a \textsf{mixed complementarity problem
(MCP)} \citep{Kinderlehrer_VI_2000}. For convenience, 
we define $\boldsymbol{h}\in\mathbb{R}^{ndofs}$ as:
\begin{align}
  \boldsymbol{h} := \boldsymbol{K}\boldsymbol{c} - \boldsymbol{f}
\end{align}
The corresponding MCP
reads:~Find $c_{\mathrm{min}}\boldsymbol{1} \preceq
\boldsymbol{c} \preceq c_{\mathrm{max}} \boldsymbol{1}$
such that for each $i\in\left\{1,...,ndofs\right\}$
\begin{subequations}
  \begin{alignat}{2}
    &h_{i}(\boldsymbol{c}) \geq 0 \quad
    &&\mathrm{if} \; c_{\mathrm{min}} = c_{i}  \\
    &h_{i}(\boldsymbol{c}) = 0 \quad
    &&\mathrm{if} \; c_{\mathrm{min}} \leq c_{i} \leq c_{\mathrm{max}} \\
    &h_{i}(\boldsymbol{c}) \leq 0 \quad
    &&\mathrm{if} \; c_{i} =c_{\mathrm{max}}
  \end{alignat}
\end{subequations}
where $c_{\mathrm{min}}$ and $c_{\mathrm{max}}$, respectively,
denote the minimum and maximum concentrations, which are
provided by the maximum principle or the non-negative
constraint. Simple complementarity conditions arise from the first-order 
optimality conditions in optimization. For bound-constrained 
optimization, $\boldsymbol{h}$ corresponds to the gradient 
of the objective functional. If one has only the non-negative 
constraints (i.e., $c_{\mathrm{min}} = 0$ and 
$c_{\mathrm{max}} = +\infty$), then the problem reduces 
to a non-linear complementarity problem, 
which is a special case of MCP. For details 
on non-linear complementarity problems, 
see \citep{Facchinei_FDVI_2003}.
Note that the feasible region, which is restricted by
the bound constraints, form a parallelepiped, which
is a convex set \citep{Boyd_CO_2004}. 

Let the feasible region $\mathcal{K}$ be a convex subset
of $\mathbb{R}^{ndofs}$. In our case, the feasible region
is restricted by constraints which are in the form of
finite number of linear equalities and inequalities. This
makes the feasible region to be a polyhedron, which is a
convex set \citep{Boyd_CO_2004}. It should be
noted that bound constraints are a special case of linear
inequalities. With this machinery at our disposal, one can
pose the second formulation based on variational inequalities,
which reads:~Find $\boldsymbol{c} \in \mathcal{K}$ such that
we have
\begin{align}
  \label{Eqn:S5_vi_formulation}
  \langle \boldsymbol{K}\boldsymbol{c};\boldsymbol{v} - \boldsymbol{c}\rangle
  \geq \langle\boldsymbol{f};\boldsymbol{v}-\boldsymbol{c}\rangle
  \quad \forall \boldsymbol{v} \in \mathcal{K}
\end{align}
  %
\begin{figure}[t]
  \centering
  \subfloat{\includegraphics[scale=0.8]{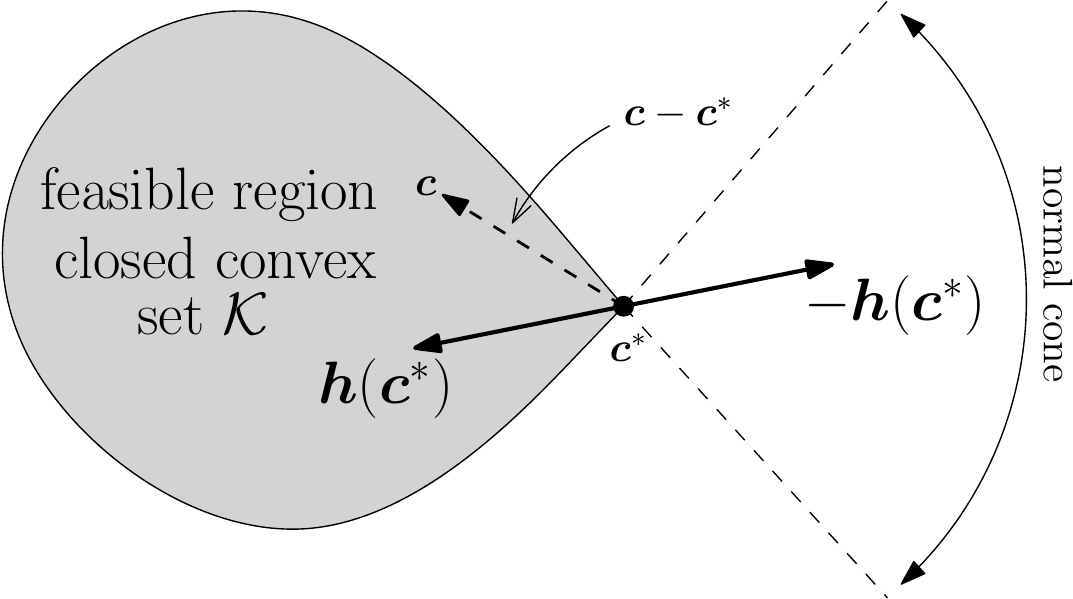}}
  \caption{ This figure illustrates the condition under which a
    solution exists for
    a variational inequality of the form $\langle\boldsymbol{h}(\boldsymbol{c}),
    \tilde{\boldsymbol{c}} - \boldsymbol{c}\rangle\geq 0\;\forall
    \tilde{\boldsymbol{c}}\in\mathcal{K}$. Here,
    $\boldsymbol{c}^{*}$ denotes a solution of the VI. The normal cone 
    of $\mathcal{K}$ at $\boldsymbol{c}^{*}$ is defined as 
    $\mathcal{N}(\boldsymbol{c}^{*}) := \left\{\boldsymbol{w}\in\mathbb{R}^{ndofs}\;
    \Bigm\vert\;\langle\boldsymbol{w};\;\boldsymbol{c}-\boldsymbol{c}^{*}\rangle\leq
    \boldsymbol{0}\;\forall\boldsymbol{c}\in\mathcal{K}\right\}$.
    \label{Fig:S3_VI_description}}
\end{figure}
Note that MCP is a special case of VIs in which the feasible 
region is a parallelepiped (i.e., one has only bound constraints). 
The conditions under which a solution exists 
for the finite-dimensional VI given in equation 
\eqref{Eqn:S5_vi_formulation} is pictorially described
in Figure \ref{Fig:S3_VI_description}.

If the coefficient matrix $\boldsymbol{K}$ is symmetric, one can
alternatively enforce maximum principles and the non-negative constraint
using QP, which has been illustrated in 
\citep{Nagarajan_IJNMF_2011,Nakshatrala_CiCP_2016}
for small-scale problems, and in \citep{Chang_JOMP_2016}
for large-scale problems in parallel environments.
Therefore, this approach is only applicable for formally
self-adjoint differential operators. The formulation can be posed as follows:
\begin{subequations}
  \begin{alignat}{2}
    &\mathop{\mathrm{minimize}}_{\boldsymbol{c} \in \mathbb{R}^{ndofs}}
    &&\quad \frac{1}{2} \langle \boldsymbol{c} ; \boldsymbol{K}
    \boldsymbol{c} \rangle - \langle \boldsymbol{c} ;
    \boldsymbol{f} \rangle \\
    &\mbox{subject to} &&\quad c_{\mathrm{min}} \boldsymbol{1}
    \preceq \boldsymbol{c} \preceq c_{\mathrm{max}}\boldsymbol{1}
  \end{alignat}
\end{subequations}
In addition, if $\boldsymbol{K}$ is
positive-definite the objective function becomes convex.
The resulting optimization problem then belongs to the 
special case of convex quadratic programming for which
sophisticated solvers exist.
%
\begin{remark}
It should be mentioned that a quick fix to eliminate negative
violations is through the so-called clipping procedure. However,
this procedure is rather \emph{ad hoc} and, more importantly, it
is not variationally consistent. On the other hand, the proposed
VI-based computational framework not only ensures non-negative 
solutions but also has a firm variational basis. We will also 
illustrate that the solutions under the proposed framework need 
not necessarily match the solution under the clipping procedure.
\end{remark}

\subsection{Theoretical results in the discrete setting}
In this paper, we are interested in problems with two different 
cases of bound constraints. In the first case, we have both lower 
and upper bounds. In the second case, we have only the lower 
bound. The lower bound typically comes the non-negative 
constraint, and the upper bound comes from maximum principles. 
We now discuss existence results for finite-dimensional VIs 
under the mentioned two cases of bound constraints. 

We begin by noting that the feasible set 
$\mathcal{K}$ will be convex and closed for both 
the sets of bound constraints. In the first case, 
the feasible set will also be bounded, which makes 
the feasible set to be compact (which, in the context 
of Euclidean spaces, is equivalent to closed and bounded). 
We therefore deal with both the cases separately.
\begin{figure}[t]
  \centering
  \subfloat{\includegraphics[scale=0.8]{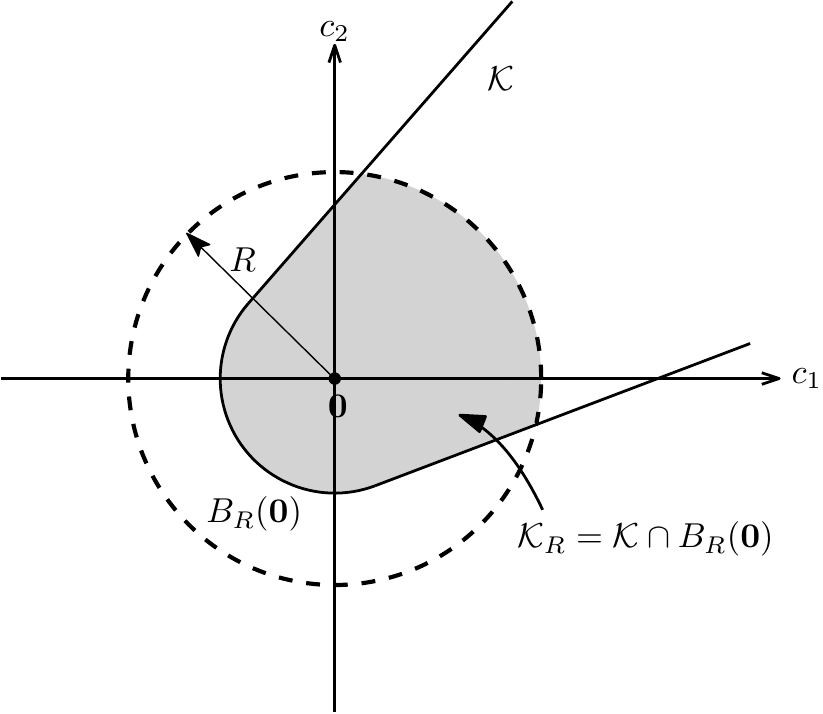}}
  \caption{A pictorial description of $B_{R}(\mathbf{0})$ and
    $\mathcal{K}_{R}$. These sets are used in Theorem
    \ref{Thm:Existence_K_R}. \label{Fig:S4_VI_lowerbound}}
\end{figure} 

\begin{theorem}{(Existence based on compactness of $\mathcal{K}$)}
  \label{Theorem:existence}
  If $\mathcal{K}$ is compact and convex, then 
  a solution exists to the finite-dimensional VI 
  \eqref{Eqn:S5_vi_formulation}.
\end{theorem}
\begin{proof}
  A proof can be constructed using the Brouwer's
  fixed point theorem and can be found in
  \citep{Facchinei_FDVI_2003}.
\end{proof}

\begin{theorem}{(Existence based on positive-definiteness
    of $\mathrm{sym}[\boldsymbol{K}]$)}
  \label{Theorem:existence_PD}
  If the symmetric part of the coefficient matrix $\boldsymbol{K}$ 
  (i.e., $\mathrm{sym}[\boldsymbol{K}]$) is positive-definite, a 
  solution to the finite-dimensional VI \eqref{Eqn:S5_vi_formulation} 
  exists. Note that the feasible set $\mathcal{K}$ need not be compact. 
\end{theorem}
\begin{proof}
  Let 
  \begin{align}
    \boldsymbol{g}(\boldsymbol{c}) := \boldsymbol{K}\boldsymbol{c} 
- \boldsymbol{f}
  \end{align}
  The VI then becomes 
  \begin{align}
    \label{Eqn:VI_modified_VI}
    \langle\boldsymbol{g}(\boldsymbol{c});\widetilde{\boldsymbol{c}} 
    - \boldsymbol{c}\rangle \quad \forall \widetilde{\boldsymbol{c}} 
    \in \mathcal{K}
  \end{align}
  Clearly, the function $\boldsymbol{g}(\boldsymbol{c})$
  is continuous. Moreover, the function $\boldsymbol{g}
  (\boldsymbol{c})$ satisfies the following coercive
  condition:
  \begin{align}
    \frac{\langle \boldsymbol{g}(\boldsymbol{c}) - 
    \boldsymbol{g}(\widetilde{\boldsymbol{c}});\boldsymbol{c} - 
    \widetilde{\boldsymbol{c}}\rangle}{\|\boldsymbol{c} - 
      \widetilde{\boldsymbol{c}}\|} \rightarrow 
    \infty 
  \end{align}
  as $\|\boldsymbol{c}\| \rightarrow \infty$. To wit, 
  since $\mathrm{sym}[\boldsymbol{K}]$ is positive-definite 
  and symmetric, the minimum eigenvalue $\lambda_{\mathrm{min}}$ 
  is real and positive. One can then write:
  \begin{align}
    \lambda_{\mathrm{min}} \|\boldsymbol{c} - \widetilde{\boldsymbol{c}}\|^2 
    \leq \left(\boldsymbol{c} - \widetilde{\boldsymbol{c}}\right) \cdot 
    \mathrm{sym}[\boldsymbol{K}] \left(\boldsymbol{c} - 
    \widetilde{\boldsymbol{c}} \right) 
    = \left(\boldsymbol{c} - \widetilde{\boldsymbol{c}}\right) \cdot 
    \boldsymbol{K} \left(\boldsymbol{c} - 
    \widetilde{\boldsymbol{c}} \right) 
    = \langle \boldsymbol{g}(\boldsymbol{c}) - 
    \boldsymbol{g}(\widetilde{\boldsymbol{c}});
    \boldsymbol{c} - \widetilde{\boldsymbol{c}} \rangle 
  \end{align}
  That is,
  \begin{align}
    \lambda_{\mathrm{min}} \|\boldsymbol{c} - \widetilde{\boldsymbol{c}}\| 
    \leq \frac{\langle \boldsymbol{g}(\boldsymbol{c}) - 
    \boldsymbol{g}(\widetilde{\boldsymbol{c}});
    \boldsymbol{c} - \widetilde{\boldsymbol{c}} \rangle}{\|\boldsymbol{c} 
      - \widetilde{\boldsymbol{c}}\|}
  \end{align}
  We thus have shown that the function $\boldsymbol{g}(\boldsymbol{c})$ 
  is continuous and coercive. Under such conditions, a 
  solution exists to the VI \eqref{Eqn:VI_modified_VI}
  (e.g., see \citep{Nagurney_2002,nagurney2012projected}).
\end{proof}

We next present another existence theorem which is
particularly useful when the feasible set is unbounded
(for example, when we have only one of the bounds --
either lower or upper bounds). Let $B_{R}(\mathbf{0})$
is a hypersphere of radius $R$ centered at $\mathbf{0}$,
and let $\mathcal{K}_{R} = \mathcal{K} \cap B_{R}(\mathbf{0})$
(see Figure \ref{Fig:S4_VI_lowerbound}). Clearly,
$\mathcal{K}_{R}$ is bounded. 
\begin{theorem}{(Existence based on $\mathcal{K}_{R}$)}
  \label{Thm:Existence_K_R}
  A solution exists to the VI \eqref{Eqn:S5_vi_formulation}
  on $\mathcal{K}$ (which need not be bounded) if and only
  if there exists $R > 0$ and a solution $\boldsymbol{c}^{*}
  \in \mathcal{K}_{R}$ that satisfies the following
  VI: 
  \begin{align}
    \langle \boldsymbol{K} \boldsymbol{c}^{*};
    \widetilde{\boldsymbol{c}}  - \boldsymbol{c}^{*}
    \rangle \quad \forall \widetilde{\boldsymbol{c}}
    \in \mathcal{K}_{R}
  \end{align}
  which is defined on a bounded set. 
\end{theorem}
\begin{proof}
  See \citep{Nagurney_2002}.
\end{proof}

\begin{theorem}{(Uniqueness)}\label{Theorem:uniqueness}
If the symmetric part of the coefficient matrix $\boldsymbol{K}$ 
is positive-definite, then the finite-dimensional VI \eqref{Eqn:S5_vi_formulation} 
has a unique solution if it exists.
\end{theorem}
\begin{proof}
On the contrary, assume that $\boldsymbol{c}_1$ and $\boldsymbol{c}_2$ are 
two different solutions of the VI \eqref{Eqn:S5_vi_formulation}. This implies that
\begin{align}
  \label{Eqn:S4_vi_c1}
  \langle \boldsymbol{K}\boldsymbol{c}_1;\boldsymbol{v} - \boldsymbol{c}_1\rangle
  \geq \langle\boldsymbol{f};\boldsymbol{v}-\boldsymbol{c}_1\rangle
  \quad \forall \boldsymbol{v} \in \mathcal{K}\\
  \label{Eqn:S4_vi_c2}
  \langle \boldsymbol{K}\boldsymbol{c}_2;\boldsymbol{v} - \boldsymbol{c}_2\rangle
  \geq \langle\boldsymbol{f};\boldsymbol{v}-\boldsymbol{c}_2\rangle
  \quad \forall \boldsymbol{v} \in \mathcal{K}
\end{align}
Since $\boldsymbol{c}_1,\boldsymbol{c}_2\in\mathcal{K}$, choose 
$\boldsymbol{v} = \boldsymbol{c}_2$ in \eqref{Eqn:S4_vi_c1} and 
$\boldsymbol{v} = \boldsymbol{c}_1$ in \eqref{Eqn:S4_vi_c2}. This
results in
\begin{align}
  \langle \boldsymbol{K}\boldsymbol{c}_1;\boldsymbol{c}_2 - 
  \boldsymbol{c}_1\rangle \geq \langle\boldsymbol{f};\boldsymbol{c}_2-\boldsymbol{c}_1\rangle
\\
  \langle \boldsymbol{K}\boldsymbol{c}_2;\boldsymbol{c}_1 - 
  \boldsymbol{c}_2\rangle \geq \langle\boldsymbol{f};\boldsymbol{c}_1-\boldsymbol{c}_2\rangle
\end{align}
Summing the above two inequalities and invoking the linearity in the second slot, 
we obtain
\begin{align}
  \langle \boldsymbol{K}(\boldsymbol{c}_1-\boldsymbol{c}_2);\boldsymbol{c}_1 - 
  \boldsymbol{c}_2\rangle \leq \boldsymbol{0}
\end{align}
which further implies that
\begin{align}
  \label{Eqn:S4_vi_contradiction}
  \langle \mathrm{sym}[\boldsymbol{K}](\boldsymbol{c}_1-\boldsymbol{c}_2);\boldsymbol{c}_1-\boldsymbol{c}_2\rangle \leq \boldsymbol{0}
\end{align}
On the other hand, the positive-definiteness of 
$\mathrm{sym}[\boldsymbol{K}]$ and our assumption
$\boldsymbol{c}_1-\boldsymbol{c}_2 \neq \boldsymbol{0}$ imply that
\begin{align}
  \langle \mathrm{sym}[\boldsymbol{K}](\boldsymbol{c}_1-\boldsymbol{c}_2);
  \boldsymbol{c}_1-\boldsymbol{c}_2\rangle > \boldsymbol{0}
\end{align}
which contradicts the inequality given by equation 
\eqref{Eqn:S4_vi_contradiction}. Hence, $\boldsymbol{c}_1=\boldsymbol{c}_2$.
\end{proof}

Using the aforementioned general existence and uniqueness
results for VIs, we now establish the existence and
uniqueness of solutions under the proposed framework
in the discrete setting. 
\begin{theorem}{(Well-posedness of the proposed framework)}
  Unique solutions exist for the VIs from the GAL, SUPG
  and DG WFs under lower bounds (which arise from the
  non-negative constraint) and under both lower and
  upper bounds (which arise from maximum principles).
\end{theorem}
\begin{proof}
  First, it should be noted that the symmetric part
  of the coefficient matrices under the GAL and DG
  formulations are positive-definite. On the other
  hand, the stabilization term under the SUPG
  formulation does not guarantee that the symmetric
  part of the coefficient matrix to be positive-definite.
  It should also be noted that the stabilization term
  in equation \eqref{Eqn:S2_AD_bilinear_stab} is
  $\mathcal{O}(h)$, where $h$ is the characteristic
  mesh size. This means that there exist a critical
  mesh size, $h_{\mathrm{crit}}$, such that if $h <
  h_{\mathrm{crit}}$ then the contribution from the
  residual terms to the coefficient matrix will
  be small, and the resulting symmetric part of
  the coefficient matrix will be positive-definite.
  
  (\emph{Existence}.) If both the lower and upper bounds are
  present, the feasible region will be compact. For this
  case of bound constraints, Theorem \ref{Theorem:existence}
  establishes the existence of solutions for the VIs arising
  from all the three WFs (i.e., GAL, SUPG and DG). If only the
  lower bounds are present, Theorem \ref{Theorem:existence_PD}
  will ensure the existence of solutions for VIs arising from
  the GAL and DG formulations, and Theorem \ref{Thm:Existence_K_R}
  will ensure the existence of solutions for the VIs arising from
  the SUPG formulation on a general mesh. Of course, if the
  mesh is adequately refined (i.e., $h < h_{\mathrm{crit}}$)
  then Theorem \ref{Theorem:existence_PD} can also ensure
  the existence of a solution for the VIs arising under
  the SUPG formulation. 
 
  (\emph{Uniqueness}.) Theorem \ref{Theorem:uniqueness}
  provides the uniqueness of solution for the VIs arising
  from the GAL and SUPG formulations. As discussed above,
  upon an adequate mesh refinement, $\mathrm{sym}[\boldsymbol{K}]$
  will be positive-definite under the SUPG formulation. On those
  meshes, Theorem \ref{Theorem:uniqueness} will provide the
  uniqueness of solutions for the VIs arising from the SUPG
  formulation. 
\end{proof}

\subsection{Computer implementation details}
In this paper, the proposed QP and VI-based
formulations for the GAL, SUPG, and DG formulations
are implemented through the Firedrake project (see Appendix \ref{A1:code} for
further details), but one can employ any other finite element
library. The primary advantage of the Firedrake project is that it 
provides easy access to parallel solvers, specifically the
PETSc and TAO libraries \citep{petsc-user-ref,petsc-efficient,tao-user-ref}
which are built on top of \textsf{Message Passing Interface (MPI)}.
Appropriate iterative solvers and preconditioners are needed for large-scale
problems, and the PETSc library provides the necessary data structures. 
The \textsf{Conjugate Gradient (CG)} method is used for symmetric problems like 
the diffusion equation whereas the \textsf{Generalized Minimal Residual (GMRES)} 
method is used for the non-symmetric advection-diffusion equation.
\emph{HYPRE}'s algebraic multi-grid package \citep{hypre-web-page} 
is used as the preconditioner.
\subsubsection{Solvers}
The main ingredient of the proposed computational framework is to solve
finite-dimensional VIs. There are several solvers
available for solving these type of inequalities in a 
large-scale parallel environment. However, the performance of 
these solvers is problem-specific. It is, therefore, necessary to 
identify the best performing VI solver for our case, which is primarily to
enforcing maximum principles and the non-negative constraint. To this end, we 
consider the following VI and QP solvers available through the PETSc and
TAO libraries:
\begin{enumerate}[leftmargin=*,label=(\roman*)]
\item \textsf{Semi-smooth (VI - SS)}: TAO's implementation of the
semi-smooth algorithm \citep{DeLuca_MP_1996,Munson_INFORMS_2001}
reformulates the MCP as a
non-smooth system of equations using the
Fischer-Burmeister function \citep{Fischer_OPT_1992}. This function,
$\phi:\mathbb{R}^2\rightarrow\mathbb{R}$, is defined as
\begin{subequations}
\begin{align}
\phi(a,b) &:= \sqrt{a^2+b^2}-a-b \\
\phi(a,b) &= 0\quad \Leftrightarrow\quad a \geq 0,\;b\geq 0,\;ab=0
\end{align}
\end{subequations}
The reformulation of the MCP is handled component-wise, and
the system of equations $\Phi(\boldsymbol{c})=0$ where
$\Phi:\mathbb{R}^{ndofs}\rightarrow\mathbb{R}^{ndofs}$ is expressed as:
\begin{align}
\Phi_i(\boldsymbol{c}) &:= \left\{\begin{array}{ll}
\phi\left(c_i-c_{\mathrm{min}},\;h_i(\boldsymbol{c})\right) &
\mathrm{if}\;-\infty < c_{\mathrm{min}} < c_{\mathrm{max}} = \infty \\
\phi\left(c_{\mathrm{max}}-c_i,\;-h_i(\boldsymbol{c})\right) &
\mathrm{if}\;-\infty = c_{\mathrm{min}} < c_{\mathrm{max}} < \infty \\
\phi\left(c_i-c_{\mathrm{min}},\;\phi\left(c_{\mathrm{max}}-c_i,
\;-h_i(\boldsymbol{c})\right)\right) &
\mathrm{if}\;-\infty < c_{\mathrm{min}} < c_{\mathrm{max}} < \infty \\
-h_i(\boldsymbol{c}) & \mathrm{if}\;-\infty = c_{\mathrm{min}} <
c_{\mathrm{max}} = \infty \\
c_{\mathrm{min}}-c_i &  \mathrm{if}\;-\infty < c_{\mathrm{min}} =
c_{\mathrm{max}} < \infty \\
\end{array}
\right.
\end{align}
It should be noted that $\Phi(\boldsymbol{c})$ is not differentiable
everywhere but it still satisfies a semi-smoothness property
\citep{Mifflin_JCO_1977,Qi_MOR_1993,Qi_MP_1993}.
The above system of equations is used to compute a descent direction, and
the solver finishes as soon as the natural merit function $\Psi(\boldsymbol{c}):=
\frac{1}{2}\|\Phi(\boldsymbol{c})\|_2^2$ meets some level of tolerance.
We also employ TAO's feasible line-search algorithm which ensures that
the solution is within the bounds by using a projected Armijo line
search \citep{Armijo_PJM_1966}.
\item \textsf{Reduced-space active-set (VI - RS)}: The reduced-space
active-set method selects an active-set and solves a reduced linear
system of equations to calculate a direction of the gradient descent.
The active and inactive sets are, respectively, defined as:
\begin{subequations}
\begin{align}
\mathcal{A}(\boldsymbol{c}) &:= \{i\in\left\{1,...,ndofs\right\}\;\Bigm|\;c_i=0\;
\mbox{and}\;h_i(\boldsymbol{c})>0\}\\
\mathcal{I}(\boldsymbol{c}) &:= \{i\in\left\{1,...,ndofs\right\}\;\Bigm|\;c_i>0\;
\mbox{or}\;h_i(\boldsymbol{c})\leq 0\}
\end{align}
\end{subequations}
The active set $\mathcal{A}(\boldsymbol{c})$ represents regions
where the lower bound is active thus the function value can be
ignored, and the inactive set $\mathcal{I}(\boldsymbol{c})$ contains
everything else. The descent direction of the active set is set to
zero whereas the descent direction of the inactive set is approximated,
and the solution is updated using a projected line search. As far as we
know, there is little documentation on the theoretical and mathematical
convergence properties for this particular algorithm, but the computational
results from \citep{Benson_OMS_2006} demonstrate that this solver is
robust and can handle a wide range of applications. For further
implementation details of these two VI solvers, we defer all interested
readers to \citep{Benson_OMS_2006,tao-user-ref} and the references within.
\item \textsf{Trust region Newton (QP - TRON)}: Unlike the previous
two solvers, the trust region Newton method \citep{Lin_SIAMJO_1999}
is an active-set solver designed for large-scale minimization problems.
It uses the gradient projection to generate a Cauchy step and the preconditioned
CG with an incomplete Cholesky factorization to generate a descent direction. 
Each iteration of the TRON algorithm solves a reduced linear system containing 
variables that lie between the lower and upper bounds. The algorithm then 
applies a trust region to the conjugate gradients to ensure convergence. 
The algorithmic scalability and hardware performance of this solver has 
been thoroughly documented in \citep{Chang_JOMP_2016}, so the computational 
results arising from QP - TRON serves primarily as a benchmark for comparison with
the VI solvers.
\end{enumerate}
We acknowledge that there may be several other QP and VI solvers which
are not covered in this paper. Nonetheless, the computational framework that we
propose is algorithm-independent and platform-agnostic, so one is free to either
employ different solvers or modify the above implementations of the
QP and VI algorithms to cater to specific needs and applications.
\subsubsection{An outline of the algorithm}
\label{SS4:proposed_framework}
The performance of non-linear and optimization-based solvers 
depends on accurate initial guesses. To this end, we propose
the following steps for the overall implementation of 
the proposed computational framework:
\begin{enumerate}[label=\textsf{Step \arabic*}:]
\item Assemble $\boldsymbol{K}$ and $\boldsymbol{f}$
\item Solve for $\boldsymbol{c}_0$.
\item Clip $\boldsymbol{c}_0$ and obtain 
$\boldsymbol{c}_{\mathrm{CLIP}}$. Formally:
\begin{align}
\label{Eqn:S4_clipping}
\boldsymbol{c}_{\mathrm{CLIP}} = \mbox{arg min}\frac{1}{2}\|\boldsymbol{c}-\boldsymbol{c}_0\|^{2}\quad\mbox{subject to:}\;c_{\mathrm{min}} \boldsymbol{1}
\preceq \boldsymbol{c} \preceq c_{\mathrm{max}}\boldsymbol{1}
\end{align}
\item Solve the bounded constraint problem under the QP or VI framework
with $\boldsymbol{c}_{\mathrm{CLIP}}$ as the initial guess.
\end{enumerate}
It should be emphasized that one need not solve equation \eqref{Eqn:S4_clipping} 
to implement clipping procedure. Instead, one trims nodal values to meet the 
desired bounds.
Since the governing equations are linear, $\boldsymbol{K}$ and $\boldsymbol{f}$ only
need to be assembled once and are reused 
for the various QP and VI evaluation routines for \textsf{Step 4}. 
Python implementations of the VI - SS, VI - RS, and QP - TRON solvers 
leveraging \textsf{petsc4py} \citep{dalcin_AWR_2011}
capabilities are shown in Listings \ref{Code:ss}, \ref{Code:rs}, and \ref{Code:tron}
respectively of Appendix \ref{A1:code}. For the steady-state 3D benchmarks in 
the next section, the KSP relative tolerance is set to $10^{-7}$
for the solver in \textsf{Step 2}, whereas the KSP relative tolerance for
approximating the gradient descents in \textsf{Step 4} is set
to $10^{-3}$. It was shown in \citep{Chang_JOMP_2016} that 
relaxing the tolerance requires more non-linear iterations but
lessens the overall solve time. Relaxing the tolerance also lessens
the arithmetic intensity where the performance is governed by the memory
bandwidth thus making it less likely to achieve good speedup on a 
shared compute node. In other words, the parallel efficiency of the QP and 
VI solvers are likely to be worse than solving the WF with standard KSP
convergence tolerances. The absolute convergence tolerances for the QP and VI
solvers are set to $10^{-8}$ although it should be mentioned that 
the optimal values depends on the application at hand. All large-scale 
computations are conducted on Intel Xeon E5-2680v2 processors where each
MPI process is restricted to a single core.

%% file: S5_VI_NR.tex
\begin{figure}[t]
\centering
\subfloat[Diffusion problem]{\includegraphics[scale=0.45]{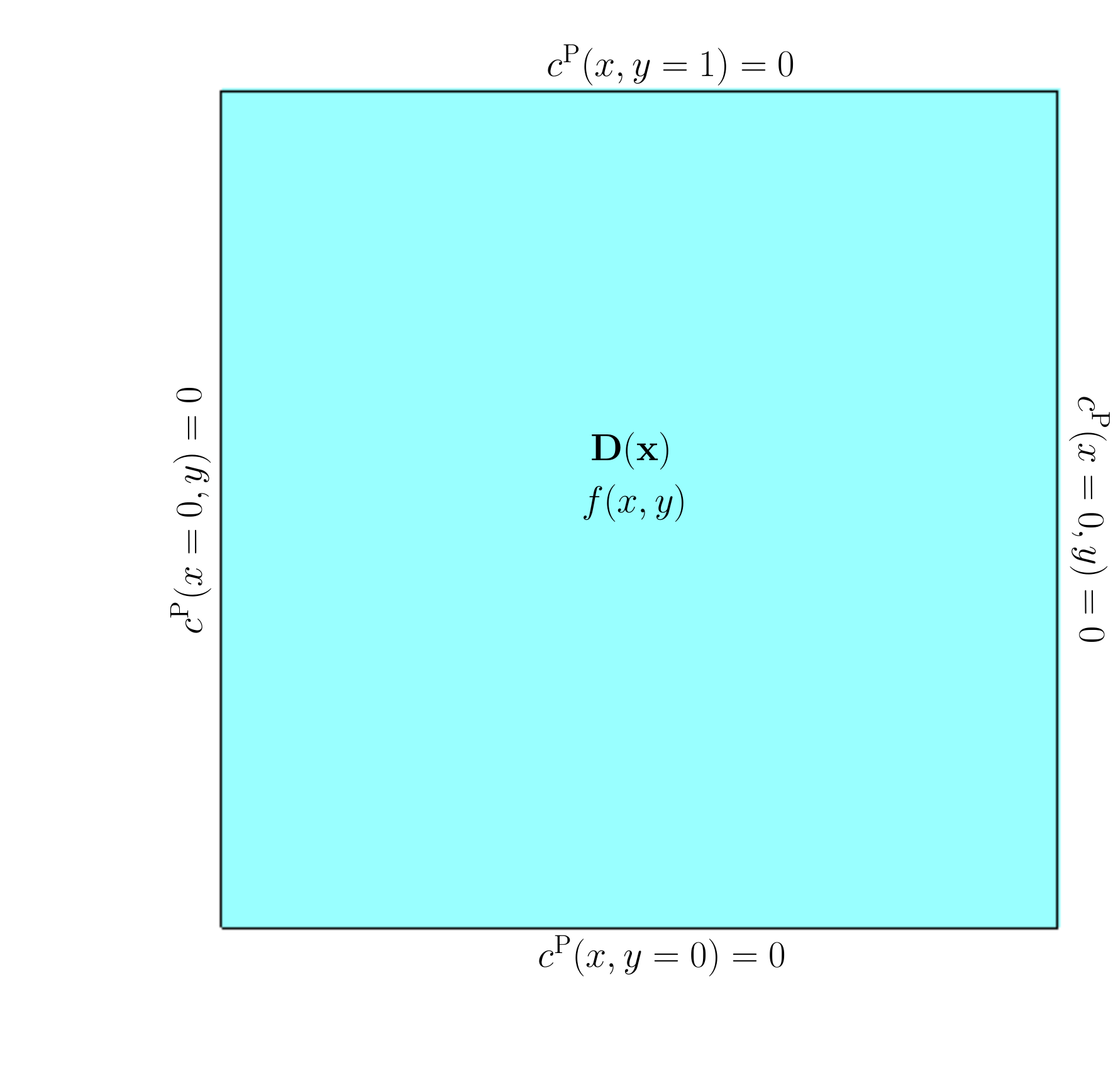}
\label{Fig:2D_D_description}}
\subfloat[Advection-diffusion problem]{\includegraphics[scale=0.45]
{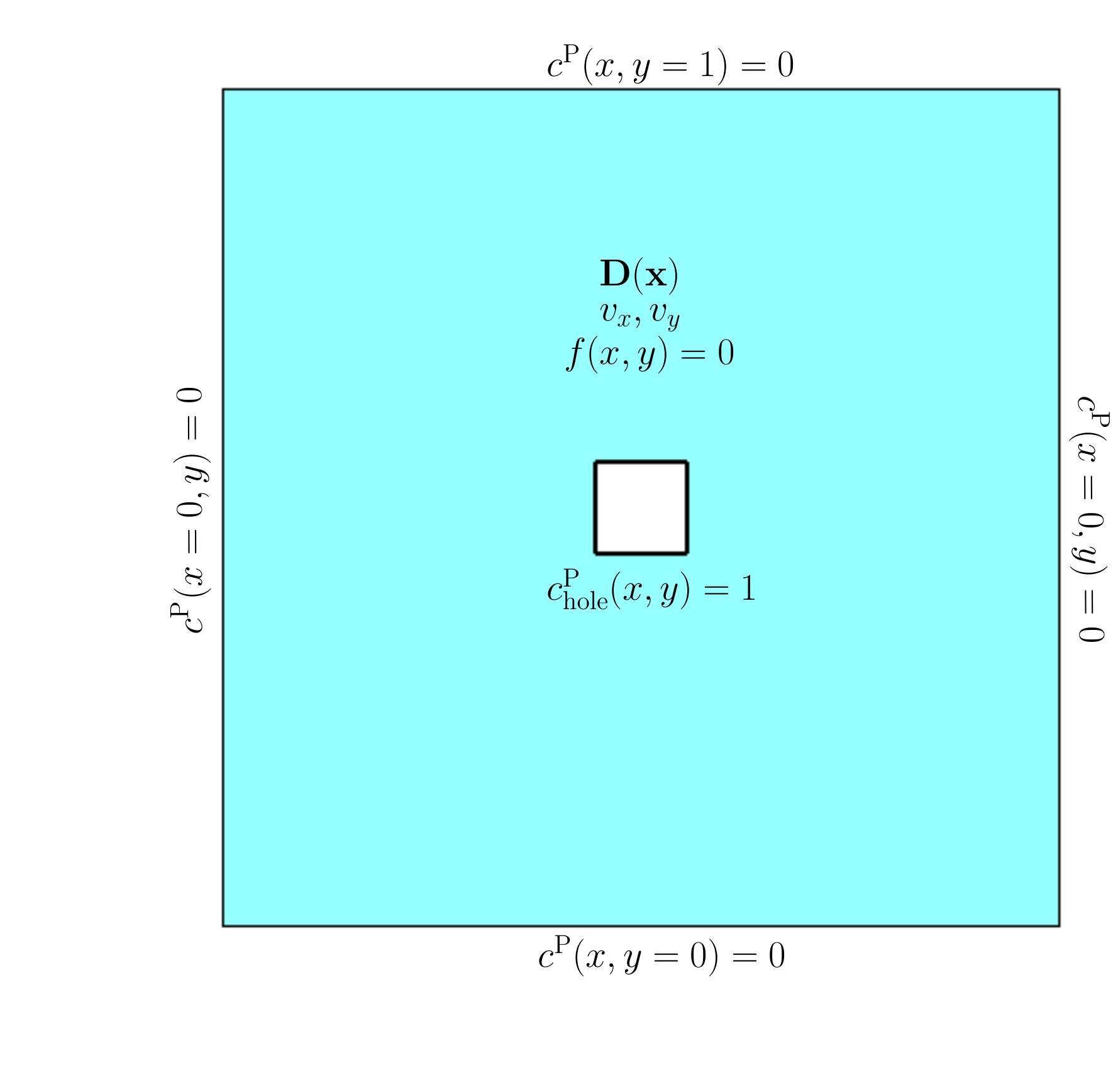}\label{Fig:2D_AD_description}}
\caption{2D benchmarks: Pictorial description of the boundary value problems 
for the diffusion and advection-diffusion examples. The finite element mesh 
for the diffusion problem contains 40,000 quadrilateral elements and 40,401 nodes. 
The mesh for the advection-diffusion problem is unstructured containing 96,430 triangular
elements and 48,663 nodes.}
\label{Fig:2D_description}
\end{figure}
\section{NUMERICAL RESULTS FOR STEADY-STATE RESPONSE}
\label{Sec:S5_VI_NR}
\subsection{2D benchmarks}
\begin{figure}[t]
\centering
\subfloat[$\boldsymbol{c}_{\mathrm{GAL}}$]{\includegraphics[scale=0.45]{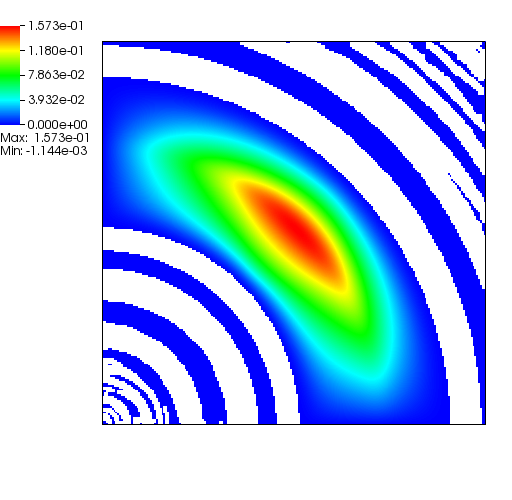}}
\subfloat[$\boldsymbol{c}_{\mathrm{TRON}}$]{\includegraphics[scale=0.45]{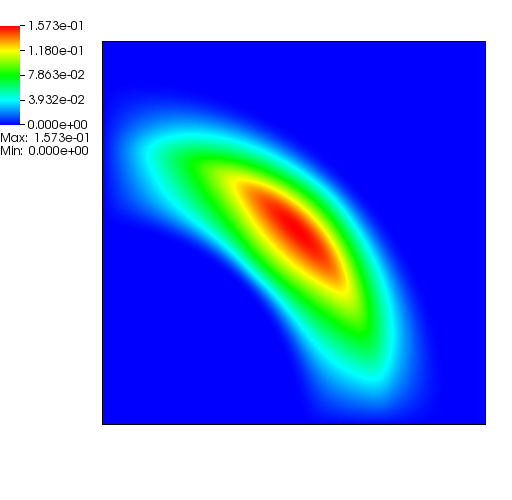}}\\
\subfloat[$\boldsymbol{c}_{\mathrm{SS}}$]{\includegraphics[scale=0.45]{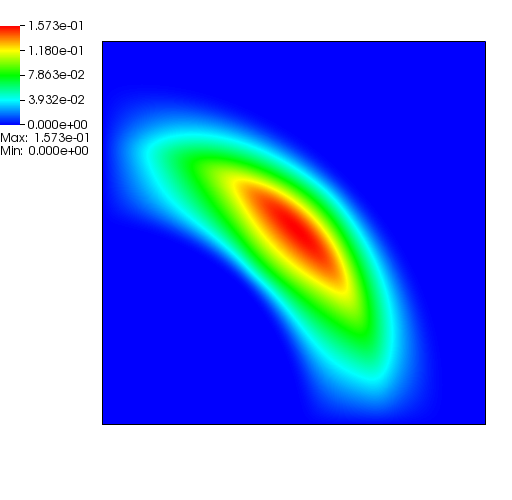}}
\subfloat[$\boldsymbol{c}_{\mathrm{RS}}$]{\includegraphics[scale=0.45]{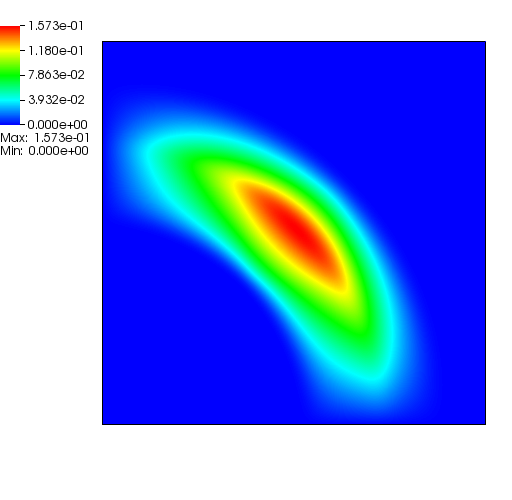}}
\caption{2D diffusion: concentrations under the Galerkin ($\boldsymbol{c}_{\mathrm{GAL}}$), TRON ($\boldsymbol{c}_{\mathrm{TRON}}$), semi-smooth ($\boldsymbol{c}_{\mathrm{SS}}$), and reduced-space active-set ($\boldsymbol{c}_{\mathrm{RS}}$) methods where the white regions represent negative concentrations.}
\label{Fig:2D_D_solutions}
\end{figure}
\begin{figure}[t]
\centering
\subfloat[$\|\boldsymbol{c}_{\mathrm{CLIP}}-\boldsymbol{c}_{\mathrm{TRON}}\|$]{\includegraphics[scale=0.45]{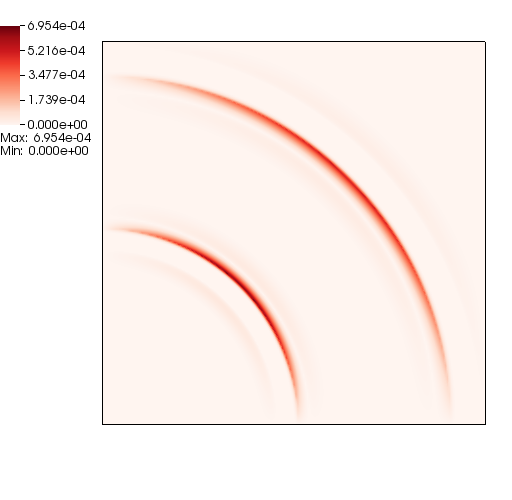}}
\subfloat[$\|\boldsymbol{c}_{\mathrm{CLIP}}-\boldsymbol{c}_{\mathrm{SS}}\|$]{\includegraphics[scale=0.45]{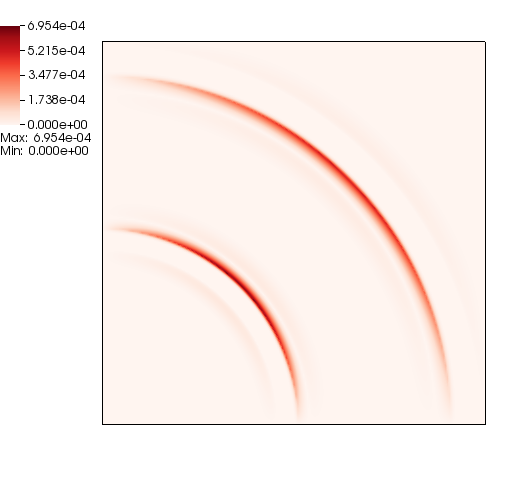}}\\
\subfloat[$\|\boldsymbol{c}_{\mathrm{CLIP}}-\boldsymbol{c}_{\mathrm{RS}}\|$]{\includegraphics[scale=0.45]{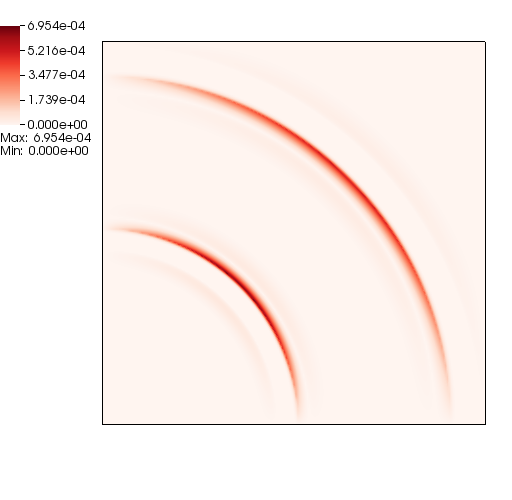}}
\caption{2D diffusion: absolute difference in concentrations between the clipped solution and the non-negative solution.}
\label{Fig:2D_D_clipped}
\end{figure}
\begin{figure}[t]
\centering
\subfloat[$\|\boldsymbol{c}_{\mathrm{TRON}}-\boldsymbol{c}_{\mathrm{SS}}\|$]{\includegraphics[scale=0.45]{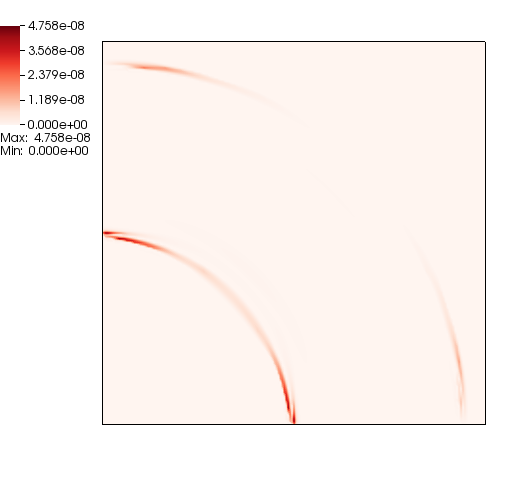}}
\subfloat[$\|\boldsymbol{c}_{\mathrm{TRON}}-\boldsymbol{c}_{\mathrm{RS}}\|$]{\includegraphics[scale=0.45]{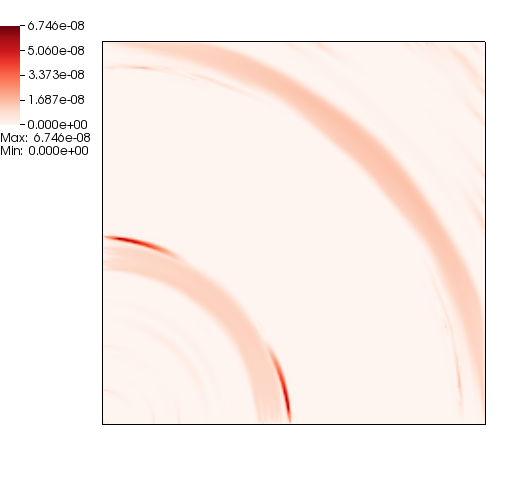}}\\
\subfloat[$\|\boldsymbol{c}_{\mathrm{SS}}-\boldsymbol{c}_{\mathrm{RS}}\|$]{\includegraphics[scale=0.45]{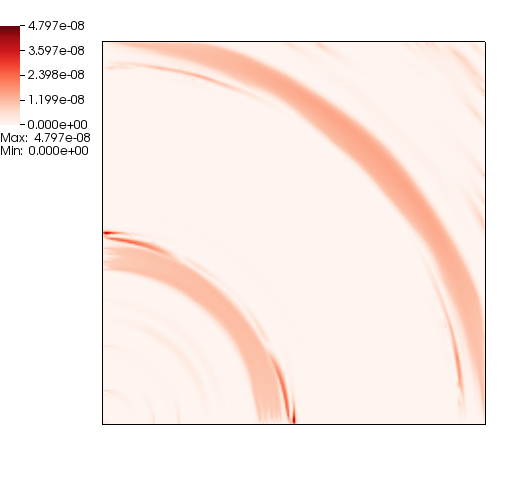}}
\caption{2D diffusion: absolute difference in concentrations between the various non-negative methodologies}
\label{Fig:2D_D_differences}
\end{figure}
We now examine 2D problems in order to demonstrate the effectiveness 
of the proposed computational algorithms for ensuring 
discrete maximum principles and the non-negative constraint. Only the 
GAL and SUPG formulations are employed in this numerical study. 
First, let us consider 
the pure diffusion equation on a bi-unit square: $\Omega := (0,1)\times(0,1)$ 
as shown in Figure \ref{Fig:2D_D_description}. The following heterogeneous 
and anisotropic diffusivity tensor similar to the one considered in 
\citep{LePotier_CRM_2005} is used:
\begin{align}
  \mathbf{D}(\mathbf{x}) &= \left(\begin{array}{cc}
   y^2+\epsilon x^2 & -(1-\epsilon)xy \\
  -(1-\epsilon)xy & x^2+\epsilon y^2 
  \end{array}
  \right)
\label{Eqn:S5_2D_D_diffusivity}
\end{align}
where $\epsilon = 10^{-4}$. The forcing function is defined as $f(x,y)=1$ if $(x,y)\in 
\left[\frac{3}{8},\frac{5}{8}\right]\times\left[\frac{3}{5},\frac{5}{8}\right]$ and 
zero elsewhere. Homogeneous boundary conditions are applied on all four sides of 
the domain. Numerical solutions under the GAL, VI - SS, VI - RS, and QP - TRON 
methods with uniform quadrilateral elements of $h$-size = 1/200 are shown in Figure 
\ref{Fig:2D_D_solutions}. All three non-negative solvers successfully 
eliminate negative concentrations, and the absolute difference plots in 
Figure \ref{Fig:2D_D_clipped} show that their results are quite different 
than from the one arising from the standard clipping procedure. 
Moreover, the absolute differences between the various QP and VI solvers, 
as seen from Figure \ref{Fig:2D_D_differences}, are extremely small and suggest 
that the QP and VI solvers have similar numerical accuracy. 

\begin{figure}[t]
\centering
\subfloat[$\boldsymbol{c}_{\mathrm{SUPG}}$]{\includegraphics[scale=0.45]{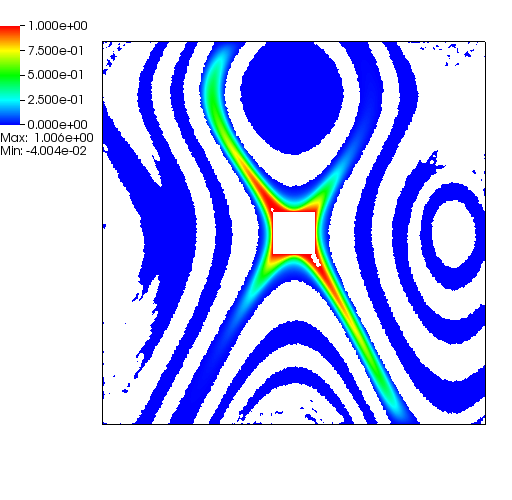}}
\subfloat[$\boldsymbol{c}_{\mathrm{SS}}$]{\includegraphics[scale=0.45]{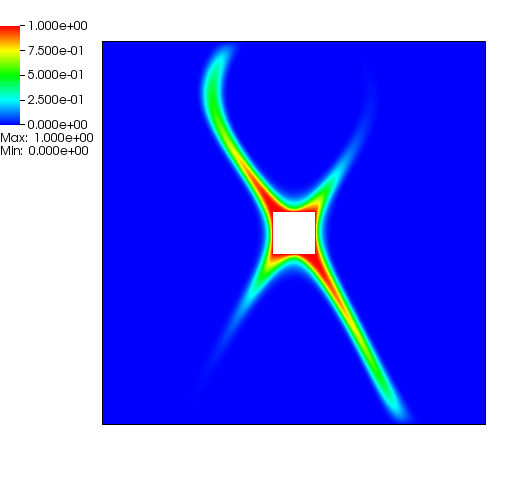}}\\
\subfloat[$\boldsymbol{c}_{\mathrm{RS}}$]{\includegraphics[scale=0.45]{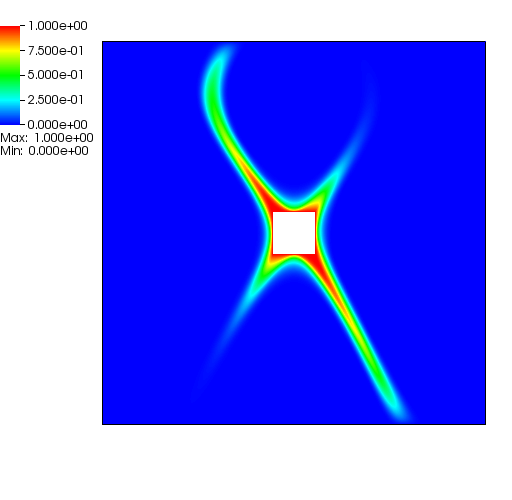}}
\caption{2D advection-diffusion: concentrations under the SUPG ($\boldsymbol{c}_{\mathrm{SUPG}}$), semi-smooth ($\boldsymbol{c}_{\mathrm{SS}}$), and reduced-space active-set ($\boldsymbol{c}_{\mathrm{RS}}$) methods where the white regions represent negative concentrations.}
\label{Fig:2D_AD_solutions}
\end{figure}
\begin{figure}[t]
\centering
\subfloat[$\|\boldsymbol{c}_{\mathrm{CLIP}}-\boldsymbol{c}_{\mathrm{SS}}\|$]{\includegraphics[scale=0.45]{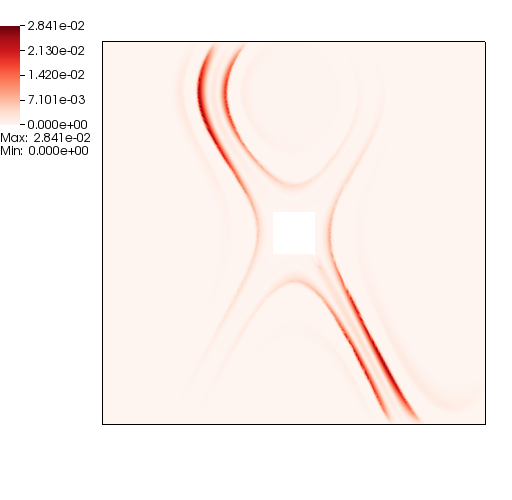}}
\subfloat[$\|\boldsymbol{c}_{\mathrm{CLIP}}-\boldsymbol{c}_{\mathrm{RS}}\|$]{\includegraphics[scale=0.45]{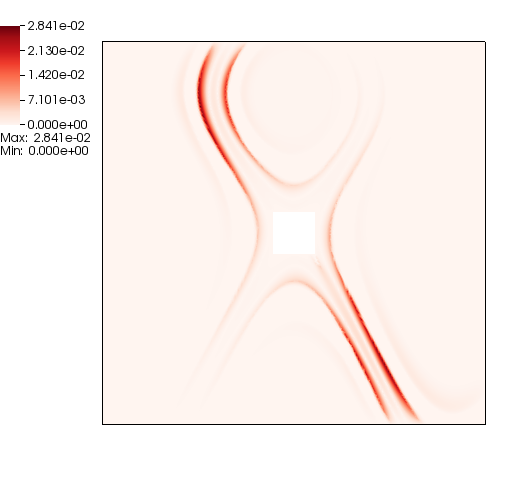}}\\
\subfloat[$\|\boldsymbol{c}_{\mathrm{SS}}-\boldsymbol{c}_{\mathrm{RS}}\|$]{\includegraphics[scale=0.45]{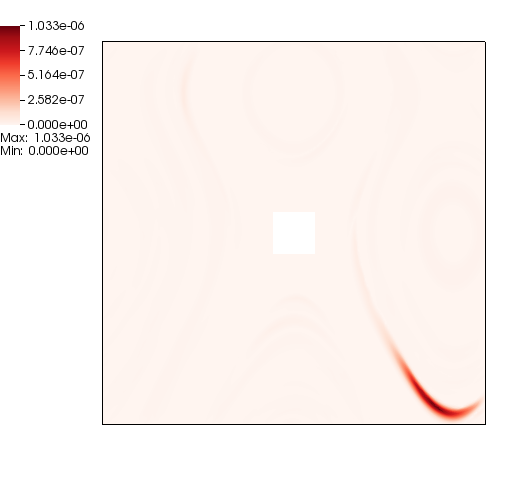}}
\caption{2D advection-diffusion: absolute difference in concentrations between the clipped and non-negative solutions.}
\label{Fig:2D_AD_clipped}
\end{figure}
Next we consider the advection-diffusion problem under the SUPG formulation 
where only VI - SS and VI - RS methods are applicable for enforcing the maximum
principle and the non-negative constraint. Consider a bi-unit 
square: $\Omega := (0,1)\times(0,1)$  with a square hole of dimension 
$\left[\frac{4}{9},\frac{5}{9}\right]\times\left[\frac{4}{9},
\frac{5}{9}\right]$ as shown in Figure \ref{Fig:2D_AD_description}. 
The mesh is discretized into 96,430 unstructured triangular elements 
and 48,663 vertices. Homogeneous boundary conditions are applied on the outside 
boundary, and a Dirichlet boundary value 
$c^{\mathrm{p}}(\mathbf{x}) = 1$ is applied on the interior boundary 
$\Gamma^{\mathrm{hole}}$. The velocity vector field $\mathbf{v}(\mathbf{x})$ 
is characterized by the following:
\begin{subequations}
\begin{align}
v_{x} &= \cos(2 \pi y^2) \\
v_{y} &= \sin(2 \pi x) + \cos(2 \pi x^2) 
\end{align}
\label{Eqn:S5_2D_flow}
\end{subequations}
and the diffusivity tensor $\mathbf{D}(\mathbf{x})$ for this
problem is the dispersion tensor:
\begin{align}
\mathbf{D}(\mathbf{x}) = \left(\alpha_T \|\mathbf{v}\|+D_{M}\right) \mathbf{I} + 
(\alpha_L - \alpha_T) \frac{\mathbf{v} \otimes \mathbf{v}}{\|\mathbf{v}\|} 
\label{Eqn:S5_dispersion_tensor}
\end{align}
where $\alpha_L = 10^{-1}$, $\alpha_T = 10^{-5}$, and $D_M = 10^{-9}$ denote the 
longitudinal dispersivity, transverse dispersivity and molecular diffusivity, 
respectively. Figure \ref{Fig:2D_AD_solutions} depicts the numerical solutions 
under the SUPG, VI - SS, and VI - RS formulations. We see that the SUPG formulation
results in negative concentrations as well as concentrations greater than the maximum prescribed boundary condition whereas the two VI solvers successfully correct these
concentrations. \emph{The absolute difference plots, as seen in Figure 
\ref{Fig:2D_AD_clipped}, indicate that the VI are also 
similar to one another and differ from the clipping procedure.} 
These 2D benchmarks suggest that the QP and VI solvers are accurate alternatives
to the clipping procedure for satisfying the discrete maximum principle and the 
non-negative constraint. Listings \ref{Code:ex1} and \ref{Code:ex2} 
contain the Firedrake project files for solving the GAL and SUPG 
formulations, respectively.
\subsection{3D benchmark}
\begin{figure}[t]
\centering
\subfloat{\includegraphics[scale=0.45]{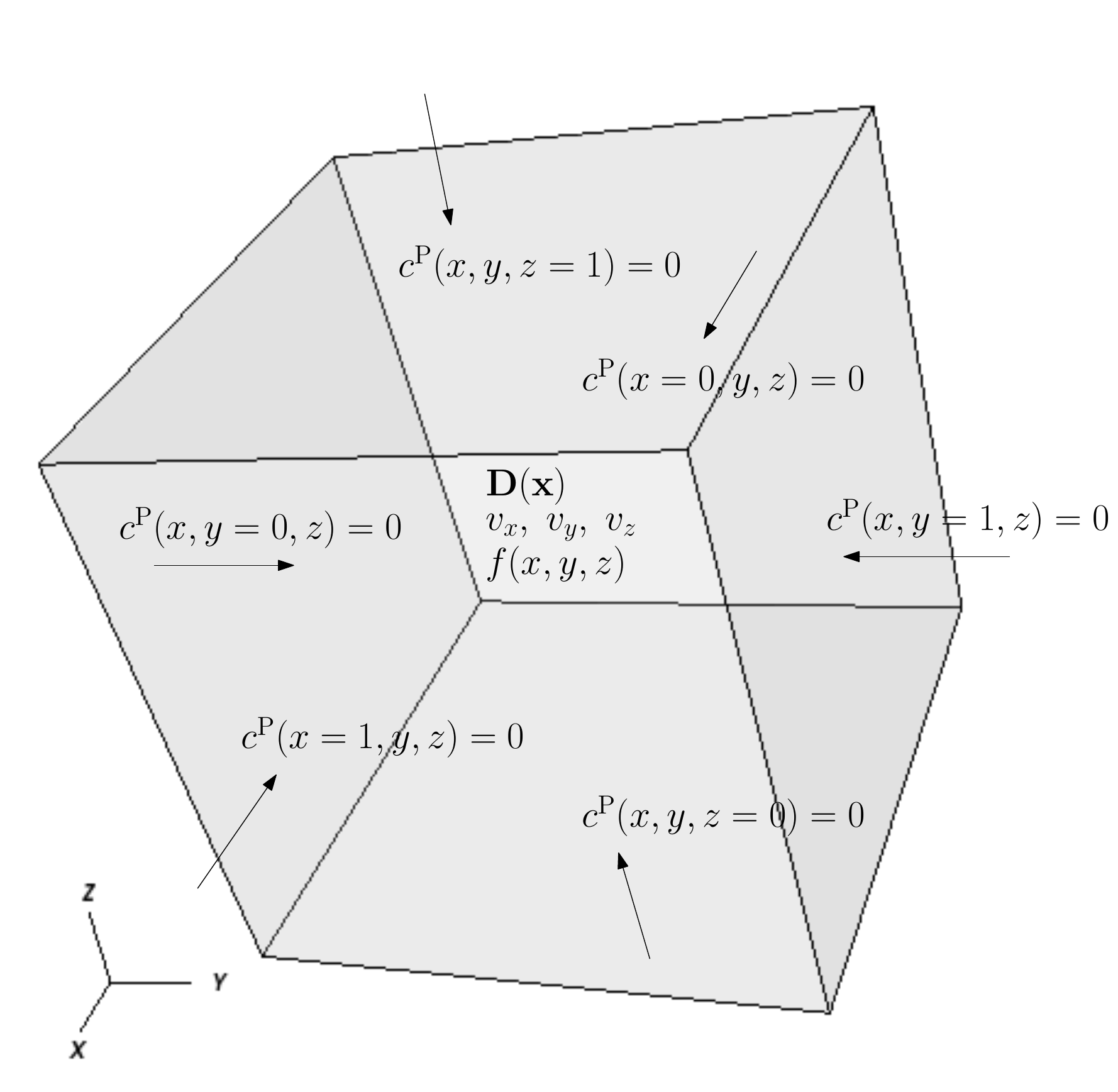}}
\subfloat{\includegraphics[scale=0.45]{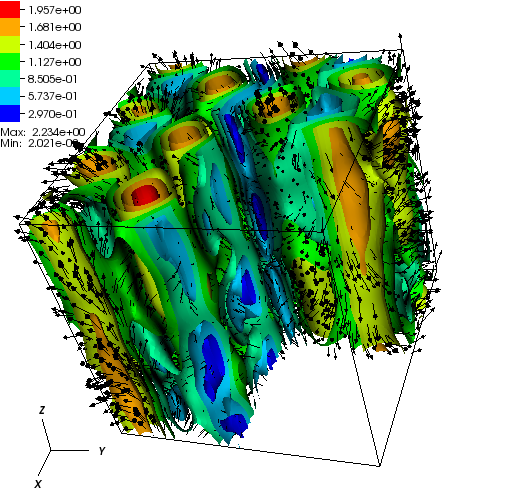}}
\caption{3D benchmarks: Left figure contains a pictorial description of the 
boundary value problem. Right figure contains the corresponding velocity contour 
and vector field for the ABC flow (right) for the steady-state diffusion 
and advection-diffusion problems.}
\label{Fig:3D_description}
\end{figure}
\begin{figure}[t]
\centering
\subfloat[GAL]{\includegraphics[scale=0.45]{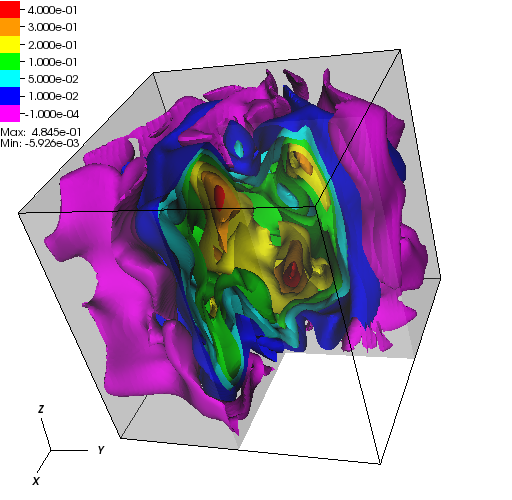}}
\subfloat[GAL with VI - SS]{\includegraphics[scale=0.45]{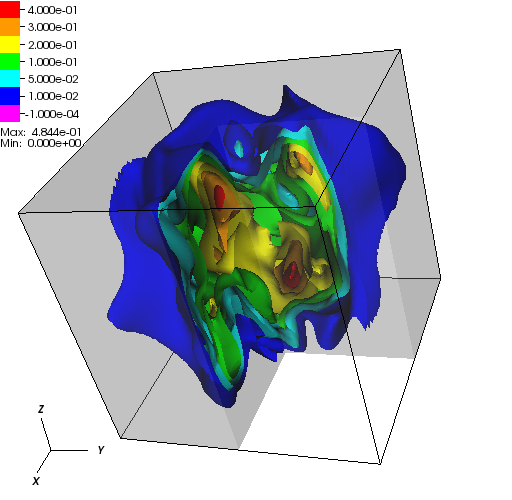}} \\
\subfloat[DG]{\includegraphics[scale=0.45]{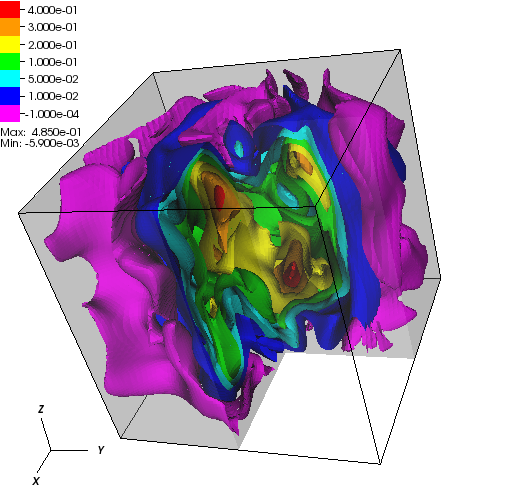}}
\subfloat[DG with VI - SS]{\includegraphics[scale=0.45]{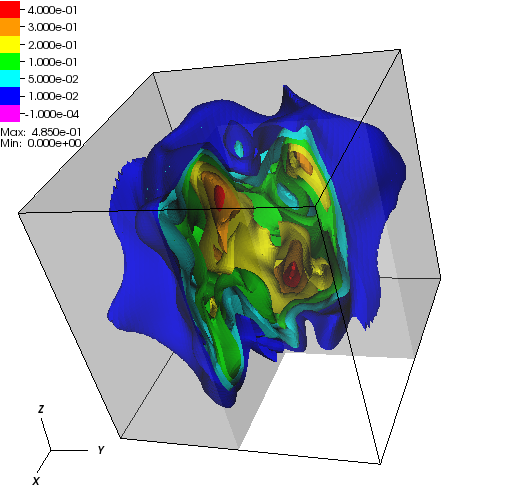}}
\caption{3D diffusion: 3D contours of the concentrations for the GAL and DG formulations with and without and VI - SS for $h$-size = 1/80, where the purple contours represent regions with negative concentrations 
(see online version for color figures.)}
\label{Fig:3D_D_contours}
\end{figure}
We now consider a 3D problem designed to capture two particular 
aspects that may arise in large-scale applications: 1) chaotic advection, 
which is pervasive in many porous media applications \citep{lester2013chaotic}, and 
2) random point sources, which in subsurface remediation problems are the sites 
where potential contaminant leaks occur. Predictive modeling involving such 
important aspects require numerical methodologies that are not accurate but 
fast and scalable in a parallel environment. Herein, our goal is to study the 
computational performance of the various QP and VI solvers under the GAL, 
SUPG, and DG formulations.

Consider a unit cube domain as shown in Figure \ref{Fig:3D_description} 
with chaotic advection flow characterized by the 
Arnold-Beltrami-Childress (ABC) flow \citep{Zhao_SIAMJAM_1993,Dombre_JFM_1986}: 
\begin{subequations}
\begin{align}
v_{x} &= 0.3 \sin(2 \pi z) + \cos(3 \pi y) \\
v_{y} &= 0.65 \sin(2 \pi x) + 0.3 \cos(5 \pi z) \\
v_{z} &= \sin(4 \pi y) + 0.65 \cos(6 \pi y) 
\end{align}
\label{Eqn:S5_abc_flow}
\end{subequations}

For this problem, we shall also let $\mathbf{D}(\mathbf{x})$ denote 
the dispersion tensor as shown in equation 
\eqref{Eqn:S5_dispersion_tensor} where $\alpha_L = 10^{-1}$, $\alpha_T = 10^{-5}$, 
and $D_M = 10^{-9}$. All six faces of the cube have homogeneous boundary conditions,
and the following forcing function consisting of 8 randomly located point sources is
used throughout the domain:
\begin{align}
  f(x,y,z) &= \left\{\begin{array}{ll}
  1 & \mathrm{if}\;(x,y,z) \in[0.4,0.2,0.1]\times[0.5,0.3,0.2] \\
  1 & \mathrm{if}\;(x,y,z) \in[0.8,0.4,0.2]\times[0.9,0.5,0.3] \\
  1 & \mathrm{if}\;(x,y,z) \in[0.5,0.7,0.3]\times[0.6,0.8,0.4] \\
  1 & \mathrm{if}\;(x,y,z) \in[0.3,0.5,0.2]\times[0.4,0.6,0.3] \\
  1 & \mathrm{if}\;(x,y,z) \in[0.5,0.2,0.6]\times[0.6,0.3,0.7] \\
  1 & \mathrm{if}\;(x,y,z) \in[0.6,0.5,0.7]\times[0.7,0.6,0.8] \\
  1 & \mathrm{if}\;(x,y,z) \in[0.4,0.7,0.8]\times[0.5,0.8,0.9] \\
  1 & \mathrm{if}\;(x,y,z) \in[0.1,0.4,0.7]\times[0.2,0.5,0.8] \\
  0 & \mathrm{otherwise}
  \end{array}
  \right.
\label{Eqn:S5_prob2_source}
\end{align}
To understand the parallel and algorithmic scalability of the QP and VI 
solvers, various levels of mesh refinement 
are considered, ranging from 1,331 to 1,030,301 degrees-of-freedom
for the GAL/SUPG formulations and ranging 
from 8,000 to 4,096,000 degrees-of-freedom for the DG formulations.
Up to 16 MPI processes are used to study the weak-scaling and strong-scaling
potential of these solvers. 

\begin{table}[t]
  \centering
  \caption{3D diffusion: minimum and maximum concentrations for various level of mesh refinement under the GAL formulation. \label{Tab:D_violations_ex1}}
  \begin{tabu} to 1.0\textwidth{@{}X[0.3]X[0.6c]X[0.6c]X[c]@{}}
    \hline
    $h$-size & Min. concentration & Max. concentration & \% degrees-of-freedom violated \\
    \hline
    1/10 & -0.0224497 & 0.368322 & 280/1,331 $\rightarrow$ 21.0\% \\
    1/20 & -0.0071611 & 0.339679 & 2,462/9,261 $\rightarrow$ 26.6\% \\
    1/30 & -0.0083804 & 0.481598 & 8,449/29,791 $\rightarrow$ 28.4\% \\
    1/40 & -0.0062918 & 0.378390 & 20,195/68,921 $\rightarrow$ 29.3\% \\
    1/50 & -0.0067679 & 0.477119 & 39,500/132,651 $\rightarrow$ 29.8\% \\
    1/60 & -0.0072030 & 0.518469 & 68,161/226,981 $\rightarrow$ 30.0\% \\
    1/70 & -0.0066007 & 0.498127 & 109,554/357,911 $\rightarrow$ 30.6\% \\
    1/80 & -0.0059264 & 0.484484 & 160,925/531,441 $\rightarrow$ 30.3\% \\
    \hline
  \end{tabu}
\end{table}
\begin{table}[t]
  \centering
  \caption{3D diffusion: minimum and maximum concentrations for various level of mesh refinement under the DG formulation. \label{Tab:D_violations_ex2}}
  \begin{tabu} to 1.0\textwidth{@{}X[0.3]X[0.6c]X[0.6c]X[c]@{}}
    \hline
    $h$-size & Min. concentration & Max. concentration & \% degrees-of-freedom violated \\
    \hline
    1/10 & -0.0226040 & 0.372831 & 3,704/8,000 $\rightarrow$ 46.3\% \\
    1/20 & -0.0071913 & 0.341955 & 27,496/64,000 $\rightarrow$ 43.0\% \\
    1/30 & -0.0082811 & 0.483264 & 91,176/216,000 $\rightarrow$ 42.2\% \\
    1/40 & -0.0062341 & 0.379389 & 213,000/512,000 $\rightarrow$ 41.6\% \\
    1/50 & -0.0067168 & 0.478146 & 410,976/1,000,000 $\rightarrow$ 41.1\% \\
    1/60 & -0.0071682 & 0.519338 & 702,504/1,728,000 $\rightarrow$ 40.7\% \\
    1/70 & -0.0065727 & 0.498775 & 1,114,856/2,744,000 $\rightarrow$ 40.6\% \\
    1/80 & -0.0058998 & 0.485012 & 1,624,496/4,096,000 $\rightarrow$ 39.7\% \\
    \hline
  \end{tabu}
\end{table}
Figure \ref{Fig:3D_D_contours} depicts the GAL and DG solutions for the diffusion
equation with and without VI - SS. It can be seen from the figures that negative
concentrations are present regardless which finite element formulation is used.
Tables \ref{Tab:D_violations_ex1} and \ref{Tab:D_violations_ex2} indicate 
that negative concentrations arise for the GAL and DG formulations, respectively, 
even as $h$-size is refined. It is interesting to note that the DG formulation not 
only has more degrees-of-freedom but has more regions with negative concentrations 
than its GAL counterpart. Using the initial guess solver from \textsf{Step 2} of
the proposed framework in \ref{SS4:proposed_framework} as a baseline for 
comparison, Tables \ref{Tab:D_solutions_ex1} and 
\ref{Tab:D_solutions_ex2} demonstrate how the wall-clock time and 
number of KSP/VI/QP solver iterations vary with $h$-refinement under a single
MPI process. It should be noted that the timings for the
QP and VI solvers consider both the assembly of the data structures 
as well as the actual solver. The heterogeneous nature of the problem 
causes the number of solvers iterations to increase with 
problem size, but the iteration counts begin to stabilize as the problem
gets bigger. The VI - RS method outperforms VI - SS in both wall-clock 
time and VI iterations but has similar performance to QP - TRON.
\begin{table}[t]
  \centering
  \caption{3D diffusion: wall-clock time and number of solver iterations 
  (KSP, VI, or QP) for various levels of mesh refinement under the GAL formulation.
   \label{Tab:D_solutions_ex1}}
  \begin{tabu} to 1.0\textwidth{@{}X|X[c]X[c]|X[c]X[c]|X[c]X[c]|X[c]X[c]@{}}
    \hline
    \multirow{2}{*}{$h$-size} & \multicolumn{2}{c|}{GAL} & \multicolumn{2}{c|}{VI - SS} & \multicolumn{2}{c|}{VI - RS} & \multicolumn{2}{c}{QP - TRON}\\
   & time (s) & iters & time (s) & iters & time (s) & iters & time (s) & iters \\
    \hline
    1/10 & 0.003 & 9 & 0.027 & 5 & 0.008 & 2 & 0.007 & 2\\
    1/20 & 0.036 & 15 & 0.477 & 12 & 0.147 & 5 & 0.135 & 5\\
    1/30 & 0.165 & 20 & 2.624 & 18 & 0.765 & 7 & 0.650 & 6\\
    1/40 & 0.525 & 24 & 7.576 & 20 & 2.246 & 8 & 1.758 & 6\\
    1/50 & 1.293 & 28 & 21.49 & 27 & 5.381 & 9 & 5.330 & 9\\
    1/60 & 2.556 & 31 & 43.72 & 30 & 12.01 & 11 & 12.20 & 11\\
    1/70 & 4.747 & 35 & 76.21 & 31 & 18.27 & 10 & 17.81 & 9\\
    1/80 & 7.962 & 39 & 140.7 & 37 & 36.40 & 13 & 38.06 & 13\\
    \hline
  \end{tabu}
\end{table}
\begin{table}[t]
  \centering
  \caption{3D diffusion: wall-clock time and number of solver iterations 
  (KSP, VI, or QP) for various levels of mesh refinement under the DG formulation.
   \label{Tab:D_solutions_ex2}}
  \begin{tabu} to 1.0\textwidth{@{}X|X[c]X[c]|X[c]X[c]|X[c]X[c]|X[c]X[c]@{}}
    \hline
    \multirow{2}{*}{$h$-size} & \multicolumn{2}{c|}{DG} & \multicolumn{2}{c|}{VI - SS} & \multicolumn{2}{c|}{VI - RS} & \multicolumn{2}{c}{QP - TRON}\\
   & time (s) & iters & time (s) & iters & time (s) & iters & time (s) & iters \\
    \hline
    1/10 & 0.030 & 10 & 0.748 & 12 & 0.221 & 5 & 0.186 & 5\\
    1/20 & 0.446 & 14 & 1.410 & 20 & 4.528 & 8 & 3.715 & 7\\
    1/30 & 2.148 & 18 & 73.52 & 27 & 23.85 & 10 & 22.64 & 10\\
    1/40 & 6.278 & 21 & 251.5 & 34 & 69.33 & 11 & 70.77 & 11\\
    1/50 & 14.29 & 24 & 623.6 & 39 & 171.7 & 13 & 170.8 & 12\\
    1/60 & 28.25 & 27 & 1290 & 45 & 360.1 & 15 & 388.6 & 15\\
    1/70 & 51.95 & 31 & 2560 & 51 & 620.2 & 16 & 639.0 & 15\\
    1/80 & 85.53 & 34 & 5049 & 54 & 1107 & 19 & 1291 & 17\\
    \hline
  \end{tabu}
\end{table}
\begin{figure}[t]
\centering
\subfloat[GAL - solve time]{\includegraphics[scale=0.5]{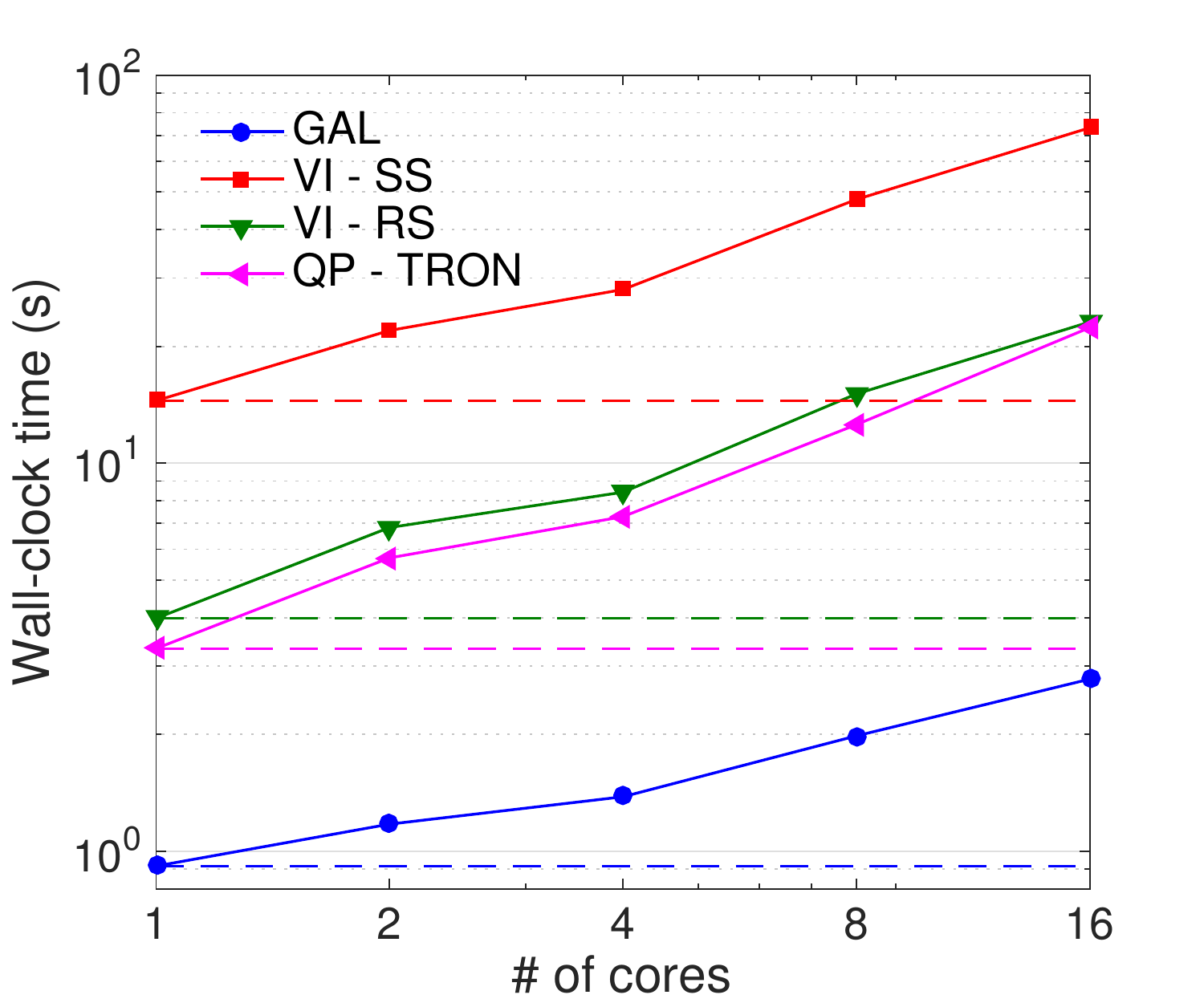}}
\subfloat[GAL - parallel efficiency]{\includegraphics[scale=0.5]{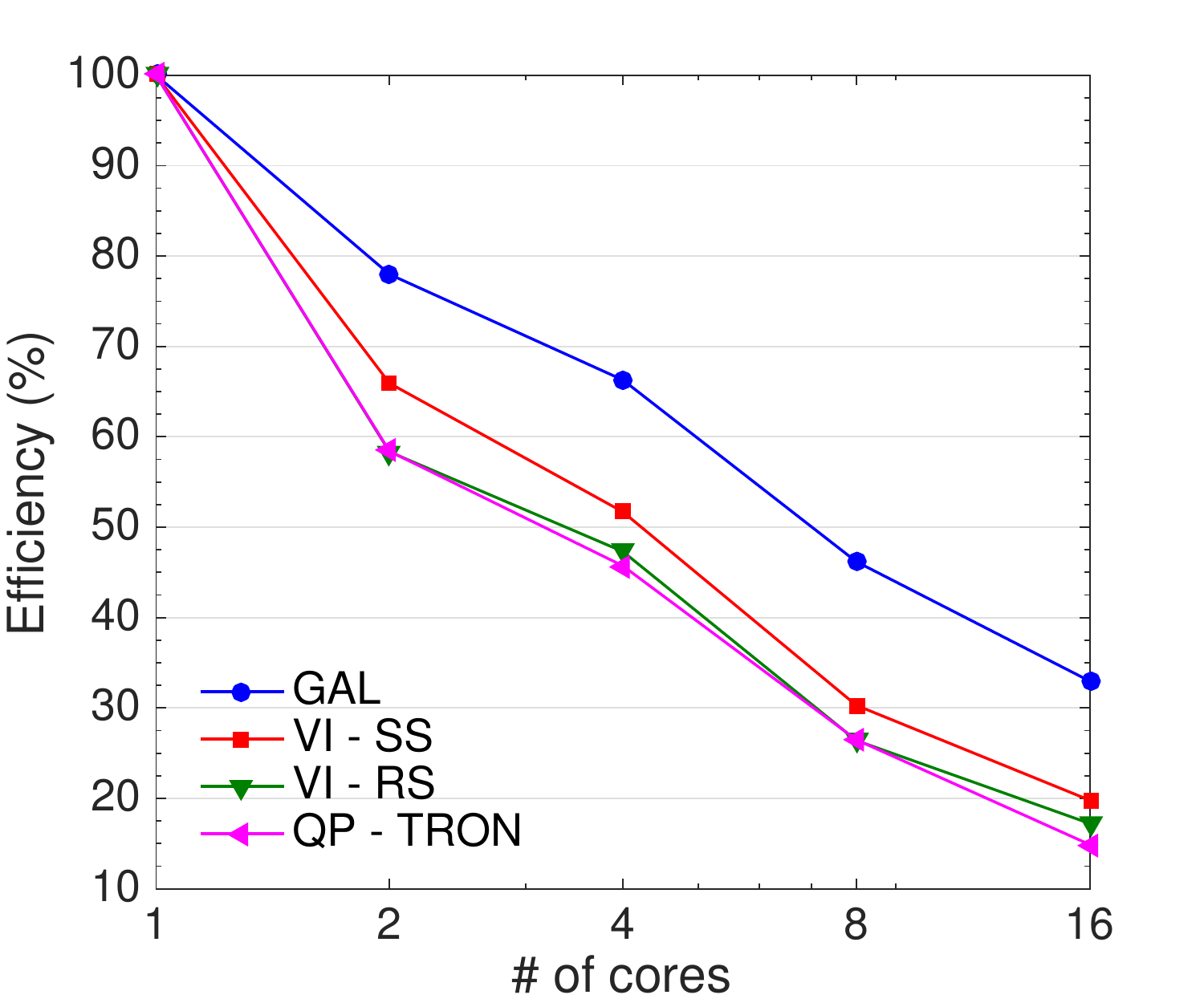}} \\
\subfloat[DG - solve time]{\includegraphics[scale=0.5]{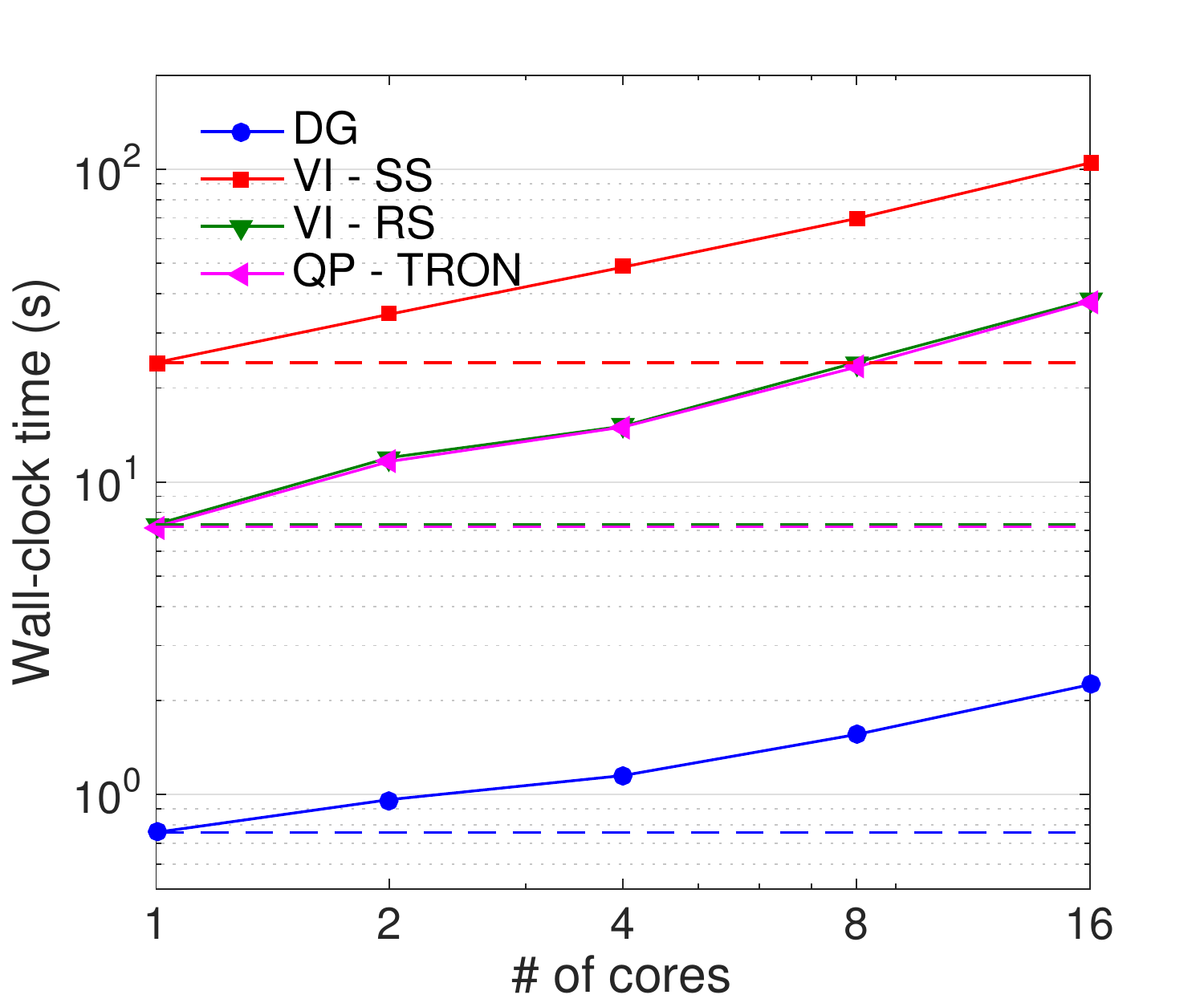}}
\subfloat[DG - parallel efficiency]{\includegraphics[scale=0.5]{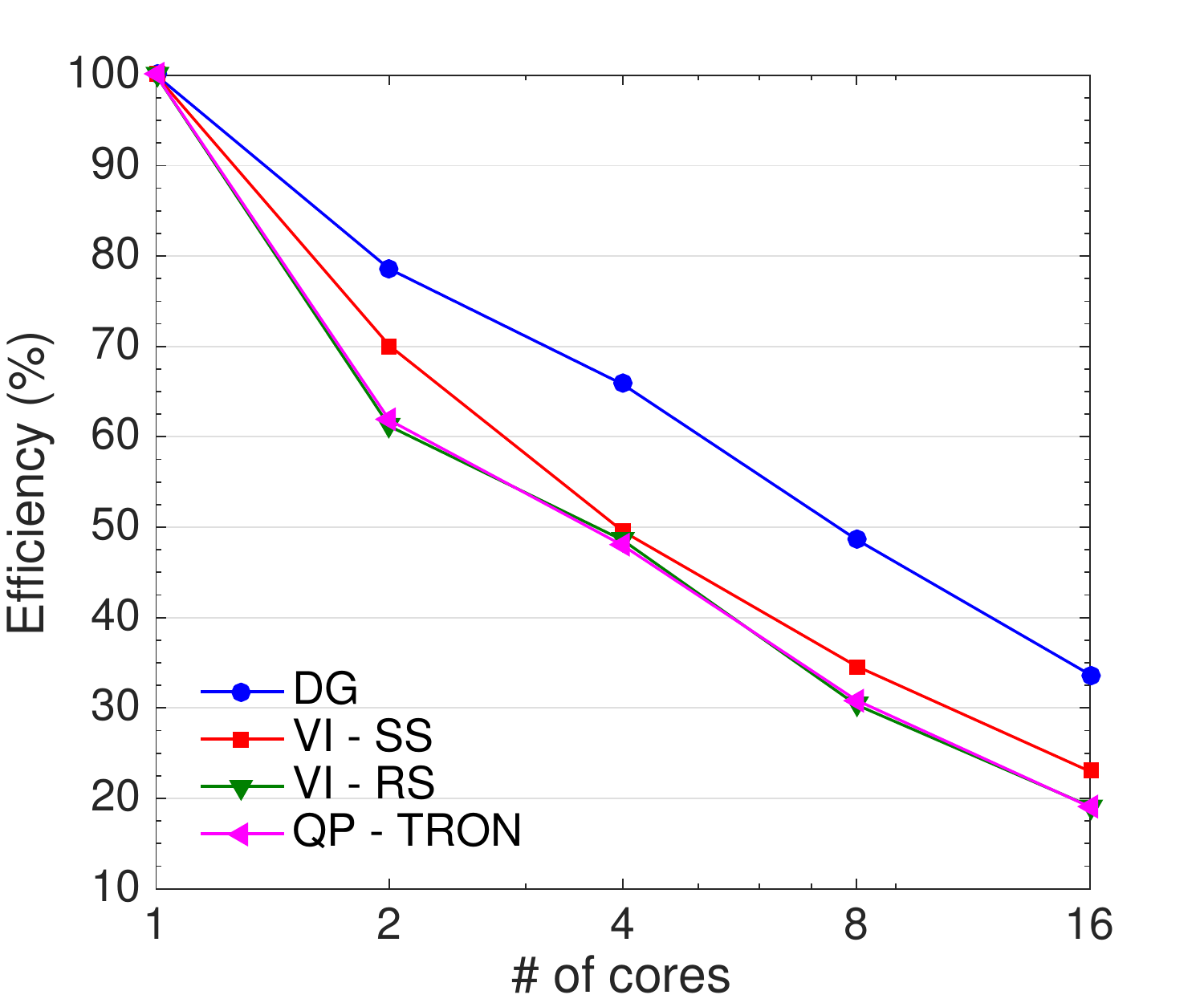}}
\caption{3D diffusion: weak-scaling plots with approximately 100k degrees-of-freedom per core and the corresponding parallel efficiencies.}
\label{Fig:3D_D_weak}
\end{figure}
\begin{figure}[t]
\centering
\subfloat[GAL - solver time]{\includegraphics[scale=0.5]{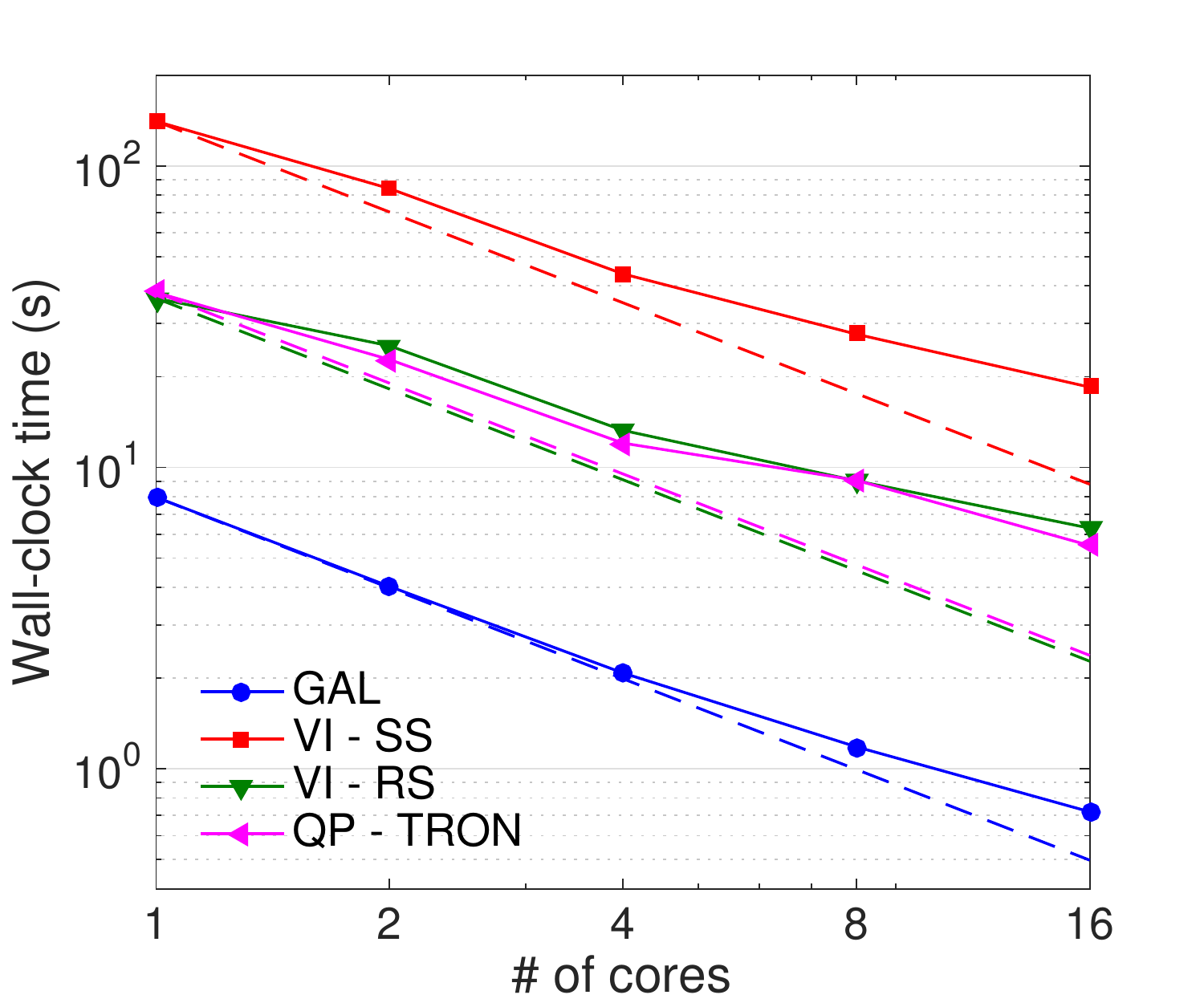}}
\subfloat[GAL - parallel efficiency]{\includegraphics[scale=0.5]{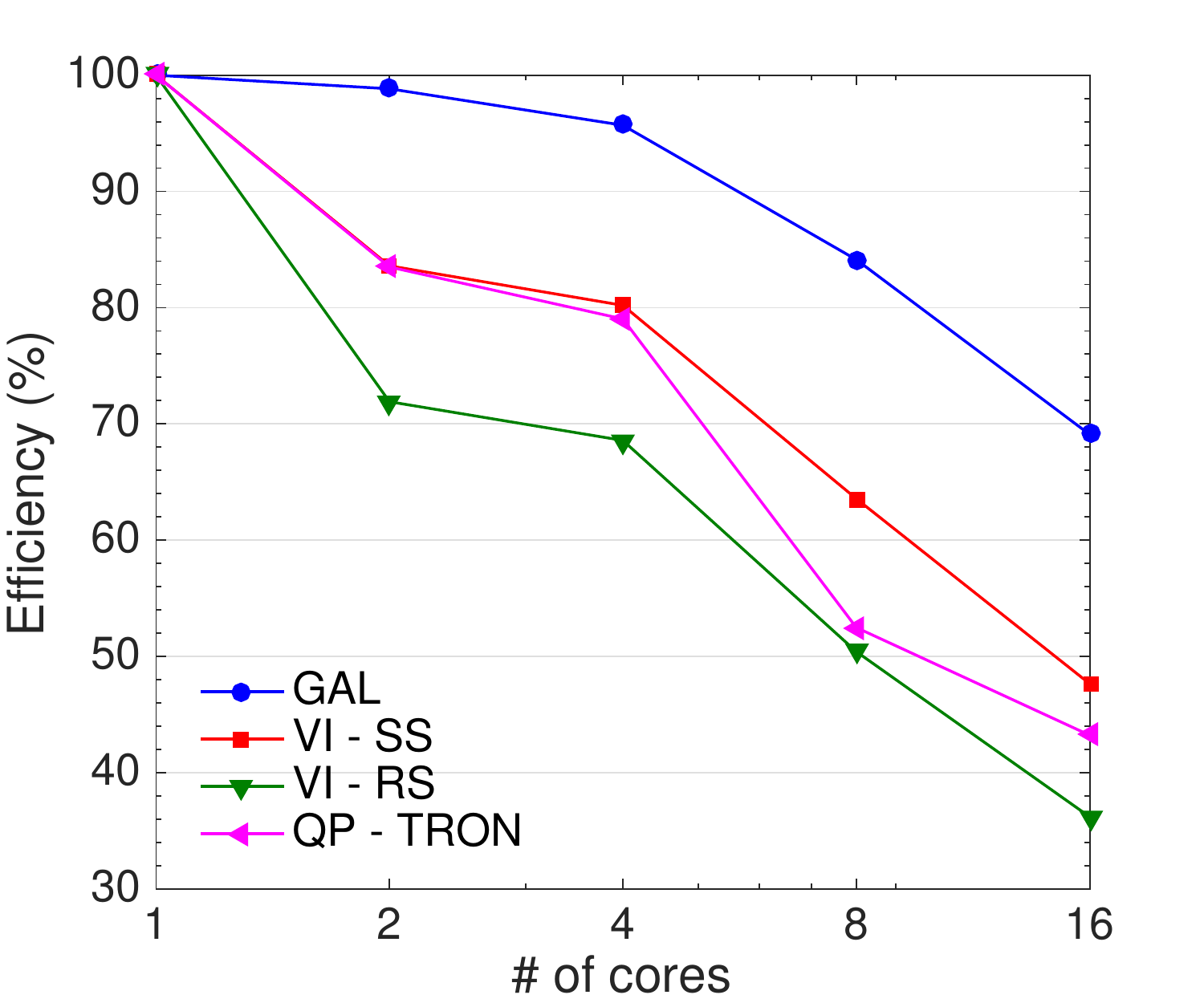}} \\
\subfloat[DG - solver time]{\includegraphics[scale=0.5]{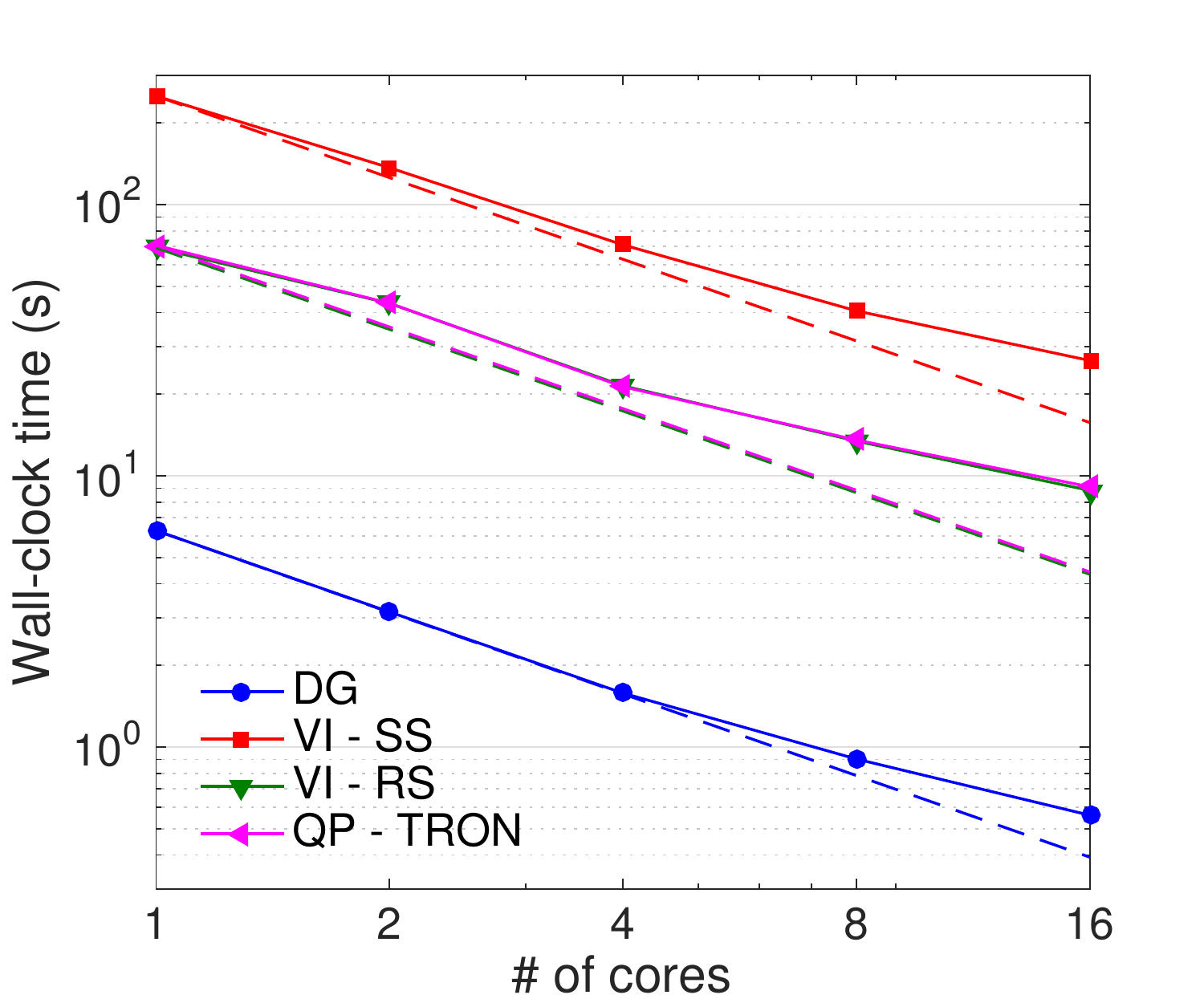}}
\subfloat[DG - parallel efficiency]{\includegraphics[scale=0.5]{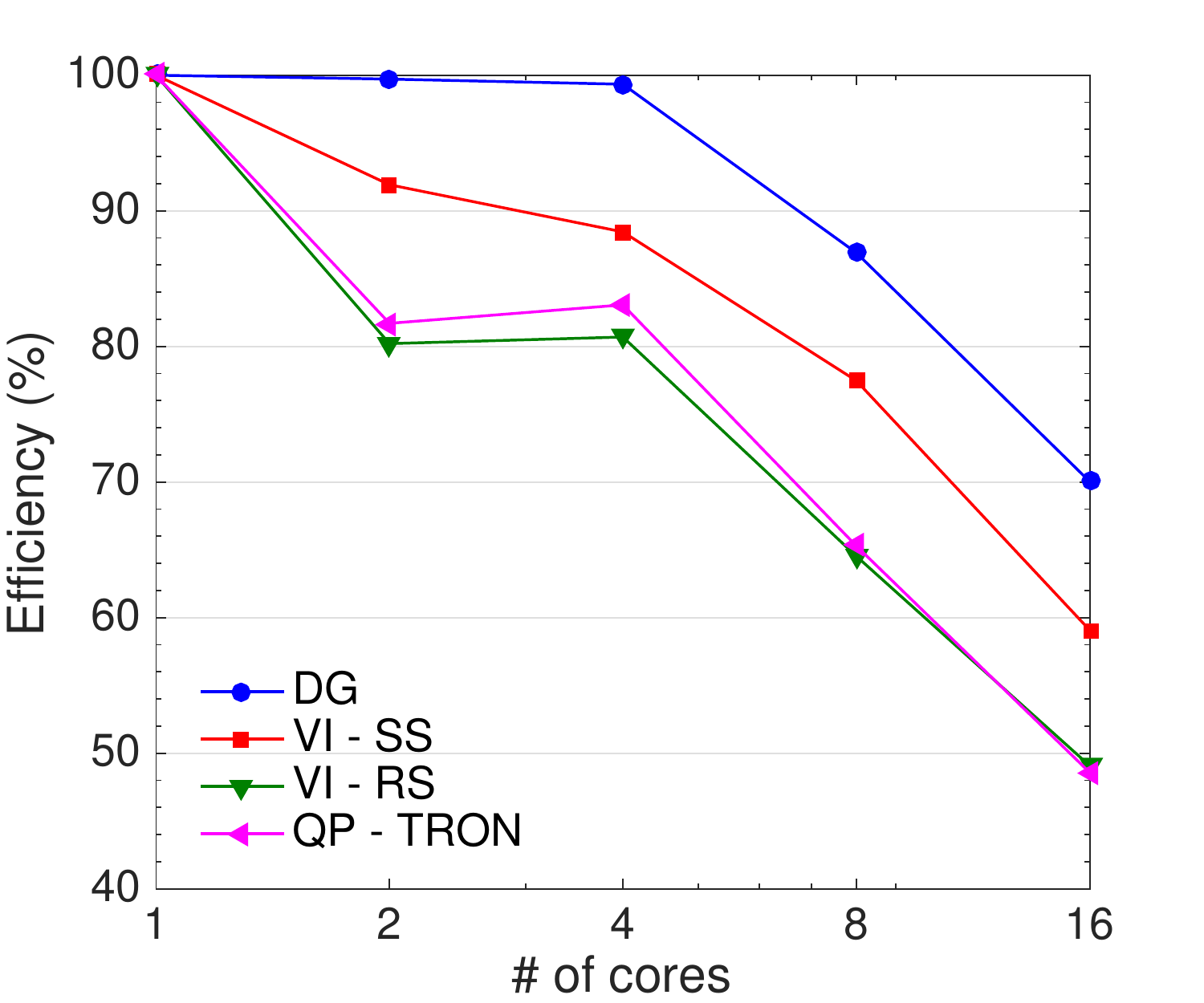}}
\caption{3D diffusion: strong-scaling plots for approximately 500k degrees of freedom ($h$-size = 1/80 and 1/40 for GAL and DG respectively) and the corresponding parallel efficiencies.}
\label{Fig:3D_D_strong}
\end{figure}

Next we perform weak-scaling studies to investigate how increasing both problem size 
and number of MPI processes affects the performance of the VI and QP solvers. 
Each MPI process will handle approximately 100k degrees-of-freedom so the $h$-sizes 
for the GAL case are 1/46, 1/58, 1/73, 1/92, and 1/116 for 1, 2, 4, 8, and 16 
processes respectively whereas the $h$-sizes for the DG case are 1/23, 1/29, 1/37, 1/46, and 1/58 for 1, 2, 4, 8, and 16 processes respectively. Figure \ref{Fig:3D_D_weak} 
contains the scaling plots as well as the parallel efficiencies in the weak sense under
the GAL and DG formulations. Generally speaking, the 
non-negative methodologies do not scale as well in the weak sense for two reasons: 1) 
the wall-clock time and solver iteration for the QP and VI methodologies are not 
linearly proportional to problem size (increasing solver iterations with $h$-size 
as seen from Tables \ref{Tab:D_solutions_ex1} and \ref{Tab:D_solutions_ex2}), 
and 2) the lower KSP relative tolerance for the
gradient descent computations make the overall solver more sensitive to the 
memory-bandwidth, meaning that speedup is reduced as the compute nodes become 
populated with more MPI processes (see Sections 4 and 5 of \citep{Chang_JOMP_2016} 
for a more thorough discussion).

However, the weak-scaling plots alone make it difficult to distinguish 
whether parallel performance deteriorates due to communication overhead 
or suboptimal algorithmic convergence. To better understand why parallel performance 
degrades as the number of MPI processes increases, we conduct strong-scaling 
studies by setting the $h$-size to 1/80 and 1/40 for
the GAL and DG formulations respectively (roughly 500k degrees-of-freedom) 
and study how increasing the number of MPI processes (hence communication overhead) 
affects the parallel performance. Figure \ref{Fig:3D_D_strong} contains the 
strong-scaling plots, and we see that the QP and VI solvers still do not scale as well. 
Regardless of the finite element formulation used, the QP and VI - RS methods 
have roughly the same strong-scaling performance whereas the VI - SS method 
has slightly better strong-scaling.  
\begin{figure}[t]
\centering
\subfloat[SUPG]{\includegraphics[scale=0.45]{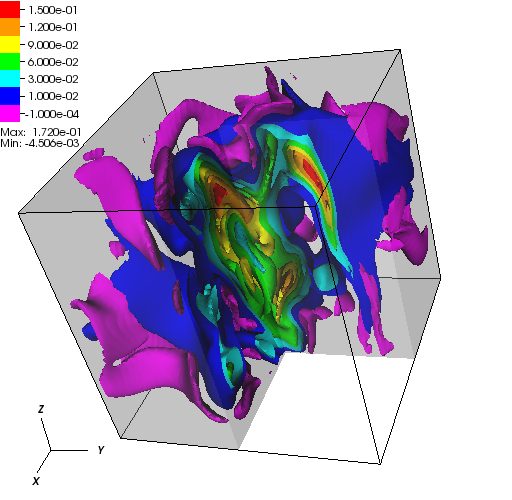}}
\subfloat[SUPG with VI - SS]{\includegraphics[scale=0.45]{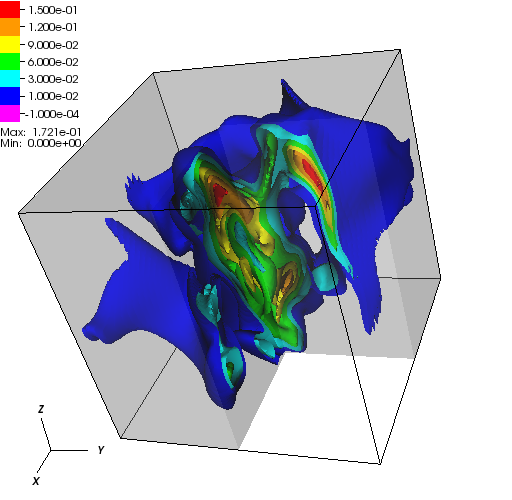}} \\
\subfloat[DG]{\includegraphics[scale=0.45]{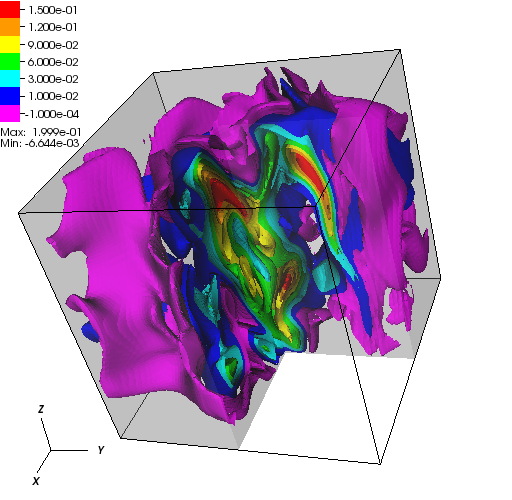}}
\subfloat[DG with VI - SS]{\includegraphics[scale=0.45]{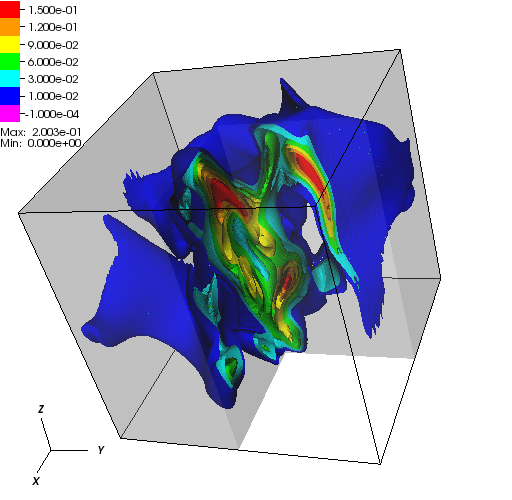}}
\caption{3D advection-diffusion: 3D contours of the concentrations for the SUPG and DG formulations with and without and VI - SS for $h$-size = 1/80, where the purple contours represent regions with negative concentrations (see online version for color figures.)}
\label{Fig:3D_AD_contours}
\end{figure}
\begin{table}[t]
  \centering
  \caption{3D advection-diffusion: minimum and maximum concentrations for various 
  level of mesh refinement under the SUPG formulation. \label{Tab:AD_violations_ex1}}
  \begin{tabu} to 1.0\textwidth{@{}X[0.3]X[0.6c]X[0.6c]X[c]@{}}
    \hline
    $h$-size & Min. concentration & Max. concentration & \% degrees-of-freedom violated \\
    \hline
    1/10 & -0.0135676 & 0.187489 & 212/1,331 $\rightarrow$ 15.9\% \\
    1/20 & -0.0068733 & 0.180922 & 2,323/9,261 $\rightarrow$ 25.1\% \\
    1/30 & -0.0091657 & 0.210942 & 7,964/29,791 $\rightarrow$ 26.7\% \\
    1/40 & -0.0055686 & 0.171690 & 18,235/68,921 $\rightarrow$ 26.5\% \\
    1/50 & -0.0064795 & 0.185440 & 35,221/132,651 $\rightarrow$ 26.6\% \\
    1/60 & -0.0063168 & 0.189047 & 61,171/226,981 $\rightarrow$ 26.9\% \\
    1/70 & -0.0053682 & 0.179675 & 99,668/357,911 $\rightarrow$ 27.8\% \\
    1/80 & -0.0045065 & 0.172049 & 147,462/531,441 $\rightarrow$ 27.7\%  \\
    \hline
   \end{tabu}
\end{table}
\begin{table}[t]
  \centering
  \caption{3D advection-diffusion: minimum and maximum concentrations for various level of mesh refinement under the DG formulation. \label{Tab:AD_violations_ex2}}
  \begin{tabu} to 1.0\textwidth{@{}X[0.3]X[0.6c]X[0.6c]X[c]@{}}
    \hline
    $h$-size & Min. concentration & Max. concentration & \% degrees-of-freedom violated \\
    \hline
    1/10 & -0.0151514 & 0.259127 & 4,976/8,000 $\rightarrow$ 62.2\% \\
    1/20 & -0.0162537 & 0.211295 & 37,464/64,000 $\rightarrow$ 58.5\% \\
    1/30 & -0.0137824 & 0.237722 & 120,296/216,000 $\rightarrow$ 55.7\% \\
    1/40 & -0.0067079 & 0.186956 & 276,832/512,000 $\rightarrow$ 54.1\% \\
    1/50 & -0.0057574 & 0.203852 & 526,080/1,000,000 $\rightarrow$ 52.6\% \\
    1/60 & -0.0065627 & 0.203093 & 891,768/1,728,000 $\rightarrow$ 51.6\% \\
    1/70 & -0.0069389 & 0.193418 & 1,410,208/2,744,000 $\rightarrow$ 51.4\% \\
    1/80 & -0.0066445 & 0.199912 & 2,069,752/4,096,000 $\rightarrow$ 50.5\% \\
    \hline
  \end{tabu}
\end{table}

For the advection-diffusion equation, the same problem is considering but 
advection due to the ABC flow is now taken into account. A Firedrake project 
implementation of the DG formulation can be found in Listing \ref{Code:ex3}.
Like the diffusion 
equation, the advection-diffusion equation also exhibits negative concentrations 
as seen from Figure \ref{Fig:3D_AD_contours}. Table 
\ref{Tab:AD_violations_ex1} depicts violations
under the SUPG formulation to be no greater than 30\% whereas the DG formulation 
exhibits huge violations as seen from Table \ref{Tab:AD_violations_ex2}. 
Moreover, the single MPI process metrics shown in Tables \ref{Tab:AD_solutions_ex1} 
and \ref{Tab:AD_solutions_ex2} clearly indicate that the advection-diffusion 
equations are generally more expensive to solve than its diffusion counterpart. 
These metrics tell us that the VI iterations also begin to stabilize as the 
problem size increases and that the VI - RS method is faster than 
the VI - SS method for the advection-diffusion equation. 

The weak-scaling plots, as seen from Figure \ref{Fig:3D_AD_weak}, indicate
that the VI solvers are identical in performance, but the 
strong-scaling plots from Figure \ref{Fig:3D_AD_strong} suggest that 
VI - SS has slightly better scaling than VI - RS. These steady-state 
numerical experiments suggest that the VI - RS is the preferred methodology for 
solving large-scale advection-diffusion equations. However, if advection becomes
negligible thus reducing the system to a symmetric diffusion problem,
one could use either QP - TRON or VI - RS as these solvers are equal in 
both parallel and algorithmic performance.
\begin{remark}
These parallel performance studies do not indicate
how truly efficient the PETSc and TAO implementations 
of the QP and VI solvers are in the context of high performance 
computing. It should be noted that a serially efficient
algorithm will likely have poor parallel efficiency due to 
dominating effects from communication overhead and memory 
latencies. An approximation of the hardware performance
for these solvers can be made through the perfect cache 
model described in Section 4 of \citep{Chang_JOMP_2016}. 
However, a more thorough and detailed performance model 
is needed to fully understand the performance of important 
memory-bandwidth limited algorithms such as the ones 
studied in this paper and is part of our future work.
\end{remark}
\begin{table}[t]
  \centering
  \caption{3D advection-diffusion: wall-clock time and number of linear (KSP) or nonlinear (VI) solve iterations for various levels of mesh refinement under the SUPG 
  formulation. \label{Tab:AD_solutions_ex1}}
  \begin{tabu} to 1.0\textwidth{@{}X|X[c]X[c]|X[c]X[c]|X[c]X[c]@{}}
    \hline
    \multirow{2}{*}{$h$-size} & \multicolumn{2}{c|}{SUPG} & \multicolumn{2}{c|}{VI - SS} & \multicolumn{2}{c}{VI - RS} \\
   & time (s) & iters & time (s) & iters & time (s) & iters \\
    \hline
    1/10 & 0.003 & 9 & 0.028 & 5 & 0.009 & 2 \\
    1/20 & 0.045 & 16 & 0.504 & 11 & 0.146 & 4 \\
    1/30 & 0.221 & 22 & 2.525 & 14 & 0.710 & 5 \\
    1/40 & 0.768 & 29 & 10.21 & 20 & 2.592 & 6 \\
    1/50 & 2.019 & 35 & 27.27 & 23 & 8.842 & 10 \\
    1/60 & 4.530 & 43 & 65.61 & 30 & 19.20 & 10 \\
    1/70 & 8.178 & 47 & 117.3 & 28 & 34.54 & 11 \\
    1/80 & 14.70 & 55 & 229.4 & 34 & 66.07 & 12 \\
    \hline
  \end{tabu}
\end{table}
\begin{table}[t]
  \centering
  \caption{3D advection-diffusion: wall-clock time and number of linear (KSP) or nonlinear (VI) solve iterations for various levels of mesh refinement under the DG 
  formulation. \label{Tab:AD_solutions_ex2}}
  \begin{tabu} to 1.0\textwidth{@{}X|X[c]X[c]|X[c]X[c]|X[c]X[c]@{}}
    \hline
    \multirow{2}{*}{$h$-size} & \multicolumn{2}{c|}{DG} & \multicolumn{2}{c|}{VI - SS} & \multicolumn{2}{c}{VI - RS} \\
   & time (s) & iters & time (s) & iters & time (s) & iters \\
    \hline
    1/10 & 0.031 & 10 & 0.807 & 13 & 0.202 & 5 \\
    1/20 & 0.459 & 14 & 14.46 & 21 & 3.890 & 8 \\
    1/30 & 2.352 & 19 & 73.84 & 27 & 19.56 & 10 \\
    1/40 & 6.819 & 22 & 251.3 & 35 & 56.94 & 11 \\
    1/50 & 15.49 & 25 & 629.1 & 40 & 137.1 & 13 \\
    1/60 & 30.65 & 28 & 1387 & 47 & 292.3 & 15 \\
    1/70 & 57.56 & 32 & 2267 & 51 & 570.0 & 18 \\
    1/80 & 97.15 & 36 & 7105 & 60 & 917.2 & 18 \\
    \hline
  \end{tabu}
\end{table}
\begin{figure}[t]
\centering
\subfloat[SUPG - solve time]{\includegraphics[scale=0.5]{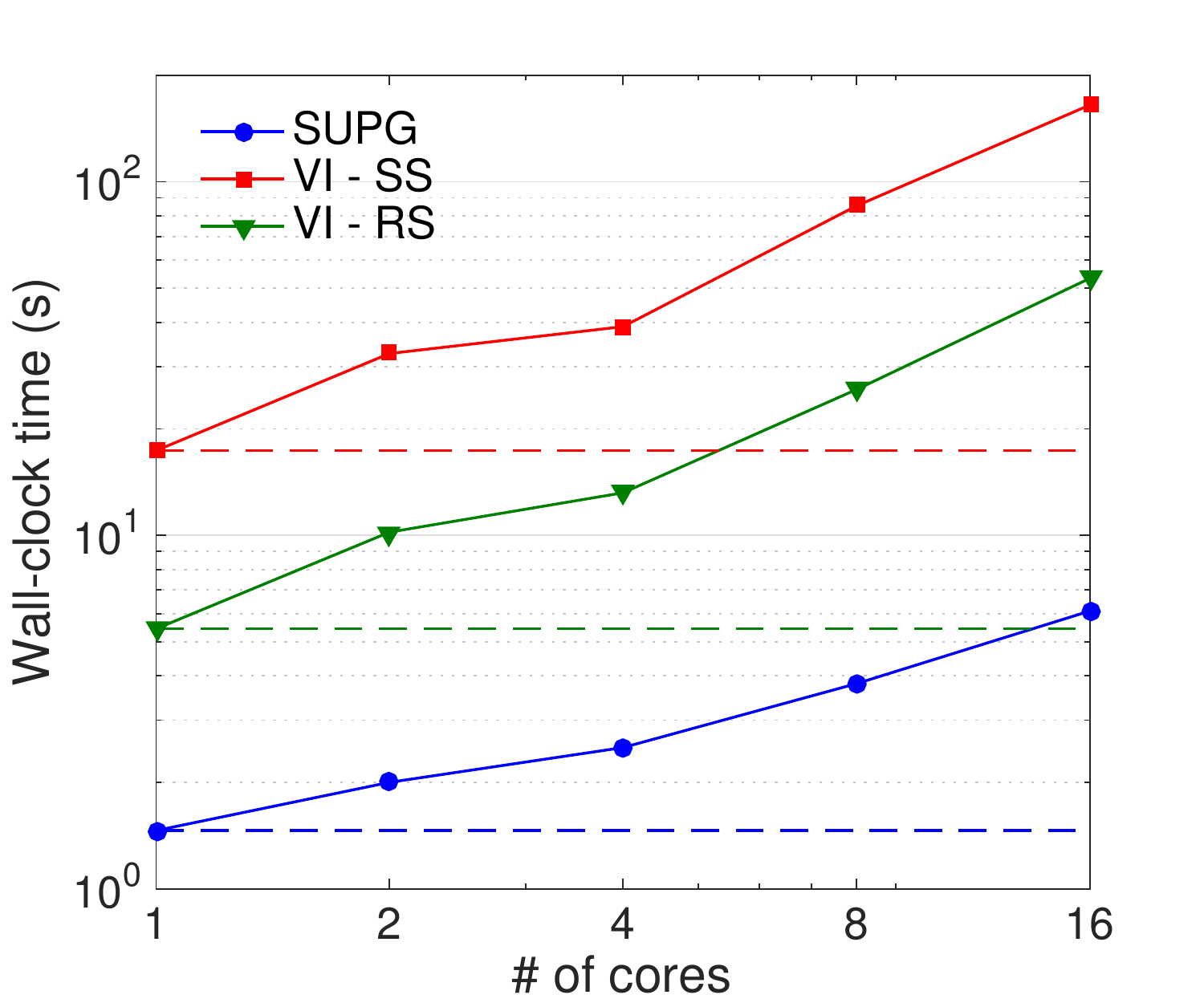}}
\subfloat[SUPG - parallel efficiency]{\includegraphics[scale=0.5]{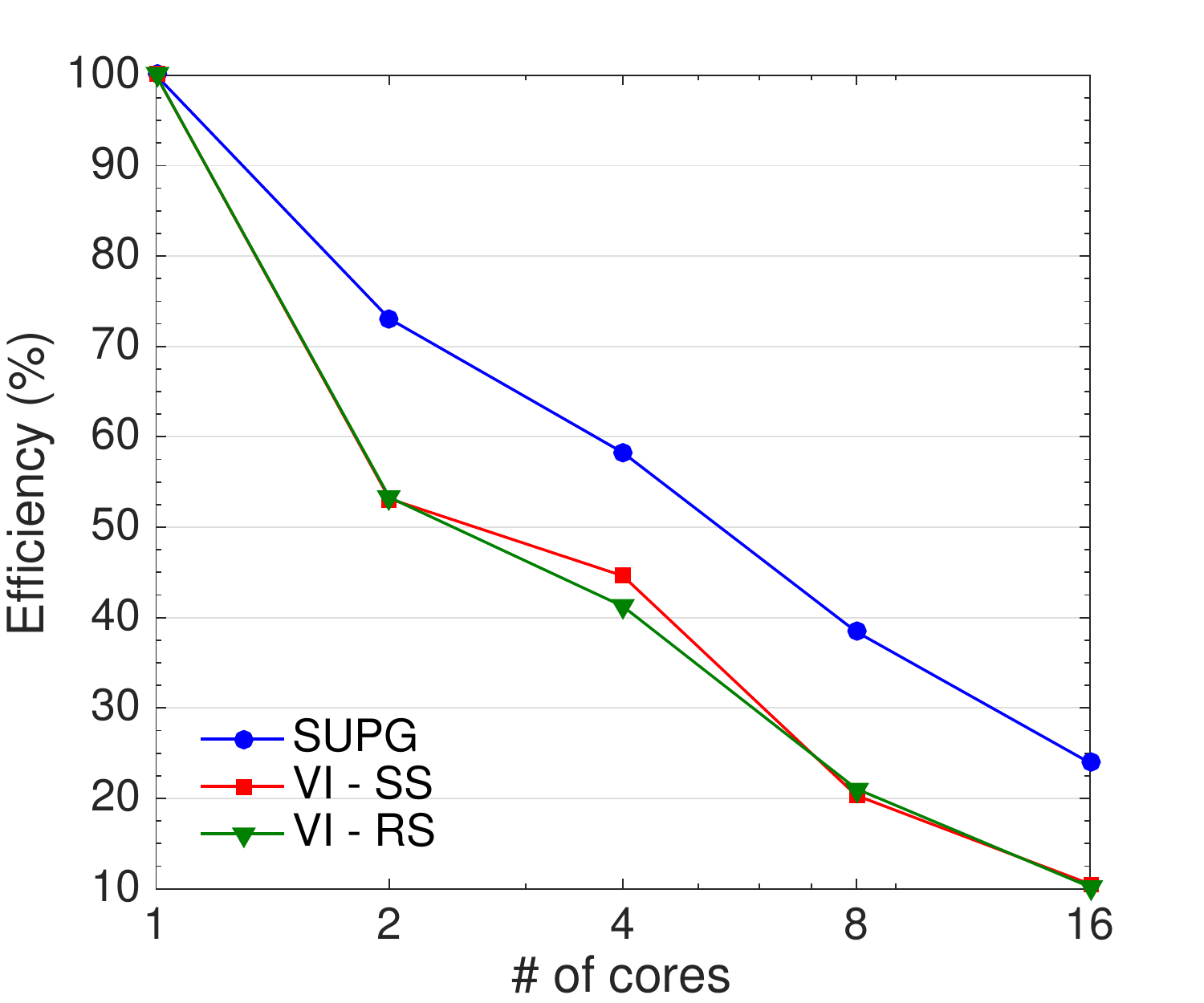}} \\
\subfloat[DG - solve time]{\includegraphics[scale=0.5]{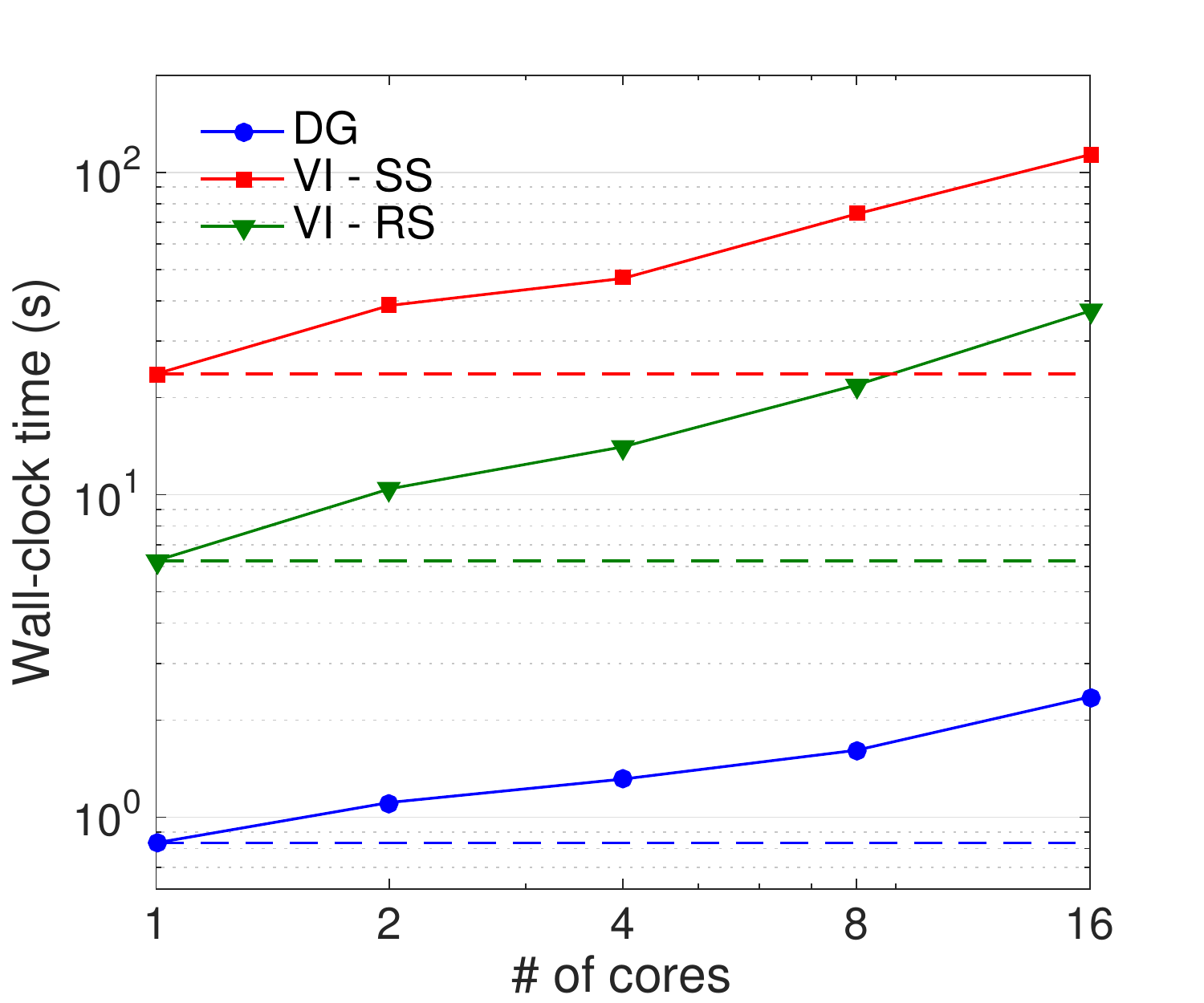}}
\subfloat[DG - parallel efficiency]{\includegraphics[scale=0.5]{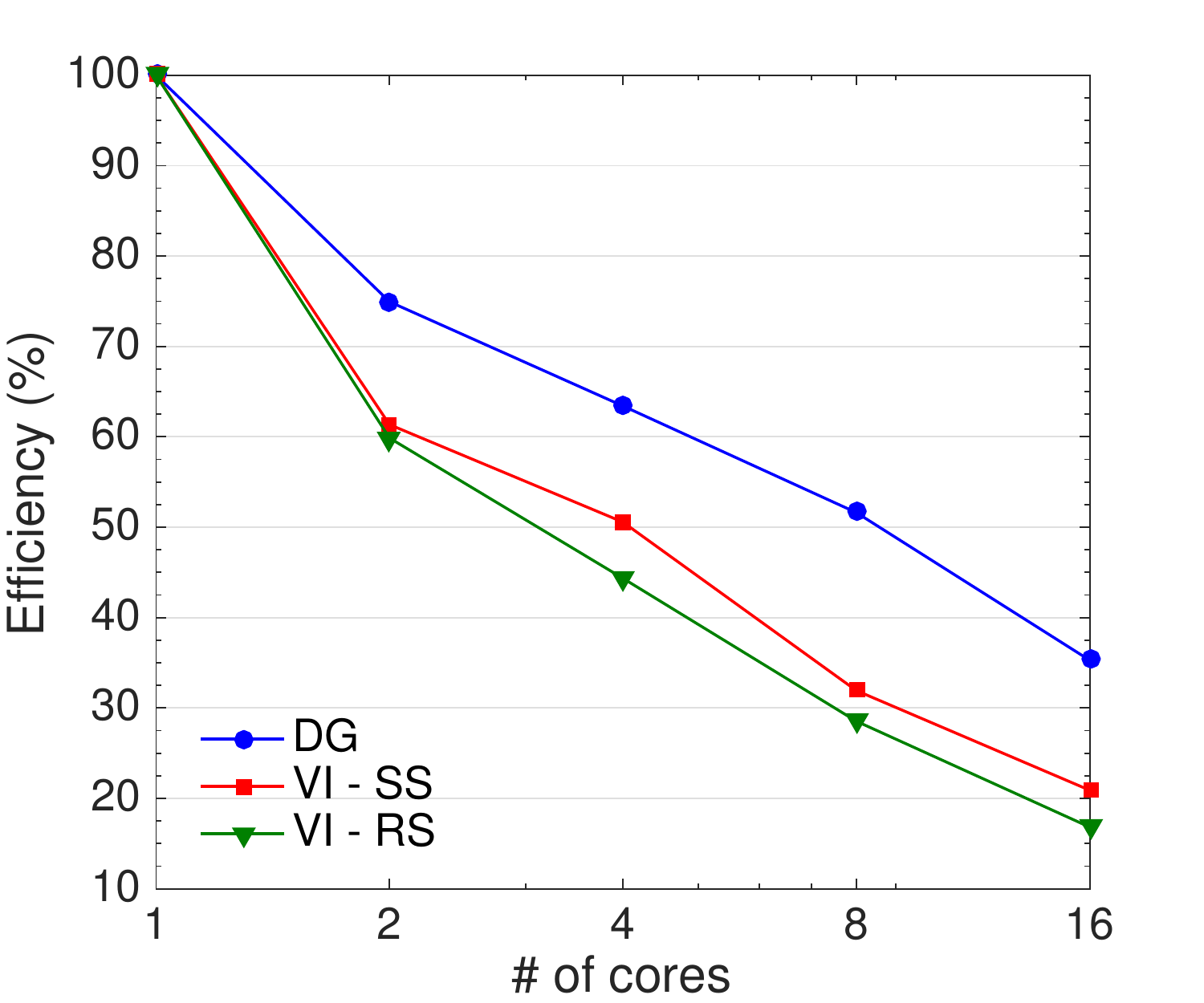}}
\caption{3D advection-diffusion: weak-scaling plots with approximately 100k degrees-of-freedom per core and the corresponding parallel efficiencies.}
\label{Fig:3D_AD_weak}
\end{figure}
\begin{figure}[t]
\centering
\subfloat[SUPG - solve time]{\includegraphics[scale=0.5]{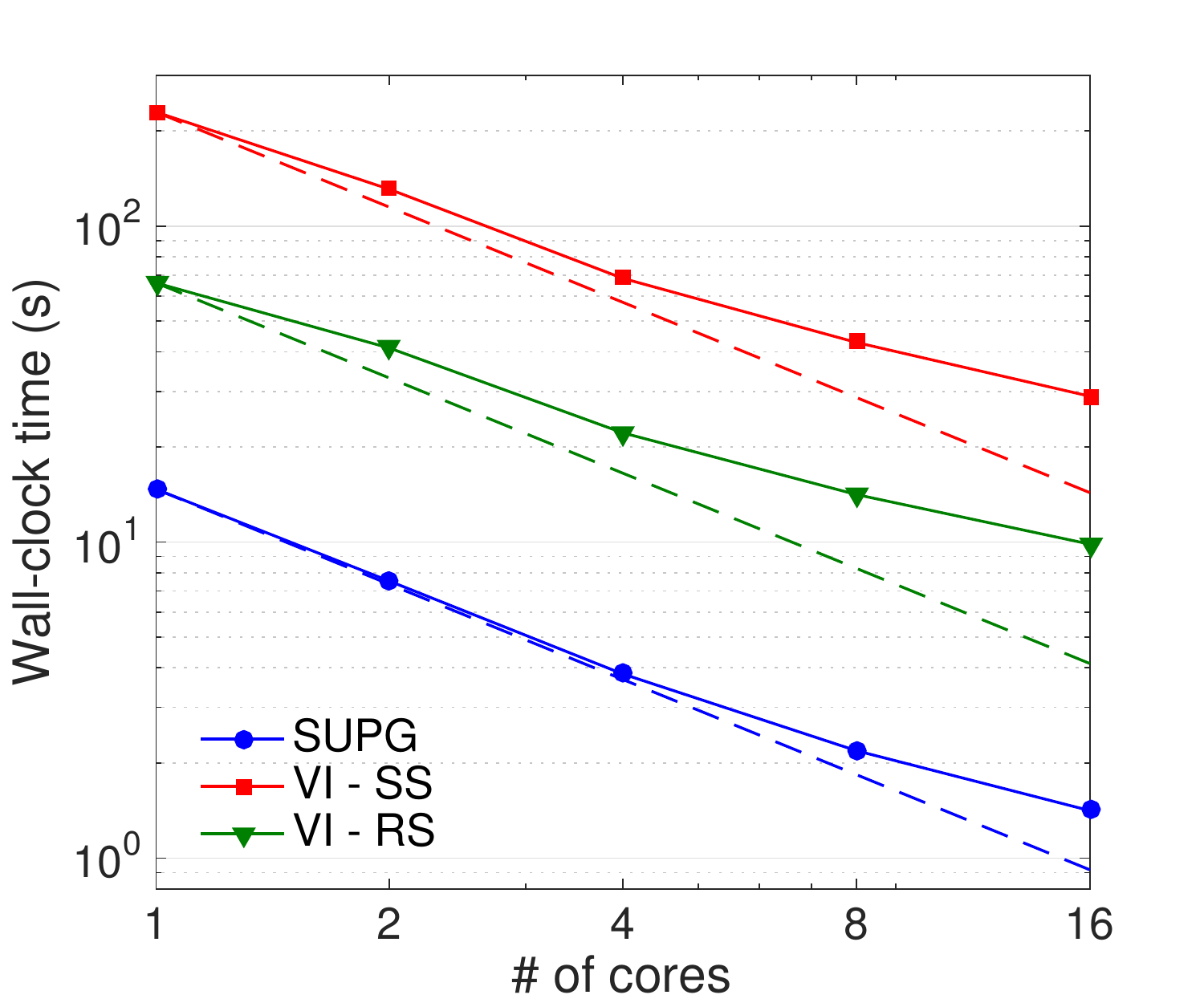}}
\subfloat[SUPG - parallel efficiency]{\includegraphics[scale=0.5]{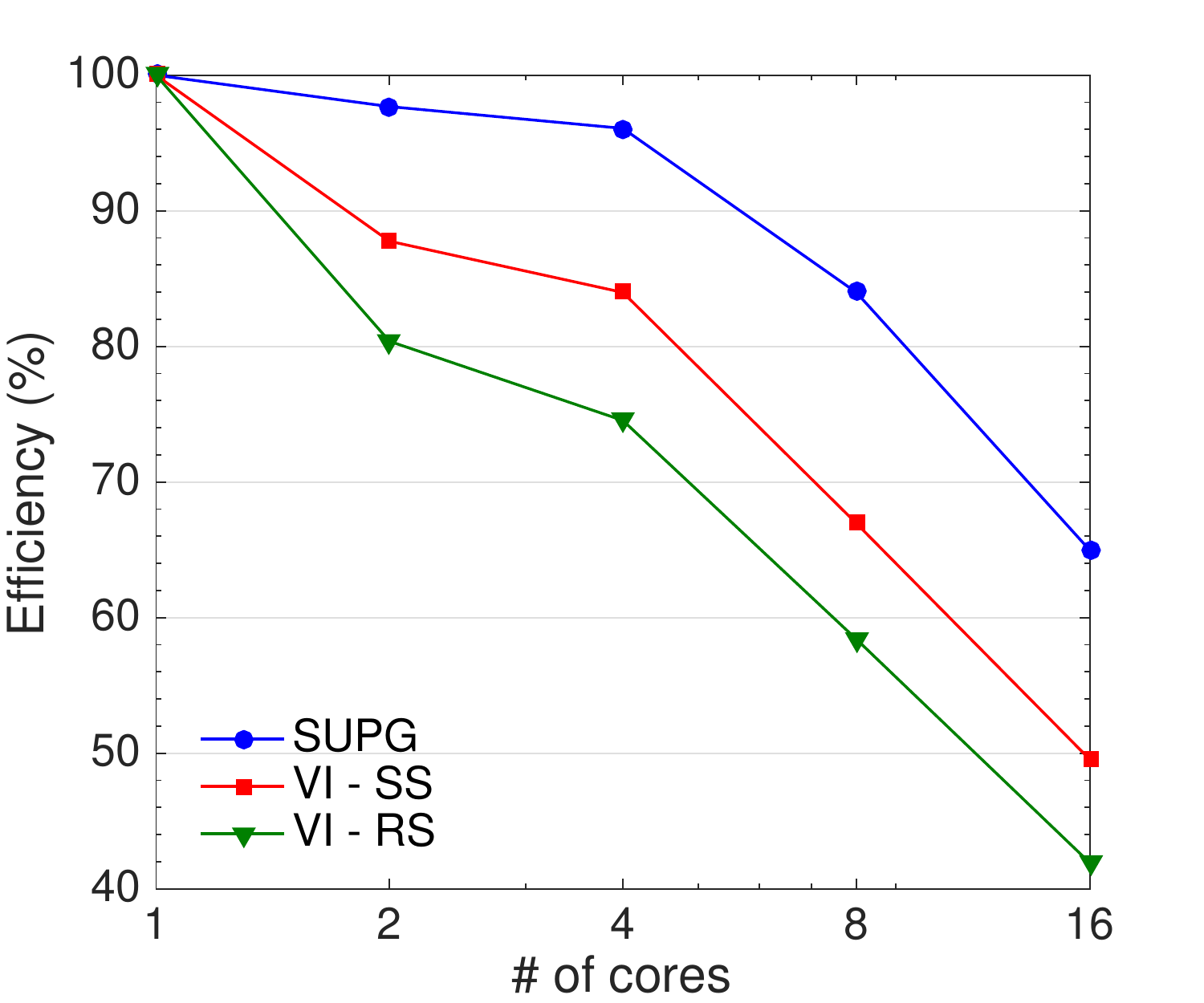}} \\
\subfloat[DG - solve time]{\includegraphics[scale=0.5]{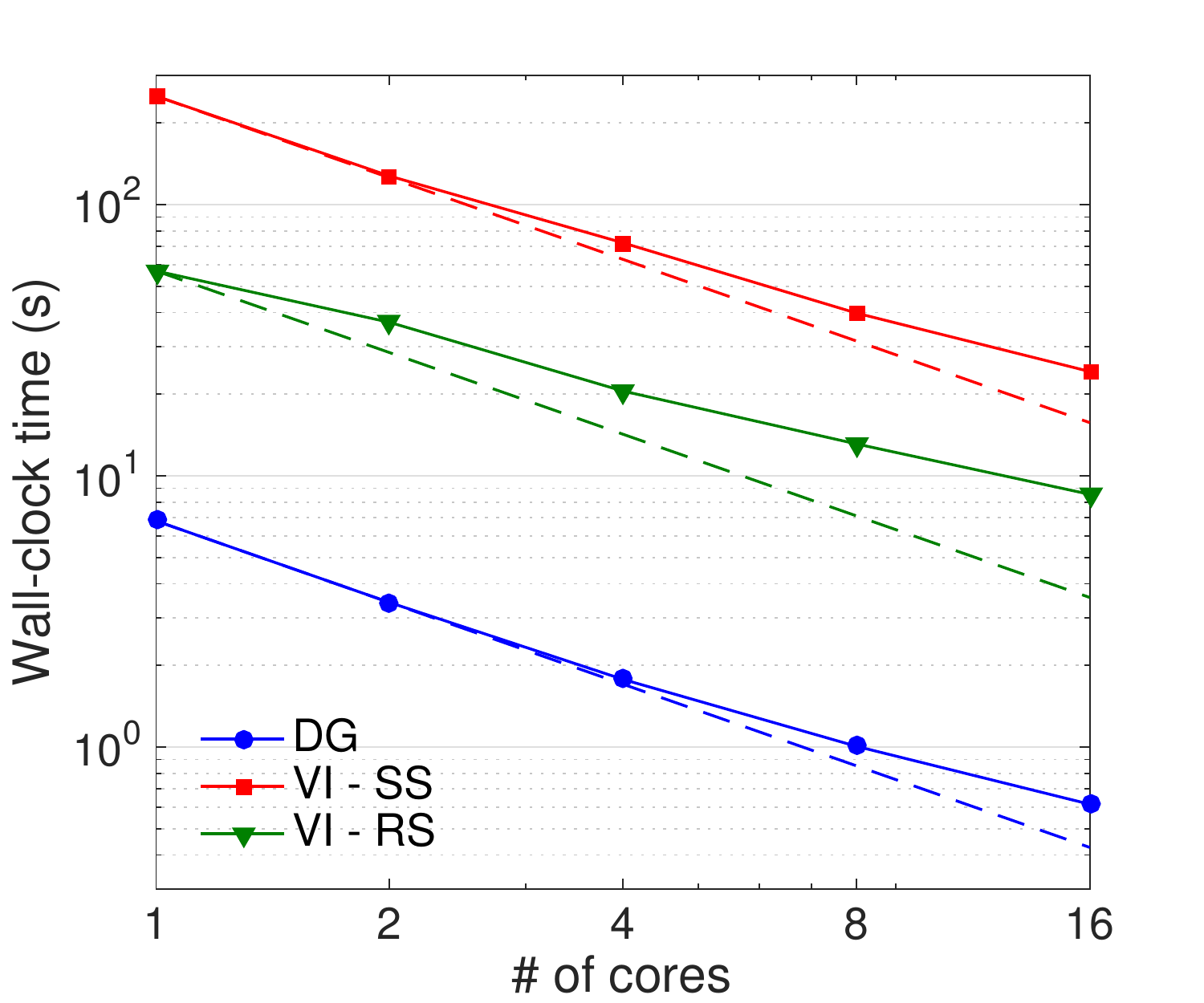}}
\subfloat[DG - parallel efficiency]{\includegraphics[scale=0.5]{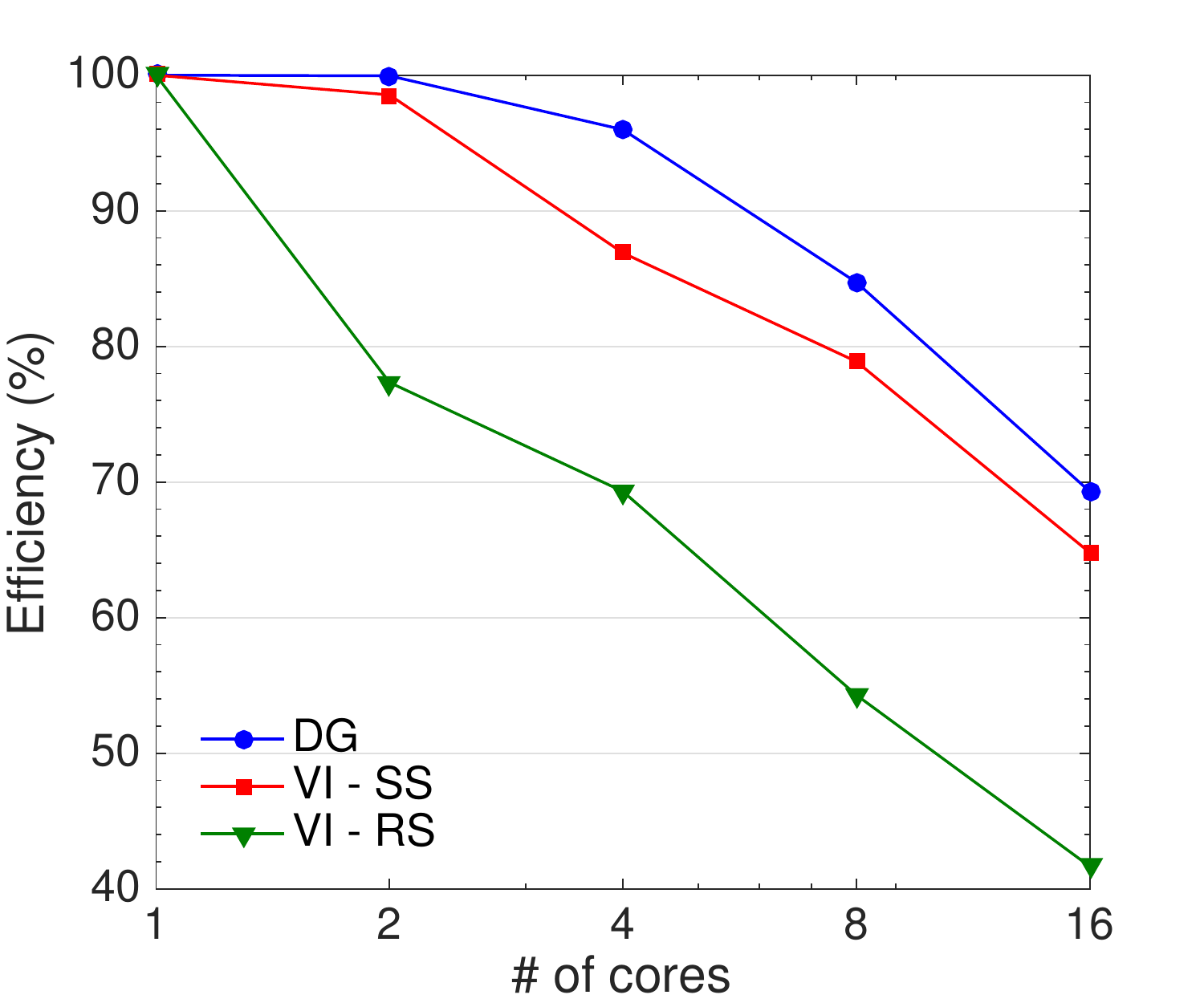}}
\caption{3D advection-diffusion: strong-scaling plots for approximately 500k degrees of freedom ($h$-size = 1/80 and 1/40 for GAL and DG respectively) and the corresponding parallel efficiencies.}
\label{Fig:3D_AD_strong}
\end{figure}

%% file: S6_VI_Transient.tex
\section{EXTENSION TO TRANSIENT ANALYSIS AND COUPLED PROBLEMS}
\label{Sec:S6_VI_TR}
We now illustrate that the proposed framework, which
is based on variational inequalities, can be extended 
to perform a transient analysis. The resulting governing
equations will then be \emph{parabolic} variational 
inequalities. This extension will be illustrated by 
considering the displacement of miscible fluids in 
porous media wherein a fluid displaces a fluid 
with higher viscosity \citep{stalkup1983miscible}.  
Some of the applications of miscible displacement include 
oil recovery and carbon-dioxide sequestration 
\citep{chen1998miscible1,chen1998miscible2}. 
The phenomenon is commonly modeled using coupled flow and 
transport equations, which will be presented below. In this 
section, we will also show how negative concentrations 
can have serious ramifications when simulating non-linear 
transport phenomenon like the displacement of miscible fluids. 

\begin{figure}[t]
\centering
\subfloat[Problem description]{\includegraphics[scale=0.7]{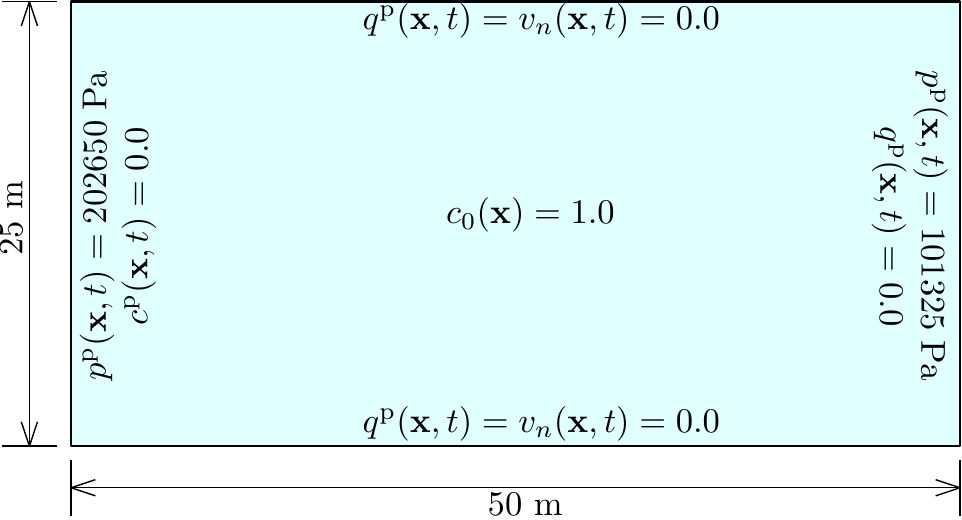}}
\subfloat[Log scale permeability field (m$^2$)]{\includegraphics[scale=0.45]{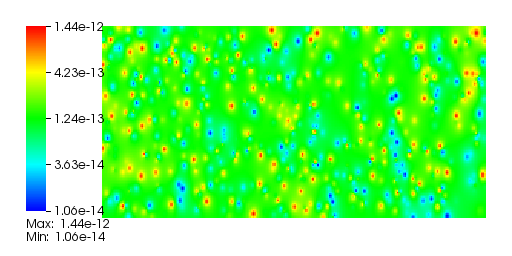}}
\caption{2D miscible displacement: Pictorial description of the boundary value problems for the coupled Darcy and advection-diffusion equations and the corresponding random permeability.}
\label{Fig:2D_transient_description}
\end{figure}
\subsection{Governing equations and temporal discretization}
We denote the time by $t \in [0,\mathcal{T}]$,
where $\mathcal{T}$ denotes the length of the time
interval of interest. For the Darcy equation, 
the boundary is divided into two parts:~$\Gamma^{p}$
and $\Gamma^{v}$, such that $\Gamma^{p} \cup
\Gamma^{v} = \partial \Omega$ and $\Gamma^{p}
\cap \Gamma^{v} = \emptyset$. $\Gamma^{p}$
and $\Gamma^{v}$ denote the parts of the boundary on 
which pressure and velocity boundary conditions are enforced
respectively. We shall denote time-dependent pressure by $p(\mathbf{x},t)$,
time-dependent velocity by $\mathbf{v}(\mathbf{x},t)$, 
concentration-dependent viscosity by $\mu(c(\mathbf{x},t))$,
permeability by $k(\mathbf{x})$, density by $\rho$, time-dependent 
specific body force by $\mathbf{b}(\mathbf{x})$,
time-dependent concentration by $c(\mathbf{x},t)$,
prescribed initial concentration by $c_0(\mathbf{x})$, 
time dependent volumetric source by $f(\mathbf{x},t)$, and
time-dependent diffusivity tensor by $\mathbf{D}(\mathbf{x},t)$. 
For the boundary conditions, the prescribed time-dependent 
concentration is denoted by $c^{\mathrm{p}}(\mathbf{x},t)$, 
prescribed time-dependent pressure by $p^{\mathrm{p}}(\mathbf{x},t)$,
prescribed time-dependent normal component of the velocity by 
$v_{n}(\mathbf{x},t)$, and prescribed time-dependent
flux by $q^{\mathrm{p}}(\mathbf{x},t)$.
The initial boundary value problem for the coupled flow and
advective-diffusive equations can be written as follows:
\begin{subequations}
  \label{Eqn:S6_GE_transient}
  \begin{alignat}{2}
    \label{Eqn:S6_GE_momentum}
    &\frac{\mu(c(\mathbf{x},t))}{k(\mathbf{x})}\mathbf{v}(\mathbf{x},t) + 
    \mathrm{grad}[p(\mathrm{x},t)] = \rho \mathbf{b}(\mathbf{x},t) 
    &&\quad \mathrm{in} \; \Omega \times (0,\mathcal{T})\\
    &\mathrm{div}[\mathbf{v}(\mathrm{x},t)] = 0
    &&\quad \mathrm{in} \; \Omega \times (0,\mathcal{T})\\
    &p(\mathbf{x},t) = p^{\mathrm{p}}(\mathbf{x},t)
    &&\quad \mathrm{on} \; \Gamma^{p} \times (0,\mathcal{T})\\
    \label{Eqn:S6_GE_velocity_bc}
    &\mathbf{v}(\mathbf{x},t)\cdot\widehat{\mathbf{n}} = v_{n}(\mathbf{x},t)
    &&\quad \mathrm{on} \; \Gamma^{v} \times (0,\mathcal{T})\\
    \label{Eqn:S6_GE_AD}
    &\frac{\partial c(\mathbf{x},t)}{\partial t} + 
    \mathbf{v}(\mathbf{x},t)\cdot\mathrm{grad}[c(\mathbf{x},t)] - 
    \mathrm{div}[\mathbf{D}(\mathbf{x},t)\mathrm{grad}[c(\mathbf{x},t)]] = 
    f(\mathbf{x},t)
    &&\quad \mathrm{in} \; \Omega \times (0,\mathcal{T})\\
    &c(\mathbf{x},t) = c^{\mathrm{p}}(\mathbf{x},t)
    &&\quad \mathrm{on} \; \Gamma^{\mathrm{D}} \times (0,\mathcal{T})\\
    &\widehat{\mathbf{n}}(\mathbf{x}) \cdot \left(\mathbf{v}(\mathbf{x},t)
    c(\mathbf{x},t) - \mathbf{D}(\mathbf{x},t)\mathrm{grad}[c(\mathbf{x},t)]\right) = 
    q^{\mathrm{p}}(\mathbf{x},t)
    &&\quad \mathrm{on} \; \Gamma^{\mathrm{N}}_{\mbox{\small inflow}} \times (0,\mathcal{T})\\
    -&\widehat{\mathbf{n}}(\mathbf{x}) \cdot \mathbf{D}(\mathbf{x},t)\mathrm{grad}[c(\mathbf{x},t)] = 
    q^{\mathrm{p}}(\mathbf{x},t)
    &&\quad \mathrm{on} \; \Gamma^{\mathrm{N}}_{\mbox{\small outflow}} \times (0,\mathcal{T}) \\
    \label{Eqn:S6_GE_initial_bc}
    &c(\mathbf{x},0) = c_{0}(\mathbf{x}) && \quad \mathrm{in} \; \Omega
  \end{alignat}
\end{subequations}
where equations \eqref{Eqn:S6_GE_momentum} through \eqref{Eqn:S6_GE_velocity_bc} 
represent the Darcy equation, and equations \eqref{Eqn:S6_GE_AD} through \eqref{Eqn:S6_GE_initial_bc}
represent the transient advection-diffusion equation. To complete the 
coupled problem, the viscosity is assumed to depend exponentially on
concentration:
\begin{subequations}
\begin{align}
  \label{Eqn:S6_viscosity1}
  &\mu(c(\mathbf{x},t)) = \mu_0\exp\big[R_cc(\mathbf{x},t)\big]\\
  \label{Eqn:S6_viscosity2}
  &\mu(c(\mathbf{x},t)) = \mu_0\exp\big[R_c(1-c(\mathbf{x},t))\big]
\end{align}
\end{subequations}
where $\mu_0$ is the base viscosity of the less viscous fluid and 
$R_c$ is the log-mobility ratio in an isothermal miscible displacement.

\begin{table}[b]
\centering
\caption{2D miscible displacement: problem parameters}
\label{Tab:2D_miscible}
\begin{tabular}{ll}
\hline
Parameter & Value \\ \hline
$\mathbf{b}(\boldsymbol{x})$ & $\left\{0,0\right\}$ m/s$^2$\\
$\mu(c(\mathbf{x},t))$ & \eqref{Eqn:S6_viscosity1} \\
$\mu_0$ & $3.95\cdot10^{-5}$ Pa s \\
$R_c$ & 3 \\
$k(\mathbf{x})$ & varies \\
$f(\mathbf{x},t)$ & 0 \\
$\rho$ & 479 kg/m$^3$ \\
$\alpha_L$ & 10$^{-1}$ m \\
$\alpha_T$ & 10$^{-5}$ m\\ 
$\alpha_D$ & 10$^{-9}$ m$^2$/s \\
Number of elements & 31,250\\
Darcy degrees-of-freedom & 94,125\\
Advection-diffusion degrees-of-freedom & 125,000\\
\hline
\end{tabular}
\end{table}
\begin{figure}[t]
\centering
\subfloat[$\boldsymbol{c}_{\mathrm{DG}}$ at $t = 0.5$ years]
{\includegraphics[scale=0.45]{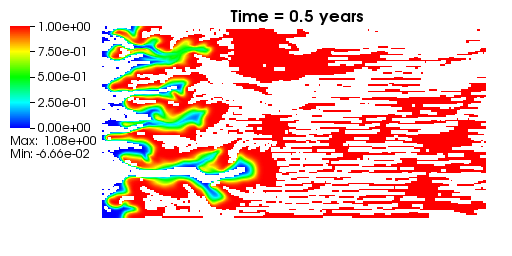}}
\subfloat[$\boldsymbol{c}_{\mathrm{DG}}$ at $t = 1.0$ years]
{\includegraphics[scale=0.45]{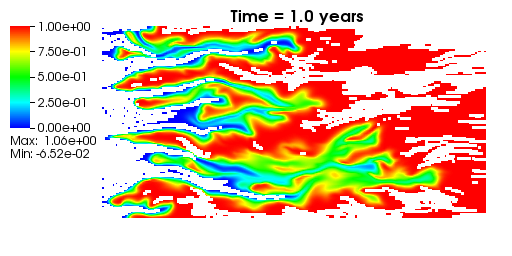}}\\
\subfloat[$\boldsymbol{c}_{\mathrm{RS}}$ at $t = 0.5$ years]
{\includegraphics[scale=0.45]{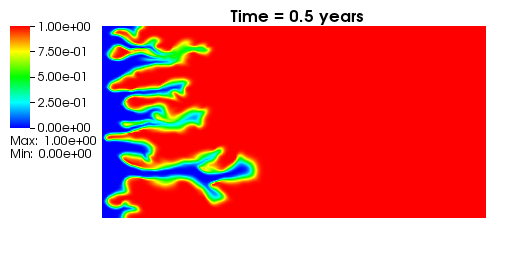}}
\subfloat[$\boldsymbol{c}_{\mathrm{RS}}$ at $t = 1.0$ years]
{\includegraphics[scale=0.45]{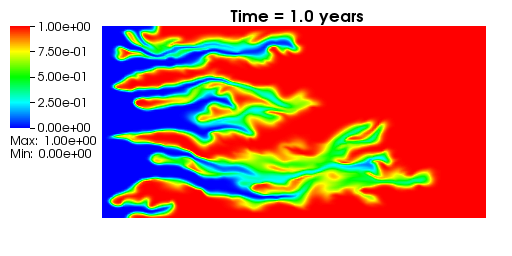}}
\caption{2D miscible displacement: Concentration fields under the DG 
($\boldsymbol{c}_{\mathrm{DG}}$) formulation and VI - RS ($\boldsymbol{c}_{\mathrm{RS}}$) 
method at various time levels. White regions denote violations of the maximum
principle and the non-negative constraint.}
\label{Fig:2D_transient_solution}
\end{figure}
\begin{figure}[t]
\centering
\subfloat[$\boldsymbol{c}_{\mathrm{CLIP}}$ - $\boldsymbol{c}_{\mathrm{RS}}$ at $t= 0.5$ years]
{\includegraphics[scale=0.45]{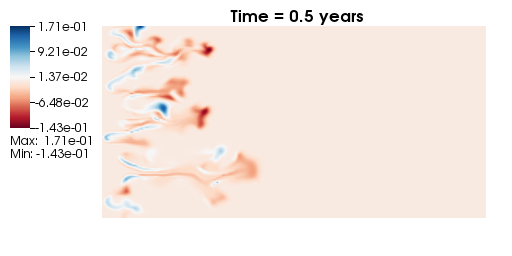}}
\subfloat[$\boldsymbol{c}_{\mathrm{CLIP}}$ - $\boldsymbol{c}_{\mathrm{RS}}$ at $t= 1.0$ years]
{\includegraphics[scale=0.45]{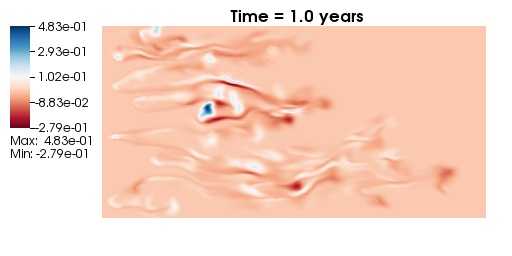}}\\
\subfloat[$\boldsymbol{c}_{\mathrm{DG}}$ - $\boldsymbol{c}_{\mathrm{RS}}$ at $t= 0.5$ years]
{\includegraphics[scale=0.45]{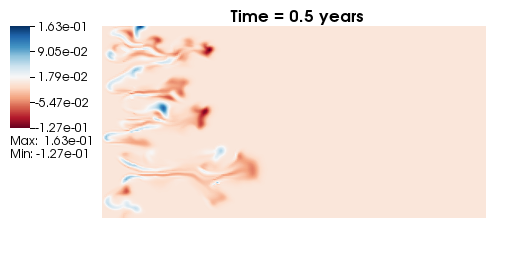}}
\subfloat[$\boldsymbol{c}_{\mathrm{DG}}$ - $\boldsymbol{c}_{\mathrm{RS}}$ at $t= 1.0$ years]
{\includegraphics[scale=0.45]{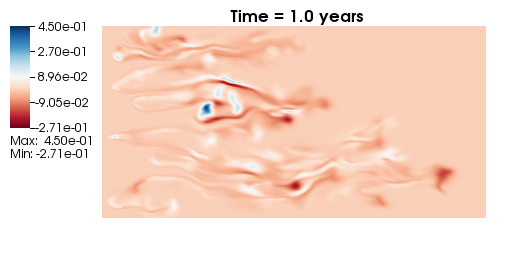}}
\caption{2D miscible displacement: Differences between concentration fields under the DG 
formulation ($\boldsymbol{c}_{\mathrm{DG}}$), VI - RS method ($\boldsymbol{c}_{\mathrm{RS}}$), 
and clipping procedure ($\boldsymbol{c}_{\mathrm{CLIP}}$).}
\label{Fig:2D_transient_difference}
\end{figure}
To solve the transient advection-diffusion equation, we employ the 
method of horizontal lines \citep{grossmann2007numerical}, which 
first discretizes the time derivatives, thereby giving rise to
time-independent equations. The time interval of interest 
is divided into $N$ sub-intervals. That is, 
\begin{align}
  [0,\mathcal{T}] := \bigcup_{n = 0}^{N} [t_n,t_{n+1}]
\end{align}
where $t_n$ denotes the $n$-th time-level. We 
assume that the time-step is uniform, which 
can be written as:
\begin{align}
  \Delta t = t_{n+1} - t_{n}
\end{align}
One can then employ the finite-dimensional 
VI solvers for these resulting equations, which were described earlier in this paper. 
This implies that we will still be solving elliptic VIs of first kind but 
at each time level. This procedure will be illustrated below using the backward Euler method. 
However, a detailed discussion on the effect of time-stepping schemes 
in meeting maximum principles can be found in \citep{Nakshatrala_CiCP_2016}.
For a transient analysis, the proposed framework outlined in 
Section \ref{SS4:proposed_framework} is modified as follows:
\begin{enumerate}[label=\textsf{Step \arabic*}:]
\item Set $t=0.0$, $n=0$, and $\boldsymbol{c}^{(n)} = \boldsymbol{c}_{0}$.
\item Solve Darcy equation:
\begin{enumerate}
\item Compute $\mu(\boldsymbol{c}^{(n)})$.
\item Assemble $\boldsymbol{K}_{vv}$, $\boldsymbol{K}_{vp}$, 
$\boldsymbol{K}_{pv}$, $\boldsymbol{K}_{pp}$, $\boldsymbol{f}_v$ and 
$\boldsymbol{f}_p$.
\item Solve for $\boldsymbol{v}^{(n)}$.
\end{enumerate}
\item Solve advection-diffusion equation:
\begin{enumerate}
\item Compute $\mathbf{D}^{(n)}$.
\item Assemble $\boldsymbol{K}_{c}$ and $\boldsymbol{f}_c$ 
using $\boldsymbol{c}^{(n)}$.
\item Solve for $\boldsymbol{c}_{\mathrm{DG}}^{(n+1)}$.
\item Clip $\boldsymbol{c}_{\mathrm{DG}}^{(n+1)}$ and obtain 
$\boldsymbol{c}^{(n+1)}_{\mathrm{CLIP}}$.
\item Solve the bounded constraint problem for $\boldsymbol{c}^{(n+1)}_{\mathrm{RS}}$
with $\boldsymbol{c}^{(n+1)}_{\mathrm{CLIP}}$ as the initial guess.
\end{enumerate}
\item Set $\boldsymbol{c}^{(n+1)}\longleftarrow\boldsymbol{c}^{(n+1)}_{\mathrm{RS}}$, $t \longleftarrow t + \Delta t$, and $n \longleftarrow n +1$.
\item If $n < N$ go to $\mathsf{Step 2}$.
\end{enumerate}
where $\boldsymbol{K}_{vv}$, $\boldsymbol{K}_{vp}$, 
$\boldsymbol{K}_{pv}$, $\boldsymbol{K}_{pp}$, 
$\boldsymbol{f}_{v}$, and $\boldsymbol{f}_{p}$ 
are the assembled  matrices and vectors for the Darcy equation, 
and $\boldsymbol{K}_{c}$ and $\boldsymbol{f}_{c}$ are for the 
transient advection-diffusion equation. The finite element 
discretization and solution strategy for the steady-state 
Darcy equations can be found in Appendix \ref{A2:darcy}.
\subsection{Numerical results}
\begin{figure}[t]
\centering
\subfloat[Problem description]{\includegraphics[scale=0.44]{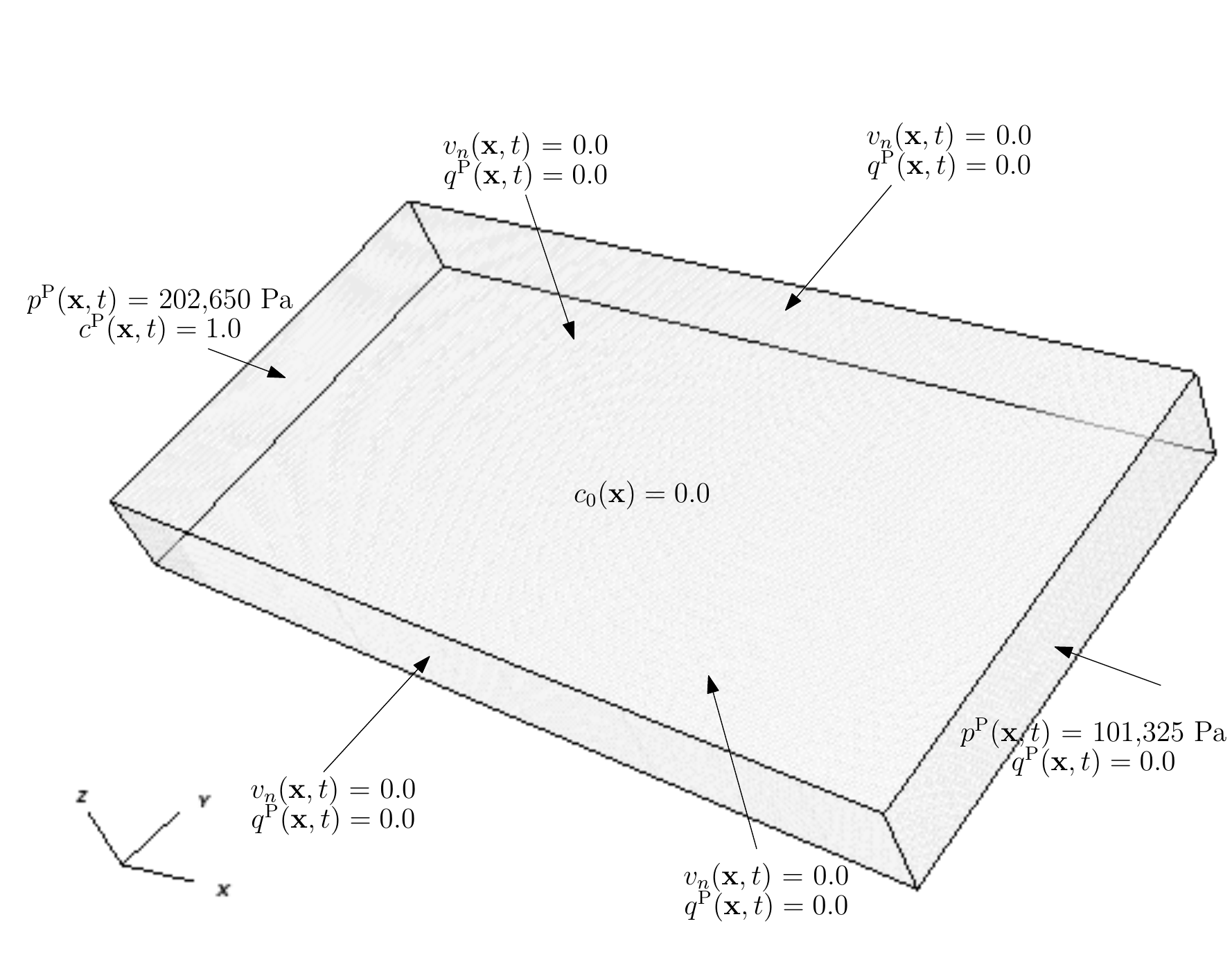}}
\subfloat[Log scale permeability field (m$^2$)]{\includegraphics[scale=0.44]{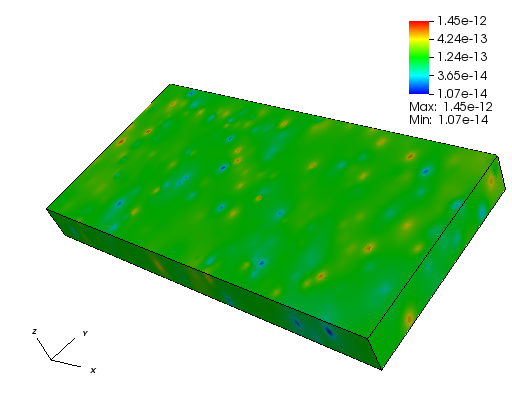}}
\caption{3D miscible displacement: Pictorial description of the boundary value problems (50m$\times$25m$\times$5m domain with 781,250 structured hexahedrons) for the coupled Darcy and advection-diffusion equations and the corresponding random permeability.}
\label{Fig:3D_transient_description}
\end{figure}

Consider a 50m$\times$25m rectangular domain with heterogeneous 
permeability, as shown in Figure \ref{Fig:2D_transient_description}.
The flow will be modeled using Darcy equations, in which the 
viscosity depends on the concentration of the attendant chemical 
species, and the transport of the chemical species will be modeled 
using advection-diffusion equations. 
For the flow subproblem, we prescribe the pressure boundary conditions on 
the left and right sides of the domain and no flow boundary 
conditions on the top and bottom. 
For the transport subproblem, an initial concentration of unity 
is prescribed everywhere in the domain, and a Dirichlet boundary condition
of zero along the left side and zero flux boundary conditions 
on the remaining sides are prescribed. 
A time-step $\Delta t = 1$ day is used 
to simulate the miscible displacement over a time interval 
$\mathcal{T}=1$ year. All other problem parameters and 
material properties can be found in Table \ref{Tab:2D_miscible}. 
Figure \ref{Fig:2D_transient_solution} depicts the concentration profiles
under the DG  and VI - RS methods at $t = 0.5$ and $t = 1.0$ years. 
It can be seen that violations in the maximum principles occur even under 
the coupled flow and transport computational framework. Furthermore, the
violations do not go away as the simulation progresses in time. As it 
may be difficult to distinguish between the VI - RS and DG or 
clipping procedures by directly plotting the solutions, we show 
the differences in the solutions in Figure \ref{Fig:2D_transient_difference}. It can be seen that there 
are significant discrepancies in the development of the plumes.

\begin{table}[b]
\centering
\caption{3D miscible displacement: problem parameters}
\label{Tab:3D_miscible}
\begin{tabular}{ll}
\hline
Parameter & Value \\ \hline
$\mathbf{b}(\boldsymbol{x})$ & $\left\{0,0,-9.81\right\}$ m/s$^2$\\
$\mu(c(\mathbf{x},t))$ & \eqref{Eqn:S6_viscosity2} \\
$\mu_0$ & $3.95\cdot10^{-5}$ Pa s \\
$R_c$ & 3 \\
$k(\mathbf{x})$ & varies \\
$f(\mathbf{x},t)$ & 0 \\
$\rho$ & 479 kg/m$^3$ \\
$\alpha_L$ & 10$^{-1}$ m \\
$\alpha_T$ & 10$^{-5}$ m\\ 
$\alpha_D$ & 10$^{-9}$ m$^2$/s \\
Number of elements & 781,250\\
Darcy degrees-of-freedom & 3,165,625\\
Advection-diffusion degrees-of-freedom & 6,250,000\\
\hline
\end{tabular}
\end{table}
\begin{figure}
\centering
\subfloat[$\boldsymbol{c}_{\mathrm{DG}}$ at $t = 0.4$ years]
{\includegraphics[scale=0.42]{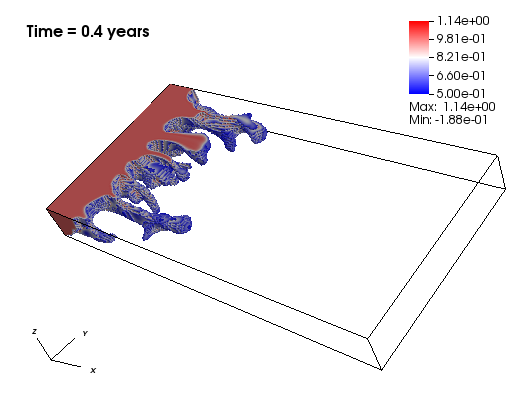}
\label{Fig:3D_transient_conc_DG1}}
\subfloat[$\boldsymbol{c}_{\mathrm{DG}}$ at $t = 1.0$ years]
{\includegraphics[scale=0.42]{Figures/3D_transient_orig_5.png}
\label{Fig:3D_transient_conc_DG2}}\\
\subfloat[$\boldsymbol{c}_{\mathrm{DG}}$ at $t = 0.4$ years]
{\includegraphics[scale=0.42]{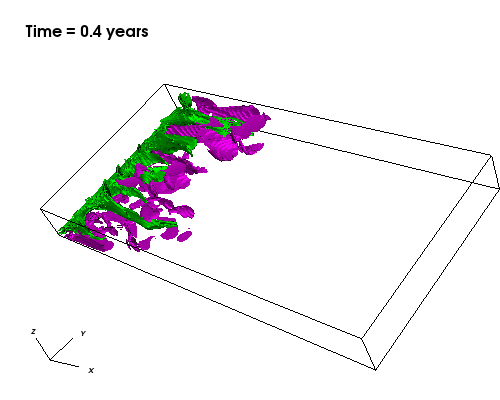}
\label{Fig:3D_transient_violations_DG1}}
\subfloat[$\boldsymbol{c}_{\mathrm{DG}}$ at $t = 1.0$ years]
{\includegraphics[scale=0.42]{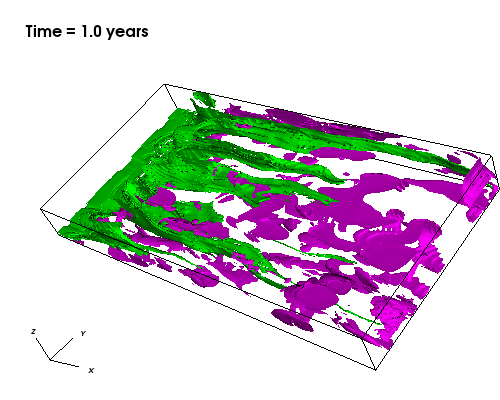}
\label{Fig:3D_transient_violations_DG2}}\\
\subfloat[$\boldsymbol{c}_{\mathrm{RS}}$ at $t = 0.4$ years]
{\includegraphics[scale=0.42]{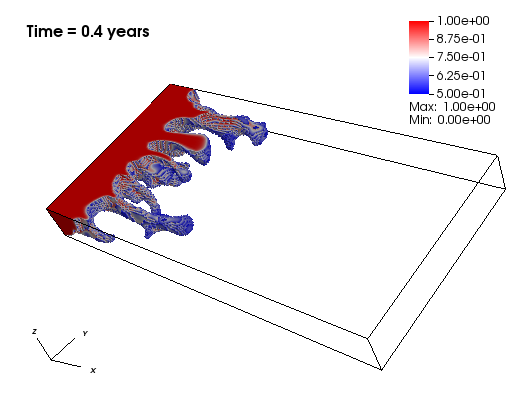}
\label{Fig:3D_transient_conc_RS1}}
\subfloat[$\boldsymbol{c}_{\mathrm{RS}}$ at $t = 1.0$ years]
{\includegraphics[scale=0.42]{Figures/3D_transient_virs_5.png}
\label{Fig:3D_transient_conc_RS2}}
\caption{3D miscible displacement: Top (DG) and bottom (VI) show regions with 
concentrations above 0.5. Middle figures show regions with concentrations above 
1.0 (green) and below 0.0 (purple) (see online version for color figures.)}
\label{Fig:3D_transient_concentrations}
\end{figure}
\begin{figure}[t]
\centering
\subfloat{\includegraphics[scale=0.45]{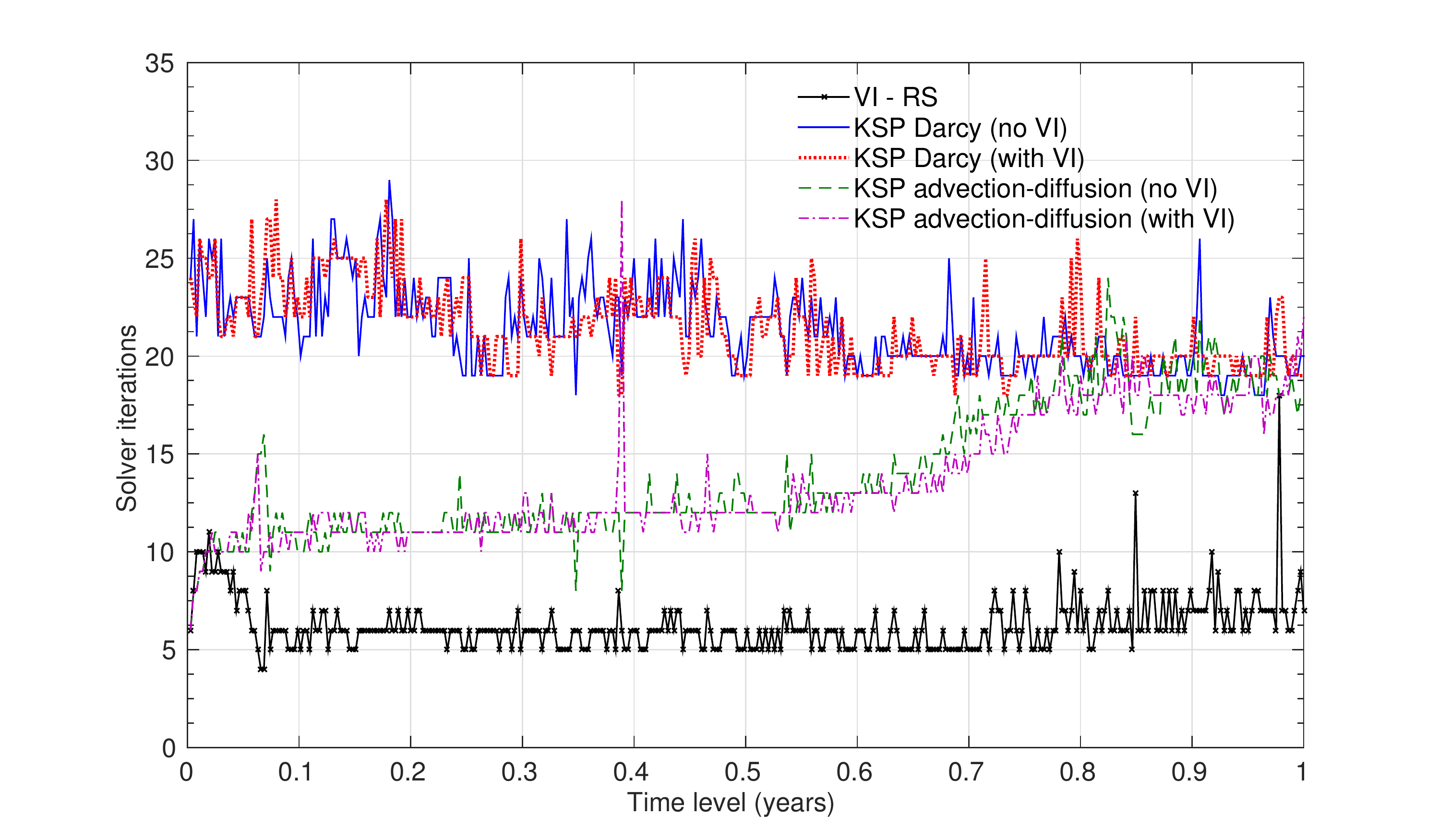}}
\caption{3D miscible displacement: Total number of KSP iterations at each time level for 
the Darcy and advection-diffusion equation with and without VI. Also shown is the 
total number of VI iterations at each time level.}
\label{Fig:3D_transient_iterations}
\end{figure}
\begin{figure}[t]
\centering
\subfloat{\includegraphics[scale=0.45]{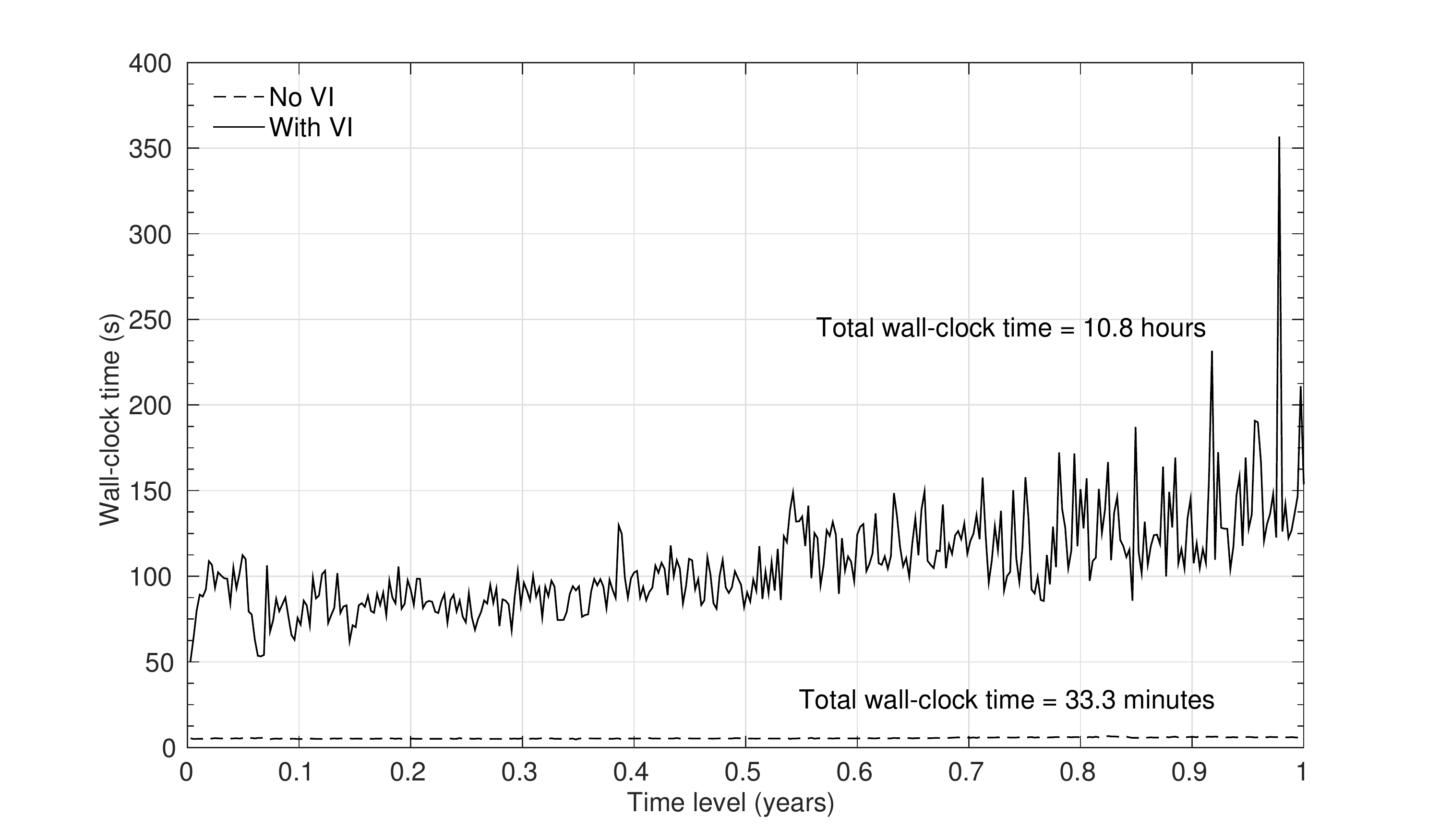}}
\caption{3D miscible displacement: The wall-clock time at each time level with and without 
VI across 40 cores. Also shown is the sum of the wall-clock time across all time levels.}
\label{Fig:3D_transient_time}
\end{figure}
To understand the performance of our VI-based solvers for 
large-scale versions of this problem, we now consider a 50m$\times$25m$\times$5m 
box domain with heterogeneous permeability as shown in Figure 
\ref{Fig:3D_transient_description}. Table \ref{Tab:3D_miscible} 
lists all the necessary problem parameters, and the same time-step and 
time interval from the previous problem is used. This problem is 
now solved in parallel across 40 MPI processes, and Figure 
\ref{Fig:3D_transient_concentrations} depicts the numerical results
under the DG formulation without VI - RS and DG formulation with VI - RS.
The exact regions where violations in the maximum principle and the non-negative
constraint occur are shown in Figures \ref{Fig:3D_transient_violations_DG1} and
\ref{Fig:3D_transient_violations_DG2}. First, we note that this proposed framework 
can successfully eliminate the violations that occur in a large-scale miscible 
displacement simulation. We also note that the development and 
displacement of the plumes is significant affected by whether VI - RS is applied or not; 
the differences between Figures \ref{Fig:3D_transient_conc_DG2} and \ref{Fig:3D_transient_conc_RS2}
are quite evident unlike the 2D example. Next, we note from Figure \ref{Fig:3D_transient_iterations} 
that enforcing the bounded constraints under the VI - RS method will not drastically 
increase the total number of KSP iterations needed for either the Darcy or advection-diffusion 
equations. However, the wall-clock time shown in Figure \ref{Fig:3D_transient_time} indicates 
that the VI - RS method is very expensive. The oscillatory behavior of both the solver 
iterations and wall-clock time at each time level is largely attributed to the heterogeneous 
nature of the problem as well as the number of maximum principle violating degrees-of-freedom 
that naturally arise out of the DG formulation. Although VI-based solvers like VI - RS 
can enforce maximum principles and the non-negative constraint, we have observed that applying 
such methodologies can make the overall advection-diffusion finite element simulation up 
to 20 times as expensive even in a parallel environment.

Before we draw any further conclusions of this paper, we acknowledge 
that we did not perform a numerical 
$h$-convergence study. This is due to, 
despite our best efforts, failure to find an  
advection-diffusion boundary value problem 
that considers anisotropy, has an analytical solution, 
and violates discrete maximum principles. We, therefore,
illustrated the performance of the proposed computational
framework through other means, as presented in the previous 
sections.

%% file: S7_VI_CR.tex
\section{CONCLUDING REMARKS}
\label{Sec:S7_VI_CR}
We presented a robust computational framework based on VIs 
for diffusion and advection-diffusion equations 
that satisfies the discrete maximum principles and the 
non-negative constraint. The framework is applicable to 
\emph{large-scale} and transient problems, and can be solved 
in a parallel setting. The main contributions of this paper and the 
salient features of the proposed formulation can be summarized as follows. 
\begin{enumerate}[label=(C\arabic*)]
\item Realizing and posing the advection-diffusion problem as 
a variational inequality (VI) to meet the discrete maximum
principles and the non-negative constraint.
\item For large-scale problems, we have demonstrated that 
QP solvers, which is a special 
case of VIs, are just as good as VI solvers for symmetric
and positive-definite problems like the diffusion equation. 
On the other hand, the proposed VI-based framework can also 
handle non-self-adjoint operators.
\item Unlike the non-negative framework proposed in \citep{Mudunuru_JCP_2016},
which is based on a mixed least-squares WF, the proposed 
framework can utilize any finite element formulation 
including single-field formulations, and these formulations need not 
result in symmetric and positive definite coefficient matrices.
\item The proposed framework allows one to leverage on existing 
state-of-the-art computational frameworks for solving VIs. In particular, 
the Firedrake project, which provides access to parallel solvers in PETSc 
and TAO libraries, can serve as a suitable platform for implementing the proposed 
framework, as illustrated in this paper.
\item This framework is suitable for many important
applications like miscible displacement, subsurface remediation, 
and transport of radionuclides. In these applications, one encounters not only 
highly anisotropic medium properties but also highly non-linear phenomena 
due to aqueous complexation and kinetic reactions.
\end{enumerate}
A logical next-step of our work is to extend
this computational framework to advective-diffusive-reactive 
equations.

%% file: A1_Code.tex
\section{Firedrake Project}
\label{A1:code}
The Firedrake project \citep{Rathgeber_ACM_2015,Luporini_ACM_2016,
Luporini_ACMACO_2015} is a python-based library that provides an automated 
system for the solution of partial differential equations using the finite 
element method. Like the FEniCS Project \citep{logg2012automated,alnaes2015fenics}, 
it is also built upon several scientific packages and can employ 
parallel computing tools across either CPUs or GPUs to obtain the solution. 
Two of its main leveraged components are the Unified Form Language (UFL) 
\citep{Alnaes_UFL_2014}, used to declare finite element discretizations 
of variational forms, and the PyOP2 system \citep{Rathgeber_SC_2012,
Markall_ISC_2013}, used for the parallel assembly of the finite element 
discrete formulations. The main difference between the FEniCS and Firedrake 
project is that all data structures, linear solvers, non-linear solvers, 
and optimization solvers for the latter are provided entirely by the PETSc and TAO 
libraries. The mesh can either be generated internally or imported from third party 
mesh generators like GMSH \citep{Geuzaine_IJNME_2009}, and the parallel 
partitioning of the mesh is achieved through packages like Chaco 
\citep{chaco_SC_1995}. Another important feature utilized in this paper is
extruded meshes. The internal mesh algorithm generates and partitions a 
2D quadrilateral base mesh and is extruded into a hexahedron mesh using 
the algorithms listed in \citep{Homolya_SIAMJOMP_2016,
McRae_SIAMJOMP_2014}. To facilitate the readers to be able to reproduce
the results presented in this paper, we provided some useful
Firedrake-related files below.
\lstset{language=Python}
\begin{lstlisting}[caption=2D GAL diffusion example,label=Code:ex1,frame=single]
# Load firedrake environment
from firedrake import *

# Create mesh
mesh = UnitSquareMesh(200,200,quadrilateral=True)

# Function spaces
P = FunctionSpace(mesh, 'Lagrange', 1)
Q = TensorFunctionSpace(mesh, 'Lagrange', 1)
u = TrialFunction(P)
v = TestFunction(P)

# Bounds
cmin = 0.0
cmax = PETSc.INFINITY
lb = Function(P)
lb.assign(cmin)
ub = Function(P)
ub.assign(cmax)

# Diffusion tensor
eps = 1e-4
D = interpolate(Expression(('eps*x[0]*x[0]+x[1]*x[1]','-(1-eps)*x[0]*x[1]',
    '-(1-eps)*x[0]*x[1]','eps*x[1]*x[1]+x[0]*x[0]'),eps=eps),Q)

# Forcing function
lb = 0.375
ub = 0.625
f = interpolate(Expression(('x[0] >= lb && x[0]<= ub && x[1] >= lb &&
    x[1] <= ub ? 1.0 : 0.0'),lb=lb,ub=ub),P)

# GAL formulation
a = dot(D * grad(u), grad(v)) * dx
L = v * f * dx

# Homogeneous boundary conditions
bcs = DirichletBC(P, Constant(0.0), (1,2,3,4))

# Assemble coefficient matrix
A = assemble(a, bcs=bcs)

# Assemble forcing vector
tmp = Function(P)
bcs.apply(tmp)
b = Function(P)
bfree = assemble(L)
rhs_bcs = assemble(action(a,tmp))
b.assign(bfree - rhs_bcs)
bcs.apply(b)  

# Create PETSc solver
initial_solver = PETSc.KSP().create(PETSc.COMM_WORLD)
initial_solver.setOptionsPrefix("initial_")
initial_solver.setOperators(A.M.handle)
initial_solver.setFromOptions()
initial_solver.setUp()

# Solve problem
solution = Function(P)
with b.dat.vec_ro as b_vec, solution.dat.vec as sol_vec:
  initial_solver.solve(b_vec,sol_vec)
\end{lstlisting}
\begin{lstlisting}[caption=2D SUPG advection-diffusion example,label=Code:ex2,frame=single]
# Load firedrake environment
from firedrake import *

# Load GMSH file
mesh = Mesh('square_hole.msh')

# Function spaces
P = FunctionSpace(mesh, 'Lagrange', 1)
V = VectorFunctionSpace(mesh, 'Lagrange', 1)
u = TrialFunction(P)
v = TestFunction(P)

# Bounds
cmin = 0.0
cmax = 1.0
lb = Function(P)
lb.assign(cmin)
ub = Function(P)
ub.assign(cmax)

# Velocity field
velocity = interpolate(Expression(('cos(2*pi*x[1]*x[1])',
    'sin(2*pi*x[0])+cos(2*pi*x[0]*x[0])')),V)

# Diffusion tensor
alphaT = Constant(1e-5)
alphaL = Constant(1e-1)
alphaD = Constant(1e-9)
normv = sqrt(dot(velocity,velocity))
Id = Identity(mesh.geometric_dimension())
D = (alphaD + alphaT*normv)*Id + 
    (alphaL - alphaT)*outer(velocity,velocity)/normv

# Forcing function
f = Constant(0.0)

# SUPG weak form
h = CellSize(mesh)
Pe = h/(2*normv)*dot(velocity,grad(v))
ar = Pe*(dot(velocity,grad(u)) - div(D*grad(u)))*dx
a = ar + v*dot(velocity,grad(u))*dx + dot(grad(v),D*grad(u))*dx
a = dot(D*grad(u), grad(v))*dx
L = (v + Pe)*f*dx

# Boundary conditions
bc1 = DirichletBC(P, Constant(0.0), (12,13,14,15)) # Outer square
bc2 = DirichletBC(P, Constant(1.0), (16,17,18,19)) # Inner square
bcs = [bc1,bcs2]

# Homogeneous boundary conditions
bcs = DirichletBC(P, Constant(0.0), (1,2,3,4))

# Assemble coefficient matrix
A = assemble(a, bcs=bcs)

# Assemble forcing vector
tmp = Function(P)
for bc in bcs:
  bc.apply(tmp)
b = Function(P)
bfree = assemble(L)
rhs_bcs = assemble(action(a,tmp))
b.assign(bfree - rhs_bcs)
for bc in bcs:
  bc.apply(b)   

# Create PETSc solver
initial_solver = PETSc.KSP().create(PETSc.COMM_WORLD)
initial_solver.setOptionsPrefix("initial_")
initial_solver.setOperators(A.M.handle)
initial_solver.setFromOptions()
initial_solver.setUp()

# Solve problem
solution = Function(P)
with b.dat.vec_ro as b_vec, solution.dat.vec as sol_vec:
  initial_solver.solve(b_vec,sol_vec)
\end{lstlisting}
\begin{lstlisting}[caption=3D DG advection-diffusion example,label=Code:ex3,frame=single]
# Load firedrake environment
from firedrake import *

# Number of elements in each spatial dimension
seed = 40

# 2D base mesh
mesh = UnitSquareMesh(seed,seed,quadrilateral=True)
# Extruded mesh
mesh = ExtrudedMesh(meshbase,seed)

# Function spaces
P = FunctionSpace(mesh, 'DG', 1)
V = VectorFunctionSpace(mesh, 'DG', 1)
u = TrialFunction(P)
v = TestFunction(P)

# Bounds
cmin = 0.0
cmax = PETSc.INFINITY
lb = Function(P)
lb.assign(cmin)
ub = Function(P)
ub.assign(cmax)

# Velocity field
velocity = interpolate(Expression(('0.3*sin(2*pi*x[2])+cos(3*pi*x[1])',
    '0.65*sin(2*pi*x[0])+0.3*cos(5*pi*x[2])',
    'sin(4*pi*x[1])+0.65*cos(6*pi*x[0])')),V)

# Diffusion tensor
alphaT = Constant(1e-5)
alphaL = Constant(1e-1)
alphaD = Constant(1e-9)
normv = sqrt(dot(velocity,velocity))
Id = Identity(mesh.geometric_dimension())
D = (alphaD + alphaT*normv)*Id + 
    (alphaL - alphaT)*outer(velocity,velocity)/normv

# Forcing function
f1 = interpolate(Expression(('x[0] >= 0.4 && x[0] <= 0.5 && x[1] >= 0.2 && 
    x[1] <= 0.3 && x[2] >= 0.1 && x[2] <= 0.2 ? 1.0 : 0.0')),P)
f2 = interpolate(Expression(('x[0] >= 0.8 && x[0] <= 0.9 && x[1] >= 0.4 && 
    x[1] <= 0.5 && x[2] >= 0.2 && x[2] <= 0.3 ? 1.0 : 0.0')),P)
f3 = interpolate(Expression(('x[0] >= 0.5 && x[0] <= 0.6 && x[1] >= 0.7 && 
    x[1] <= 0.8 && x[2] >= 0.3 && x[2] <= 0.4 ? 1.0 : 0.0')),P)
f4 = interpolate(Expression(('x[0] >= 0.3 && x[0] <= 0.4 && x[1] >= 0.5 && 
    x[1] <= 0.6 && x[2] >= 0.2 && x[2] <= 0.3 ? 1.0 : 0.0')),P)
f5 = interpolate(Expression(('x[0] >= 0.5 && x[0] <= 0.6 && x[1] >= 0.2 && 
    x[1] <= 0.3 && x[2] >= 0.6 && x[2] <= 0.7 ? 1.0 : 0.0')),P)
f6 = interpolate(Expression(('x[0] >= 0.6 && x[0] <= 0.7 && x[1] >= 0.5 && 
    x[1] <= 0.6 && x[2] >= 0.7 && x[2] <= 0.8 ? 1.0 : 0.0')),P)
f7 = interpolate(Expression(('x[0] >= 0.4 && x[0] <= 0.5 && x[1] >= 0.7 && 
    x[1] <= 0.8 && x[2] >= 0.8 && x[2] <= 0.9 ? 1.0 : 0.0')),P)
f8 = interpolate(Expression(('x[0] >= 0.1 && x[0] <= 0.2 && x[1] >= 0.4 && 
    x[1] <= 0.5 && x[2] >= 0.7 && x[2] <= 0.8 ? 1.0 : 0.0')),P)
f = f1 + f2 + f3 + f4 + f5 + f6 + f7 + f8

# Parameters
h = Constant(1/float(seed)) # h-size
gamma = Constant(8/3) # Penalty term
n = FacetNormal(mesh) # Unit outward normal
vn = 0.5*(dot(velocity,n) + abs(dot(velocity,n))) # Upwinding term

# DG weak formulation
a = inner(D * grad(u), grad(v)) * dx(degree=(3,3)) \
    - dot(jump(v,n),avg(D*grad(u)))*(dS_h + dS_v) \
    - dot(avg(D*grad(v)),jump(u,n))*(dS_h + dS_v) \
    + gamma/h*dot(jump(v,n),jump(u,n))*(dS_h + dS_v) \
    - dot(grad(v),velocity*u)*dx(degree=(3,3)) \
    + dot(jump(v),vn('+')*u('+')-vn('-')*u('-'))*(dS_h+dS_v) \
    + dot(v, vn*u)*(ds_v+ds_t+ds_b)
L = v * f * dx(degree=(3,3))

# Boundary conditions
bc1 = DirichletBC(V, Constant(0.0), (1,2,3,4), method="geometric")
bc2 = DirichletBC(V, Constant(0.0), "bottom", method="geometric")
bc3 = DirichletBC(V, Constant(0.0), "top", method="geometric")
bcs = [bc1, bc2, bc3]

# Homogeneous boundary conditions
bcs = DirichletBC(P, Constant(0.0), (1,2,3,4))

# Assemble coefficient matrix
A = assemble(a, bcs=bcs)

# Assemble forcing vector
tmp = Function(P)
for bc in bcs:
  bc.apply(tmp)
b = Function(P)
bfree = assemble(L)
rhs_bcs = assemble(action(a,tmp))
b.assign(bfree - rhs_bcs)
for bc in bcs:
  bc.apply(b)   

# Create PETSc solver
initial_solver = PETSc.KSP().create(PETSc.COMM_WORLD)
initial_solver.setOptionsPrefix("initial_")
initial_solver.setOperators(A.M.handle)
initial_solver.setFromOptions()
initial_solver.setUp()

# Solve problem
solution = Function(P)
with b.dat.vec_ro as b_vec, solution.dat.vec as sol_vec:
  initial_solver.solve(b_vec,sol_vec)
\end{lstlisting}
\begin{lstlisting}[caption=GMSH geometry file for Listing \ref{Code:ex2},label=Code:gmsh,frame=single]
Point(1) = {0, 0, 0, 1.0};
Point(2) = {1, 0, 0, 1.0};
Point(3) = {1, 1, 0, 1.0};
Point(4) = {0, 1, 0, 1.0};
Point(5) = {4/9, 4/9, 0, 1.0};
Point(6) = {5/9, 4/9, 0, 1.0};
Point(7) = {5/9, 5/9, 0, 1.0};
Point(8) = {4/9, 5/9, 0, 1.0};
Line(1) = {1, 2}; 
Line(2) = {2, 3}; 
Line(3) = {3, 4}; 
Line(4) = {4, 1}; 
Line(5) = {5, 6}; 
Line(6) = {6, 7}; 
Line(7) = {7, 8}; 
Line(8) = {8, 5}; 
Line Loop(9) = {4, 1, 2, 3}; 
Line Loop(10) = {8, 5, 6, 7}; 
Plane Surface(11) = {9, 10};
Physical Line(12) = {4};
Physical Line(13) = {1};
Physical Line(14) = {2};
Physical Line(15) = {3};
Physical Line(16) = {7};
Physical Line(17) = {6};
Physical Line(18) = {5};
Physical Line(19) = {8};
Physical Surface(20) = {11};
\end{lstlisting}
\begin{lstlisting}[caption=Semi-smooth (VI - SS) method,label=Code:ss,frame=single]
# Create TAO object
ss_solver = PETSc.TAO().create(PETSc.COMM_WORLD)
ss_solver.setOptionsPrefix("ss_")

# Semi-smooth call-backs
def ss_formJac(tao, petsc_x, petsc_J, petsc_JP, A=None, a=None, bcs=None):
  A = assemble(a, bcs=bcs, tensor=A)
  A.M._force_evaluation()  
def ss_formFunc(tao, petsc_x, petsc_g, A=None, a=None, b=None, bcs=None):
  A = assemble(a, bcs=bcs, tensor=A)
  with b.dat.vec as b_vec:
    A.M.handle.mult(petsc_x, petsc_g)
    petsc_g.axpy(-1.0,b_vec) 
    
# Setup
ss_con = Function(solution.function_space())
with ss_con.dat.vec as con_vec, lb.dat.vec as lb_vec, ub.dat.vec as ub_vec:
  ss_solver.setConstraints(ss_formFunc, con_vec, kargs={'A':A,'a':a,'b':b,'bcs':bcs})
  ss_solver.setJacobian(ss_formJac,A.M.handle, kargs={'A':A,'a':a,'bcs':bcs})
  ss_solver.setType(PETSc.TAO.Type.SSFLS) # can also be ASFLS/SSILS/ASILS
  ss_solver.setVariableBounds(lb_vec,ub_vec)
  ss_solver.setFromOptions()

# Solve problem
def viss():
  with solution.dat.vec as sol_vec:
    ss_solver.solve(sol_vec)
\end{lstlisting}
\begin{lstlisting}[caption=Reduced-space active set (VI - RS) method,label=Code:rs,frame=single]
# Create SNES object
rs_solver = PETSc.SNES().create(PETSc.COMM_WORLD)
rs_solver.setOptionsPrefix("rs_")
  
# Reduced-space active-set call-backs
def rs_formJac(snes, petsc_x, petsc_J, petsc_JP):
  pass 
def rs_formFunc(snes, petsc_x, petsc_g, A=None, b=None):
  with b.dat.vec as b_vec:
    A.M.handle.mult(petsc_x, petsc_g)
    petsc_g.axpy(-1.0,b_vec)  
rs_con = Function(solution.function_space())

# Solve problem
def virs():
  with solution.dat.vec as sol_vec, lb.dat.vec as lb_vec, ub.dat.vec as ub_vec:
    with rs_con.dat.vec as con_vec:
      rs_solver.setFunction(rs_formFunc, con_vec, kargs={'A':A,'b':b})
    rs_solver.setJacobian(rs_formJac,A.M.handle)
    rs_solver.setType(PETSc.SNES.Type.VINEWTONRSLS)
    rs_solver.setVariableBounds(lb_vec,ub_vec)
    rs_solver.setFromOptions()
    rs_solver.solve(None,sol_vec)
    rs_solver.reset()
\end{lstlisting}
\begin{lstlisting}[caption=Trust region Newton (QP - TRON) 
method,label=Code:tron,frame=single]
# Create TAO object
tron_solver = PETSc.TAO().create(PETSc.COMM_WORLD)
tron_solver.setOptionsPrefix("tron_")
  
# TRON call-backs
def tron_formHess(tao, petsc_x, petsc_H, petsc_HP):
  pass
def tron_formObjGrad(tao, petsc_x, petsc_g, A=None, b=None):
  with b.dat.vec_ro as b_vec:
    A.M.handle.mult(petsc_x, petsc_g)
    xtHx = petsc_x.dot(petsc_g)
    xtf = petsc_x.dot(b_vec)
    petsc_g.axpy(-1.0,b_vec)
    return 0.5*xtHx - xtf

# Setup
with lb.dat.vec as lb_vec, ub.dat.vec as ub_vec:
  tron_solver.setVariableBounds(lb_vec,ub_vec)
tron_solver.setObjectiveGradient(tron_formObjGrad, kargs={'A':A,'b':b})
tron_solver.setHessian(tron_formHess,A.M.handle)
tron_solver.setType(PETSc.TAO.Type.TRON)
tron_solver.setFromOptions()
  
# Solve problem
def qptron():
  with solution.dat.vec as sol_vec:
    tron_solver.solve(sol_vec)
\end{lstlisting}

%% file: A2_Darcy.tex
\section{Solution strategy for the Darcy equation}
\label{A2:darcy}
\subsection{Weak formulation}
Let $\mathbf{w}(\mathbf{x})$ and $q(\mathbf{x})$ represent
the weighting functions for velocity and pressure respectively.
The relevant function spaces read as follows:
\begin{subequations}
\begin{align}
    &\mathcal{V} := \left\{\mathbf{v} 
    \in \left(L_2(\Omega)\right)^{d} \; \Bigm\vert \; 
    \mathrm{div}[\mathbf{v}] \in L_2(\Omega), \; 
    \mathbf{v} \cdot \widehat{\mathbf{n}}
    = v_n \; \mathrm{on} 
    \; \Gamma^{v} \right\} \\
    &\mathcal{W} := \left\{\mathbf{w}
    \in \left(L_2(\Omega)\right)^{d} \; \Bigm\vert \; 
    \mathrm{div}[\mathbf{w}] \in L_2(\Omega), \; 
    \mathbf{w} \cdot \widehat{\mathbf{n}}  = 0 \; 
    \mathrm{on} \; \Gamma^{v} 
    \right\}  \\
    &\mathcal{P} :=  L_2(\Omega)
\end{align}
\end{subequations}
where $L_2(\Omega)$ is the space of square integrable functions. 
The WF under the classical mixed formulation for the Darcy equations 
\eqref{Eqn:S6_GE_momentum} through \eqref{Eqn:S6_GE_velocity_bc} 
reads:~Find $\mathbf{v}(\mathbf{x})\in\mathcal{V}$ and 
$p(\mathbf{x})\in\mathcal{P}$ such that we have:
\begin{align}
  \mathcal{B}(\mathbf{w},q;\mathbf{v},p) = \mathcal{L}(\mathbf{w},q)\quad
  \forall \mathbf{w}(\mathbf{x}) \in \mathcal{W},\;q(\mathbf{x})\in\mathcal{P}
\end{align}
where the bilinear form and linear functional are:
\begin{align}
  \label{Eqn:AppendixB_bilinear}
  \mathcal{B}(\mathbf{w},q;\mathbf{v},p) &:= \left(\mathbf{w}(\mathbf{x});\;
  \frac{\mu(c(\mathbf{x}))}{k(\mathbf{x})}\mathbf{v}(\mathbf{x})\right)_{\Omega}
  -\left(\mathrm{div}[\mathbf{w}(\mathbf{x})];\;p(\mathbf{x})\right)_{\Omega}
  -\left(q(\mathbf{x});\;\mathrm{div}[\mathbf{w}(\mathbf{x})]\right)_{\Omega}\\
  \label{Eqn:AppendixB_linear}
  \mathcal{L}(\mathbf{w},q) &:= \left(\mathbf{w}(\mathbf{x});\;\rho\mathbf{b}(\mathbf{x})
  \right)_{\Omega} - \left(\mathbf{w}(\mathbf{x})\cdot\widehat{\mathbf{n}}(\mathbf{x});\;
  p^{\mathrm{P}}\right)_{\Gamma^p}
\end{align}
The lowest order Raviart-Thomas (RT0) space \citep{raviart1977mixed} is employed 
because it ensures element-wise mass conservation. To map the RT0 element onto
quadrilateral and extruded hexahedrons, contravariant Piola mapping is used (see 
\citep{rognes2009efficient,Bercea2016} for further details).
The discrete formulations may be assembled into the following 
block format:
\begin{align}
\label{Eqn:AppendixB_system}
\begin{pmatrix}
\boldsymbol{K}_{vv} & \boldsymbol{K}_{vp}\\
\boldsymbol{K}_{pv} & \boldsymbol{K}_{pp}
\end{pmatrix}
\begin{pmatrix}
\boldsymbol{v} \\
\boldsymbol{p}
\end{pmatrix} =
\begin{pmatrix}
\boldsymbol{f}_{v} \\
\boldsymbol{f}_{p}
\end{pmatrix}
\end{align}
where the terms in equation \eqref{Eqn:AppendixB_bilinear} 
respectively correspond to $\boldsymbol{K}_{vv}$, $\boldsymbol{K}_{vp}$, 
and $\boldsymbol{K}_{pv}$, and equation \eqref{Eqn:AppendixB_linear} 
corresponds to $\boldsymbol{f}_{v}$. It should be noted that 
$\boldsymbol{K}_{pp}$ and $\boldsymbol{f}_{p}$ are a zero 
matrix and zero vector, respectively.

\subsection{Preconditioning methodology} Equation \eqref{Eqn:AppendixB_system}
is a saddle-point system which is tricky to precondition effectively 
for large-scale problems. Several classes of iterative solvers and preconditioning
strategies exist for these types of problems \citep{benzi2005numerical,
elman2006finite,murphy2000note}. One could alternatively employ 
hybridization techniques \citep{cockburn2009unified} which introduces
Lagrange multipliers to reduce the difficulty of solving such problems. 
In this paper, we employed a Schur complement approach 
to precondition the saddle-point system. Conceptually, the problem 
at hand is a 2$\times$2 block 
matrix:
\begin{align}
\boldsymbol{K} = \begin{pmatrix}
\boldsymbol{K}_{vv} & \boldsymbol{K}_{vp}\\
\boldsymbol{K}_{pv} & \boldsymbol{0}
\end{pmatrix}
\end{align}
which admits a full factorization of
\begin{align}
\boldsymbol{K} = \begin{pmatrix}
\boldsymbol{I} & \boldsymbol{0}\\
\boldsymbol{K}_{pv}\boldsymbol{K}_{vv}^{-1} & \boldsymbol{I}
\end{pmatrix}
\begin{pmatrix}
\boldsymbol{K}_{vv} & \boldsymbol{0}\\
\boldsymbol{0} & \boldsymbol{S}
\end{pmatrix}
\begin{pmatrix}
\boldsymbol{I} & \boldsymbol{K}_{vv}^{-1}\boldsymbol{K}_{vp}\\
\boldsymbol{0} & \boldsymbol{I}
\end{pmatrix}
\end{align}
where $\boldsymbol{I}$ is the identity matrix and 
\begin{align}
\boldsymbol{S}=-\boldsymbol{K}_{pv}\boldsymbol{K}_{vv}^{-1}\boldsymbol{K}_{vp}
\end{align}
is the Schur complement. The inverse can therefore be written as:
\begin{align}
\boldsymbol{K}^{-1} = \begin{pmatrix}
\boldsymbol{I} & -\boldsymbol{K}_{vv}^{-1}\boldsymbol{K}_{vp}\\
\boldsymbol{0} & \boldsymbol{I}
\end{pmatrix}
\begin{pmatrix}
\boldsymbol{K}_{vv}^{-1} & \boldsymbol{0}\\
\boldsymbol{0} & \boldsymbol{S}^{-1}
\end{pmatrix}
\begin{pmatrix}
\boldsymbol{I} & \boldsymbol{0}\\
-\boldsymbol{K}_{pv}\boldsymbol{K}_{vv}^{-1} & \boldsymbol{I}
\end{pmatrix}
\end{align}
The task at hand is to approximate $\boldsymbol{K}_{vv}^{-1}$ and
$\boldsymbol{S}^{-1}$. Since $\boldsymbol{K}_{vv}$ is a mass matrix
for the Darcy equation, we can invert it using the 
ILU(0) (incomplete lower upper) solver. We note that the 
Schur complement is spectrally a Laplacian, so we can employ
a diagonal mass-lumping of $\boldsymbol{K}_{vv}$ to give
a good approximation to $\boldsymbol{K}_{vv}^{-1}$. That is,
we can use
\begin{align}
\boldsymbol{S}_p = -\boldsymbol{K}_{pv}\mathrm{diag}\left(
\boldsymbol{K}_{vv}\right)^{-1}\boldsymbol{K}_{vp}
\end{align}
to precondition the inner solver inverting $\boldsymbol{S}$.
For this block we employ the multi-grid V-cycle on $\boldsymbol{S}_p$ 
using the Trilinos ML package (\citep{ml_users_guide}).
These blocks are symmetric and positive-definite so 
one could employ the CG solvers to obtain the inverses. 
When the inverses are obtained, only a single sweep of
flexible GMRES is needed to obtain the full solution.
However, instead of individually solving for 
$\boldsymbol{K}_{vv}^{-1}$ and $\boldsymbol{S}_p$, 
we could alternatively apply a single sweep of ILU(0) and multi-grid, respectively, 
and rely on GMRES to solve the entire block system. 
By providing less accurate approximations of the
inner individual blocks, the number of GMRES iterations for the
overall system increases but the numerical accuracy
remains the same. We have found that this methodology 
is computationally less expensive and more practical for 
large-scale computations. One could alternatively employ 
one of the factorizations (either lower or upper) to 
decrease the computational cost associated with setting 
up the preconditioner.
Below are some necessary PETSc command-line options for
the described Schur complement approach.
\begin{lstlisting}[caption=PETSc solver options for the Schur complement approach,label=Code:schur,frame=single]
parameters = {
  # Outer solver
  "ksp_type": "gmres",
  
  # Schur complement with full factorization
  "pc_type": "fieldsplit",
  "pc_fieldsplit_type": "schur",
  "pc_fieldsplit_schur_fact_type": "full",
    
  # Diagonal mass lumping
  "pc_fieldsplit_schur_precondition": "selfp",
  
  # Single sweep of ILU(0) for the mass matrix
  "fieldsplit_0_ksp_type": "preonly",
  "fieldsplit_0_pc_type": "ilu",
  
  # Single sweep of multi-grid for the Schur complement
  "fieldsplit_1_ksp_type": "preonly",
  "fieldsplit_1_pc_type": "ml"
}
\end{lstlisting}